%% file: reactive.tex
\documentclass[hidelinks]{eptcs}
\def\publicationstatus{
An extended abstract of this paper appears in the proceedings of CONCUR 2020.
doi:\href{http://dx.doi.org/10.4230/LIPIcs.CONCUR.2020.6}{10.4230/LIPIcs.CONCUR.2020.6}}
\usepackage{latexsym,underscore,amsmath,amssymb,color}
\usepackage[shortlabels]{enumitem}
\setlist[enumerate]{nosep}
\newfont{\bbb}{bbm10 scaled 1100}                       
\newfont{\bbbs}{bbm10 scaled 900}                       
\newcommand{\IN}{\mbox{\bbb N}}                         
\newcommand{\INs}{\mbox{\bbbs N}}                       
\newcommand{\IO}{\mbox{\bbb O}}                         
\newcommand{\IP}{\mbox{\bbb P}}                         
\newcommand{\IQ}{\mbox{\bbb Q}}                         
\newcommand{\IT}{\mbox{\bbb E}}                         
\DeclareSymbolFont{frenchscript}{OMS}{ztmcm}{m}{n}
\DeclareMathSymbol{\Pow}{\mathord}{frenchscript}{80}    
\DeclareMathSymbol{\C}{\mathord}{frenchscript}{67}      
\DeclareMathSymbol{\B}{\mathrel}{frenchscript}{66}      
\DeclareMathSymbol{\R}{\mathord}{frenchscript}{82}      
\DeclareMathSymbol{\BS}{\mathord}{frenchscript}{83}     
\newcommand{\I}{\mathcal{I}}                            
\newcommand{\Rn}{\mathcal{R}}                           
\newcommand{\BR}{\mathrel{\color{blue}\B}}              
\newcommand{\BRB}{\mathrel{\color{blue}\B\,\bis{\,}}}   
\renewcommand{\phi}{\varphi}
\renewcommand{\epsilon}{\varepsilon}
\newcommand{\E}{\mathcal{E}}                            
\makeatletter
\newif\if@qeded
\global\@qededtrue
\def\qed{\hfill$\Box$\global\@qededtrue}
\def\qedneeded{\global\@qededfalse}
\def\qedifneeded{\if@qeded\else\qed\fi}
\makeatother
\newtheorem{defi}{Definition}
\newtheorem{theo}[defi]{Theorem}
\newtheorem{prop}[defi]{Proposition}
\newtheorem{lemm}{Lemma}
\newtheorem{coro}[defi]{Corollary}
\newtheorem{exam}[defi]{Example}

\newenvironment{definition}[1]{\begin{defi} \rm \label{df:#1} }{\end{defi}}
\newenvironment{definitionq}[2]{\begin{defi}[#2] \rm \label{df:#1} }{\end{defi}}
\newenvironment{theorem}[1]{\begin{theo} \rm \label{thm:#1} }{\end{theo}}
\newenvironment{proposition}[1]{\begin{prop} \rm \label{pr:#1} }{\end{prop}}
\newenvironment{propositionq}[2]{\begin{prop}[#2] \rm \label{pr:#1} }{\end{prop}}
\newenvironment{lemma}[1]{\begin{lemm} \rm \label{lem:#1} }{\end{lemm}}
\newenvironment{corollary}[1]{\begin{coro} \rm \label{cor:#1} }{\end{coro}}
\newenvironment{example}[1]{\begin{exam} \rm \label{ex:#1} }{\end{exam}}
\newenvironment{proof}{\qedneeded\begin{trivlist} \item[\hspace{\labelsep}\bf Proof:]}{\qedifneeded\end{trivlist}}

\newcommand{\df}[1]{Definition~\ref{df:#1}}
\newcommand{\thm}[1]{Theorem~\ref{thm:#1}}
\newcommand{\pr}[1]{Proposition~\ref{pr:#1}}
\newcommand{\lem}[1]{Lemma~\ref{lem:#1}}
\newcommand{\cor}[1]{Corollary~\ref{cor:#1}}

\newcommand{\tab}[1]{Table~\ref{tab:#1}}
\newcommand{\fig}[1]{Figure~\ref{fig:#1}}
\newcommand{\Sec}[1]{Section~\ref{sec:#1}}
\makeatletter
\def\comesfrom{\@transition\leftarrowfill}
\def\goesto{\@transition\rightarrowfill}
\def\ngoesto{\@transition\nrightarrowfill}
\def\Goesto{\@transition\Rightarrowfill}
\def\nGoesto{\@transition\nRightarrowfill}
\def\xmapsto{\@transition\mapstofill}
\def\nxmapsto{\@transition\nmapstofill}
\def\@transition#1{\@@transition{#1}}
\newbox\@transbox
\newbox\@arrowbox
\def\@@transition#1#2%
   {\setbox\@transbox\hbox
      {\vrule height 1.5ex depth .8ex width 0ex\hskip0.25em$\scriptstyle#2$\hskip0.25em}
   \ifdim\wd\@transbox<1.5em
      \setbox\@transbox\hbox to 1.5em{\hfil\box\@transbox\hfil}\fi
   \setbox\@arrowbox\hbox to \wd\@transbox{#1}
   \ht\@arrowbox\z@\dp\@arrowbox\z@
   \setbox\@transbox\hbox{$\mathop{\box\@arrowbox}\limits^{\box\@transbox}$}
   \dp\@transbox\z@\ht\@transbox 10pt
   \mathrel{\box\@transbox}}
\def\nrightarrowfill{$\m@th\mathord-\mkern-6mu%
  \cleaders\hbox{$\mkern-2mu\mathord-\mkern-2mu$}\hfill
  \mkern-6mu\mathord\not\mkern-2mu\mathord\rightarrow$}
\def\Rightarrowfill{$\m@th\mathord=\mkern-6mu%
  \cleaders\hbox{$\mkern-2mu\mathord=\mkern-2mu$}\hfill
  \mkern-6mu\mathord\Rightarrow$}
\def\nRightarrowfill{$\m@th\mathord=\mkern-6mu%
  \cleaders\hbox{$\mkern-2mu\mathord=\mkern-2mu$}\hfill
  \mkern-6mu\mathord\not\mathord\Rightarrow$}
\def\mapstofill{$\m@th\mathord\mapstochar\mathord-\mkern-6mu%
  \cleaders\hbox{$\mkern-2mu\mathord-\mkern-2mu$}\hfill
  \mkern-6mu\mathord\rightarrow$}
\def\nmapstofill{$\m@th\mathord\mapstochar\mathord-\mkern-6mu%
  \cleaders\hbox{$\mkern-2mu\mathord-\mkern-2mu$}\hfill
  \mkern-6mu\mathord\not\mkern-2mu\mathord\rightarrow$}
\makeatother 
\newcommand{\plat}[1]{\raisebox{0pt}[0pt][0pt]{#1}}     
\newcommand{\rec}[1]{\plat{$                            
        \stackrel{\mbox{\tiny $/$}}
        {\raisebox{-.3ex}[.3ex]{\tiny $\backslash$}}
        \!\!#1\!\!
        \stackrel{\mbox{\tiny $\backslash$}}
        {\raisebox{-.3ex}[.3ex]{\tiny $/$}}$}}
\newcommand{\rbis}[2]{\mathrel{\,                       
        \raisebox{.3ex}{$\underline{\makebox[.7em]{$\leftrightarrow$}}$}
                  \,^{#1}_{#2}\!}}
\newcommand{\bis}[1]{\rbis{}{#1}}                       
\newcommand{\rt}{{\rm t}}                         
\newcommand{\Var}{{\it Var}}                      
\newcommand{\spar}[1]{\mathbin{\|^{}_{#1}}}        
\newcommand{\RS}{{\cal S}}                        
\newcommand{\CCSP}{\mbox{CCSP}_\rt^\theta}         
\newcommand{\diam}[1]{\langle#1\rangle}           
\newcommand{\rep}[1]{\textbf{[}#1\textbf{]}}      

\def\titlerunning{Reactive Bisimulation Semantics for a Process Algebra with Time-Outs}
\title{Reactive Bisimulation Semantics\\ for a Process Algebra with Time-Outs}
\author{Rob van Glabbeek
\institute{Data61, CSIRO, Sydney, Australia}
\institute{School of Computer Science and Engineering,
University of New South Wales, Sydney, Australia}
\email{rvg@cs.stanford.edu}
}
\def\authorrunning{R.J. van Glabbeek}
\begin{document}
\maketitle

\begin{abstract}
This paper introduces the counterpart of strong bisimilarity for labelled transition systems extended
with time-out transitions. It supports this concept through a modal characterisation, congruence
results for a standard process algebra with recursion, and a complete axiomatisation.
\end{abstract}

\section{Introduction}

This is a contribution to classic untimed non-probabilistic process algebra, modelling
systems that move from state to state by performing discrete, uninterpreted actions.
A system is modelled as a process-algebraic expression, whose standard semantics is a state in a labelled
transition system (LTS).  An LTS consists of a set of states, with action-labelled transitions between them.
The execution of an action is assumed to be instantaneous, so when any time elapses the system must
be in one of its states. With ``untimed'' I mean that I will refrain from quantifying the passage of
time; however, whether a system can pause in some state or not will be part of my model.

Following \cite{Mi90ccs}, I consider \emph{reactive} systems that interact with their environments
through the synchronous execution of visible actions $a$, $b$, $c$, ... taken from an alphabet $A$.
At any time, the environment \emph{allows} a set of actions $X\subseteq A$, while \emph{blocking}
all other actions. At discrete moments the environment can change the set of actions it allows.
In a metaphor from \cite{Mi90ccs}, the environment of a system can be seen as a
user interacting with it. This user has a button for each action $a\in A$, on which it can exercise
pressure. When the user exercises pressure \emph{and} the system is in a state where it can perform
action $a$, the action occurs. For the system this involves taking an $a$-labelled transition to a
following state; for the environment it entails the button going down, thus making the
action occurrence observable. This can trigger the user to alter the set of buttons on which it
exercises pressure.

The current paper considers two special actions that can occur as transition labels: the traditional
\emph{hidden action} $\tau$ \cite{Mi90ccs}, modelling the occurrence of an instantaneous action from
which we abstract, and the \emph{time-out} action $\rt$, modelling the end of a time-consuming 
activity from which we abstract. The latter was introduced in \cite{vG21} and constitutes the main novelty
of the present paper with respect to \cite{Mi90ccs} and forty years of research in process algebra.
Both special actions are assumed to be unobservable, in the sense that their occurrence cannot trigger
any state-change in the environment. Conversely, the environment cannot cause or block the occurrence
of these actions.

Following \cite{vG21}, I model the passage of time in the following way.
When a system arrives in a state $P$, and at that time $X$ is the set of actions allowed by the
environment, there are two possibilities.
If $P$ has an outgoing transition $P \goesto\alpha Q$ with $\alpha \in X\cup \{\tau\}$,
the system immediately takes one of the outgoing transitions $P \goesto\alpha Q$ with $\alpha \in X\cup \{\tau\}$,
without spending any time in state $P$. The choice between these actions is entirely nondeterministic.
The system cannot immediately take a transition $\goesto b$ with $b\in A{\setminus}X$, because the
action $b$ is blocked by the environment. Neither can it immediately take a transition $P\goesto\rt Q$,
because such transitions model the end of an activity with a finite but positive duration that
started when reaching state $P$.

In case $P$ has no outgoing transition $P \goesto\alpha Q$ with $\alpha \in X\cup \{\tau\}$,
the system idles in state $P$ for a positive amount of time. This idling can end in two possible ways.
Either one of the time-out transitions  $P \goesto\rt Q$ occurs, or the environment spontaneously
changes the set of actions it allows into a different set $Y$ with the property that
$P \goesto a Q$ for some $a \in Y$. In the latter case a transition $P \goesto a Q$ occurs, with $a \in Y$.
The choice between the various ways to end a period of idling is entirely nondeterministic.
It is possible to stay forever in state $P$ only if there are no outgoing time-out transitions $P \goesto\rt Q$.

The addition of time-outs enhances the expressive power of LTSs and process algebras.
The process $a.P + \rt .b.Q$, for instance, models a choice between $a.P$ and $b.Q$ where the former
has priority. In an environment where $a$ is allowed it will always choose $a.P$ and never $b.Q$;
but in an environment that blocks $a$ the process will, after some delay, proceed with $b.Q$.
Such a priority mechanism cannot be modelled in standard process algebras without time-outs,
such as CCS \cite{Mi90ccs}, CSP \cite{BHR84,Ho85} and ACP \cite{BW90,Fok00}.
Additionally, mutual exclusion cannot be correctly modelled in any of these standard process
algebras \cite{GH15b}, but adding time-outs makes it possible---see \Sec{conclusion}
for a more precise statement.

In \cite{vG21} I characterised the coarsest reasonable semantic equivalence on LTSs with
time-outs---the one induced by \emph{may testing}, as proposed by De Nicola \& Hennessy \cite{DH84}.
In the absence of time-outs, may testing yields \emph{weak trace equivalence}, where two processes are
defined equivalent iff they have the same \emph{weak traces}: sequence of actions the system
can perform, while eliding hidden actions. In the presence of time-outs weak trace equivalence fails
to be a congruence for common process algebraic operators, and may testing yields its congruence
closure, characterised in \cite{vG21} as \emph{(rooted) failure trace equivalence}.

The present paper aims to characterise one of the finest reasonable semantic equivalences on LTSs
with time-outs---the counterpart of strong bisimilarity for LTSs without time-outs.
Naturally, strong bisimilarity can be applied verbatim to LTSs with time-outs---and has been in
\cite{vG21}---by treating $\rt$ exactly like any visible action. Here, however, I aim to take into
account the essence of time-outs, and propose an equivalence that satisfies some natural laws
discussed in \cite{vG21}, such as $\tau.P + \rt.Q = \tau.P$ and
\mbox{$a.P + \rt.(Q + \tau.R + a.S) = a.P + \rt.(Q + \tau.R)$}.
To motivate the last law, note that the time-out transition
$a.P + \rt.(Q + \tau.R + a.S)  \goesto\rt Q + \tau.R + a.S$
can occur only in an environment that blocks the action $a$, for otherwise $a$ would have taken
place before the time-out went off. The occurrence of this transition is not observable by the
environment, so right afterwards the state of the environment is unchanged, and the action $a$ is
still blocked. Therefore, the process $Q + \tau.R + a.S$ will, without further ado, proceed with the
$\tau$-transition to $R$, or any action from $Q$, just as if the $a.S$ summand were not present.

Standard process algebras and LTSs without time-outs can model systems whose behaviour is
triggered by input signals from the environment in which they operate. This is why they are 
called ``reactive systems''. By means of time-outs one can additionally model systems whose
behaviour is triggered by the \emph{absence} of input signals from the environment, during a
sufficiently long period. This creates a greater symmetry between a system and its environment, as
it has always been understood that the environment or user of a system can change its behaviour as
a result of sustained inactivity of the system it is interacting with.
Hence one could say that process algebras and LTSs enriched with time-outs form a more faithful
model of reactivity. It is for this reason that I use the name \emph{reactive bisimilarity}
for the appropriate form of bisimilarity on systems modelled in this fashion.

\Sec{reactive} introduces strong reactive bisimilarity as the proper counterpart of strong
bisimilarity in the presence of time-out transitions. Naturally, it coincides with strong
bisimilarity when there are no time-out transitions.
\Sec{modal} derives a modal characterisation; a reactive variant of the Hennessy-Milner logic.
\Sec{timeout bisimulations} offers an alternative characterisation of strong reactive bisimilarity
that will be more convenient in proofs, although it is lacks the intuitive appeal to be used as the
initial definition. Appendix~\ref{reduction}, reporting on work by Max Pohlmann~\cite{Po21}, offers
yet another characterisation of strong reactive bisimilarity; one that reduces it to strong
bisimilarity in a context that models a system together with its environment.

\Sec{ccsp} recalls the process algebra CCSP, a common mix of CCS and CSP, and adds the time-out
action, as well as two auxiliary operators that will be used in the forthcoming axiomatisation.
\Sec{guarded} states that in this process algebra one can express all countably branching
transition systems, and only those, or all and only the finitely branching ones when restricting to
guarded recursion.

\Sec{congruence} recalls the concept of a congruence, focusing on the congruence property for the
recursion operator, which is commonly the hardest to establish. It then shows that the simple
\emph{initials equivalence}, as well as Milner's strong bisimilarity, are congruences.
Due to the presence of negative premises in the operational rules for the auxiliary operators, these
proofs are not entirely trivial. Using these results as a stepping stone, \Sec{full congruence}
shows that strong reactive bisimilarity is a congruence for my extension of CCSP\@.
Here the congruence property for one of the auxiliary operators with negative premises is needed in
establishing the result for the common CCSP operators, such as parallel composition.

\Sec{RSP} shows that guarded recursive specifications have unique solutions up to strong reactive bisimilarity.
Using this, \Sec{axioms} provides a sound and complete axiomatisation for processes with guarded recursion.
My completeness proof combines three innovations in establishing completeness of process algebraic axiomatisations.
First of all, following \cite{GM20}, it applies to \emph{all} processes in a Turing powerful
language like guarded CCSP, rather than the more common fragment merely employing finite sets of
recursion equations featuring only choice and action prefixing.
Secondly, instead of the classic technique of \emph{merging guarded recursive equations}
\cite{Mi84,Mi89a,Wa90,vG93a,LDH05}, which in essence proves two bisimilar systems $P$ and $Q$
equivalent by equating both to an intermediate variant that is essentially a \emph{product} of $P$ and $Q$,
I employ the novel method of \emph{canonical solutions} \cite{GF20,LY20}, which equates both $P$ and $Q$
to a canonical representative within the bisimulation equivalence class of $P$ and $Q$---one that has
only one reachable state for each bisimulation equivalence class of states of $P$ and $Q$.
In fact I tried so hard, and in vain, to apply the traditional technique of merging guarded
recursive equations, that I came to believe that it fundamentally does not work for this axiomatisation.
The third innovation is the use of the axiom of choice \cite{Zermelo08} in defining the transition
relation on my canonical representative, in order to keep this process finitely branching.

\Sec{conclusion} describes a worthwhile gain in expressiveness caused by the addition of
time-outs, and presents an agenda for future work.

\section{Reactive bisimilarity}\label{sec:reactive}

A \emph{labelled transition system} (LTS) is a triple $(\IP,Act,\rightarrow)$ with $\IP$ a set (of
\emph{states} or \emph{processes}), $Act$ a set (of \emph{actions}) and ${\rightarrow}\in\IP\times Act\times\IP$.
In this paper I consider LTSs with $Act:= A\uplus\{\tau,\rt\}$, where $A$ is a set of \emph{visible actions},
$\tau$ is the \emph{hidden action}, and $\rt$ the \emph{time-out} action.
The set of \emph{initial} actions of a process $P \in\IP$ is $\I(P):=\{\alpha\in A\cup\{\tau\} \mid P{\goesto \alpha}\}$.
Here $P{\goesto \alpha}$ means that there is a $Q$ with $P \goesto\alpha Q$.

\begin{definition}{reactive bisimilarity}
A \emph{strong reactive bisimulation} is a symmetric relation
${\R} \subseteq (\IP \times \Pow(A) \times \IP) \cup (\IP \times \IP)$
(meaning that $(P,X,Q)\in \R \Leftrightarrow (Q,X,P)\in \R$ and  $(P,Q)\in \R \Leftrightarrow (Q,P)\in \R$),
such that,
\begin{itemize}\itemsep 0pt \parsep 0pt
\item if $(P,Q)\in\R$ and $P \goesto{\tau} P'$, then there exists a $Q'$ such that $Q\goesto{\tau} Q'$ and $(P',Q')\in\R$,
\item if $(P,Q)\in\R$ then $(P,X,Q)\in\R$ for all $X\subseteq A$,
\end{itemize}
and for all $(P,X,Q)\in\R$,
\begin{itemize}\itemsep 0pt \parsep 0pt
\item if $P \goesto{a} P'$ with $a\mathbin\in X$, then there exists a $Q'$ such that $Q\goesto{a} Q'$ and
  $(P',Q')\in\R$,
\item if $P \goesto{\tau} P'$, then there exists a $Q'$ such that $Q\goesto{\tau} Q'$ and $(P',X,Q')\in\R$,
\item if $\I(P)\cap (X\cup\{\tau\})=\emptyset$, then $(P,Q)\in\R$, and
\item if $\I(P)\cap (X\cup\{\tau\})=\emptyset$ and $P \goesto{\rt} P'$, then
  $\exists Q'$ such that $Q\goesto{\rt} Q'$ and $(P',X,Q')\in\R$.
\end{itemize}
Processes $P,Q\mathbin\in\IP$ are \emph{strongly $X$-bisimilar}, denoted $P \rbis{X}{r} Q$, if
$(P,X,Q)\in\R$ for some strong reactive bisimulation $\R$.
They are \emph{strongly reactive bisimilar}, denoted $P \bis{r} Q$, if $(P, Q)\in\R$ for
some strong reactive bisimulation $\R$.
\end{definition}
Intuitively, $(P,X,Q)\in\R$ says that processes $P$ and $Q$ behave the same way, as witnessed by the
relation $\R$, when placed in the environment $X$---meaning any environment that allows exactly the
actions in $X$ to occur---whereas $(P,Q)\in\R$ says they behave the same way in an environment that
has just been triggered to change.
An environment can be thought of as an unknown process placed in parallel with $P$ and $Q$, using
the operator $\spar{A}$, enforcing synchronisation on all visible actions.
The environment $X$ can be seen as a process $\sum_{i\in I}a_i.R_i + \rt.R$ where $\{a_i \mid i\in I\} = X$.
A triggered environment, on the other hand, can execute a sequence of instantaneous hidden actions before
stabilising as an environment $Y$, for $Y\subseteq A$. During this execution, actions can be blocked
and allowed in rapid succession. Since the environment is unknown, the bisimulation should be robust
under any such environment.

The first clause for $(P,X,Q)\in\R$ is like the common transfer property of strong bisimilarity \cite{Mi90ccs}:
a visible $a$-transition of $P$ can be matched by one of $Q$, such that the resulting processes $P'$
and $Q'$ are related again.
However, I require it only for actions $a\in X$, because actions $b\in A{\setminus}X$ cannot happen
at all in the environment $X$, and thus need not be matched by $Q$.
Since the occurrence of $a$ is observable by the environment, this can trigger the environment to
change the set of actions it allows, so $P'$ and $Q'$ ought to be related in a triggered environment.

The second clause is the transfer property for $\tau$-transitions. Since these are not observable by
the environment, they cannot trigger a change in the set of actions allowed by it, so the 
resulting processes $P'$ and $Q'$ should be related only in the same environment $X$.

The first clause for $(P,Q)\in\R$ expresses the transfer property for $\tau$-transitions in a
triggered environment. Here it may happen that the $\tau$-transition occurs before the environment stabilises,
and hence $P'$ and $Q'$ will still be related in a triggered environment.
A similar transfer property for $a$-transitions is already implied by the next two clauses.

The second clause allows a triggered environment to stabilise into any environment $X$.

The first two clauses for $(P,X,Q)\in \R$ imply that if $(P,X,Q)\in \R$ then
$\I(P)\cap (X\cup\{\tau\}) = \I(Q) \cap (X\cup\{\tau\})$. So $P \bis{r} Q$ implies $\I(P)=\I(Q)$. The condition
$\I(P)\cap (X\cup\{\tau\})=\emptyset$ is necessary and sufficient for the system to remain a positive
amount of time in state $P$ when $X$ is the set of allowed actions.  The next clause says that
during this time the environment may be triggered to change the set of actions it allows by an event
outside our model, that is, by a time-out in the environment. So $P$ and $Q$ should be related in a triggered environment.

The last clause says that also a $\rt$-transition of $P$ should be matched by one of $Q$.
Naturally, the $\rt$-transition of $P$ can be taken only when the system is idling in $P$, i.e.,
when $\I(P)\cap (X\cup\{\tau\})=\emptyset$. The resulting processes $P'$
and $Q'$ should be related again, but only in the same environment allowing~$X$.

\begin{proposition}{equivalence}
Strong $X$-bisimilarity and strong reactive bisimilarity are equivalence relations.
\end{proposition}
\begin{proof}
$\rbis{X}{r}$, $\bis{r}$ are reflexive, as $\{(P,X,P), (P,P)\mid P\mathbin\in\IP \wedge X\mathbin\subseteq A\}$ is
  a strong reactive bisimulation.
\\
$\rbis{X}{r}$ and $\bis{r}$ are symmetric, since strong reactive bisimulations are symmetric by definition.
\\
$\rbis{X}{r}$ and $\bis{r}$ are transitive, for if $\R$ and $\BS$ are strong reactive bisimulations, then so is\\[1ex]
\small\mbox{}\hfill$\R;\BS = \{(P,X,R)\mid \exists Q.~(P,X,Q)\in \R \wedge (Q,X,R)\in\BS\}
 \cup \{(P,R)\mid \exists Q.~ (P,Q)\in \R \wedge (Q,R)\in\BS\}$.\hfill
\end{proof}
Note that the union of arbitrarily many strong reactive bisimulations is itself a strong reactive bisimulation.
Therefore the family of relations $\bis{r}$\,, $\rbis{X}{r}$ for $X\subseteq A$ can be seen as
a strong reactive bisimulation. 

To get a firm grasp on strong reactive bisimilarity, the reader is invited to check the two laws
mentioned in the introduction, and then to construct a strong reactive bisimulation between the two
systems depicted in \fig{L3}. Here $P$, $Q$, $R$ and $S$ are arbitrary subprocesses.
\begin{figure}
\input{L3}
\centerline{\box\graph}\vspace{1ex}
\caption{Two strongly reactive bisimilar processes}
\label{fig:L3}
\end{figure}
The four processes that are targets of $\rt$-transitions always run in an environment that
blocks $b$. In an environment that allows $a$, the branch $b.R$ disappears, so that the left branch of
the first process can be matched with the left branch of the second process, and similarly for the two
right branches. In an environment that blocks $a$, this matching won't fly, as the branch $b.R$ now
survives. However, the branches $a.Q$ will disappear, so that the left branch of the first process
can be matched with the right branch of the second, and vice versa.
\begin{figure}[h]
\input{pairs}
\centerline{\box\graph}\vspace{1ex}
\caption{Reactive bisimilarity is not fully determined by reactive $X$-bisimilarity}
\label{fig:pairs}
\end{figure}

\noindent
The processes $U$ and $V$ of \fig{pairs} show that the pairs that occur in a strong reactive
bisimulation are not completely determined by the triples. One has $U \rbis{X}{r} V$ for any $X\subseteq A$,
yet $U \not\hspace{-.7pt}\bis{r} V$. In particular, when $a \in X$ the branch $\rt.R$ is redundant,
and when $a \notin X$ the branch $a.Q$ is redundant.

Appendix~\ref{reduction}, reporting on work by Max Pohlmann \cite{Po21}, offers a context $\C$ with the property
that $P \bis{r} Q$ iff $\C(P) \bis{} \C(Q)$, thereby reducing strong reactive bisimilarity to 
strong bisimilarity. The context $\C$ places a system in a most general environment in which it
could be running. This result allows any toolset for checking strong bisimilarity to be applicable
for checking strong reactive bisimilarity.

\subsection{A more general form of reactive bisimulation}

The following notion of a \emph{generalised strong reactive bisimulation} (gsrb) generalises that of
a strong reactive bisimulation; yet it induces the same concept of strong reactive bisimilarity.
This makes the relation convenient to use for further analysis. I did not introduce it as the
original definition, because it lacks a strong motivation.

\begin{definition}{reactive bisimulation}
A \emph{gsrb} is a symmetric relation
${\R} \mathbin\subseteq (\IP \times \Pow(A) \times \IP) \cup (\IP \times \IP)$
such that, for all $(P,Q)\mathbin\in\R$,
\begin{itemize}\itemsep 0pt \parsep 0pt
\item if $P \goesto{\alpha} P'$ with $\alpha\mathbin\in A\cup\{\tau\}$, then there exists a $Q'$ such
  that $Q\goesto{\alpha} Q'$ and $(P',Q')\in\R$,
\item if $\I(P)\cap (X\cup\{\tau\})\mathbin=\emptyset$ with $X\subseteq A$ and $P \goesto{\rt} P'$, then
  $\exists Q'$ with $Q\goesto{\rt} Q'$ and $(P',X,Q')\in\R$,
\end{itemize}
and for all $(P,Y,Q)\in\R$,
\begin{itemize}\itemsep 0pt \parsep 0pt
\item if $P \goesto{a} P'$ with either $a\mathbin\in Y$ or $\I(P)\cap (Y\cup\{\tau\})\mathbin=\emptyset$,
  then $\exists Q'$ with $Q\goesto{a} Q'$ and $(P',Q')\in\R$,
\item if $P \goesto{\tau} P'$, then there exists a $Q'$ such that $Q\goesto{\tau} Q'$ and $(P',Y,Q')\in\R$,
\item if $\I(P)\cap (X\cup Y \cup\{\tau\})\mathbin=\emptyset$ with $X\mathbin\subseteq A$ and $P \goesto{\rt} P'$
  then $\exists Q'$ with $Q\goesto{\rt} Q'$ and $(P',X,Q')\mathbin\in\R$.
\end{itemize}
\end{definition}
Unlike \df{reactive bisimilarity}, a gsrb needs the triples $(P,X,Q)$ only after
encountering a $\rt$-transition; two systems without $\rt$-transitions can be related without using
these triples at all.

\begin{proposition}{reactive bisimulation}
$P \bis{r} Q$ iff there exists a gsrb $\R$ with  $(P, Q)\in\R$.\\
Likewise, $P \rbis{X}{r} Q$ iff there exists a gsrb $\R$ with $(P,X,Q)\in\R$.
\end{proposition}
\begin{proof}
Clearly, each strong reactive bisimulation satisfies the five clauses of \df{reactive bisimulation}
and thus is a gsrb.
In the other direction, given a gsrb $\B$, let\vspace{-2pt}
\[\R := {\B} \begin{array}[t]{l}\cup~ \{(P,X,Q) \mid (P,Q)\in{\B} \wedge X \subseteq A\}\\
  \cup~ \{(P,Q),(P,X,Q) \mid \exists Y.~ (P,Y,Q)\in{\B} \wedge \I(P)\cap (Y\cup\{\tau\})\mathbin=\emptyset
  \wedge X \subseteq A\}\,.
\vspace{-2pt}
\end{array}\]
It is straightforward to check that $\R$ satisfies the six clauses of \df{reactive bisimilarity}.
\end{proof}
The above proof has been formalised in \cite{Po21}, using the interactive proof assistant Isabelle. 
The formalisation takes up around 250 lines of code.

\section{A modal characterisation of strong reactive bisimilarity}\label{sec:modal}

The Hennessy-Milner logic \cite{HM85} expresses properties of the behaviour of processes in an LTS\@.

\begin{definition}{formulas}
The class $\IO$ of \emph{infinitary HML formulas} is defined as follows, where $I$ ranges over all index sets and $\alpha$ over $A\cup\{\tau\}$:\vspace{-5pt}
\[
  \phi ~~::=~~ {\displaystyle\bigwedge_{i\in I}}\,\phi_i ~~|~~ \neg\phi ~~|~~ \diam{\alpha}\phi\vspace{-2pt}
\]
$\top$ abbreviates the empty conjunction, and $\phi_1\land\phi_2$ stands for $\bigwedge_{i\in\{1,2\}}\phi_i$.
\end{definition}

\noindent
$P\models\phi$ denotes that process $P$ satisfies formula $\phi$. The
first two operators represent the standard Boolean operators
conjunction and negation. By definition, $P\models\diam{\alpha}\phi$
iff \plat{$P\goesto{\alpha}P'$} for some $P'$ with $P'\models\phi$.

A famous result stemming from \cite{HM85} states that\vspace{-3pt}
\[ P \bis{} Q ~~\Leftrightarrow~~ \forall \phi\in\IO.~(P \models \phi \Leftrightarrow Q \models \phi)\vspace{-3pt} \]
where $\bis{\,}$ denotes strong bisimilarity \cite{Mi90ccs,HM85}, formally defined in \Sec{strong}.
It states that the Hennessy-Milner logic yields a \emph{modal characterisation} of strong bisimilarity.
I will now adapt this result to obtain a modal characterisation of strong reactive bisimilarity.

To this end I extend the Hennessy-Milner logic with a new modality $\diam{X}$, for $X\subseteq A$,
and auxiliary satisfaction relations  ${\models}_X \subseteq \IP \times \IO$ for each $X \subseteq A$.
The formula $P \models \diam{X}\phi$ says that in an environment $X$, allowing exactly the actions in $X$,
process $P$ can perform a time-out transition to a process that satisfies $\phi$.
$P \models_X \phi$ says that $P$ satisfies $\phi$ when placed in environment $X$.
The relations $\models$ and $\models_X$ are the smallest ones satisfying:
\[\begin{array}{l@{\qquad\mbox{if}\qquad}l}
  P \models \bigwedge_{i\in I} \phi_i & \forall i\in I.~ P \models \phi_i \\
  P \models \neg \phi & P \not\models \phi \\
  P \models \diam{\alpha} \phi  \mbox{~~~with~} \alpha\in A\cup\{\tau\} & \exists P'.~ P \goesto{\alpha} P' \wedge P' \models \phi \\
  P \models \diam{X}\phi \mbox{~~~with~} X\subseteq A & \I(P)\cap (X\cup\{\tau\}) = \emptyset 
                  \wedge \exists P'.~ P\goesto{\rt}P' \wedge P' \models_X \phi \\[2ex]

  P \models_X \bigwedge_{i\in I} \phi_i & \forall i\in I.~ P \models_X \phi_i \\
  P \models_X \neg \phi & P \not\models_X \phi \\
  P \models_X \diam{a} \phi  \mbox{~~~with~} a\in A & a \in X \wedge \exists P'.~ P \goesto{a} P' \wedge P' \models \phi \\
  P \models_X \diam{\tau} \phi & \exists P'.~ P \goesto\tau P' \wedge P' \models_X \phi \\
  P \models_X \phi & \I(P)\cap (X\cup\{\tau\}) = \emptyset \wedge P \models \phi \\
\end{array}\]
Note that a formula $\diam{a}\phi$ is less often true under $\models_X$ than under $\models$, due to
the side condition $a \in X$. This reflects the fact that $a$ cannot happen in an environment that
blocks it.
The last clause in the above definition reflects the fifth clause of \df{reactive bisimilarity}.
If $\I(P)\cap (X\cup\{\tau\}) = \emptyset$, then process $P$, operating in environment $X$, idles for a
while, during which the environment can change. This ends the blocking of actions $a \notin X$
and makes any formula valid under $\models$ also valid under $\models_X$.

\begin{example}{reactive HML}
Both systems from \fig{L3} satisfy
$\diam{\emptyset} \diam\tau \diam{b}\top \wedge \diam{\emptyset} \diam\tau \neg \diam{b}\top \wedge
 \diam{\{a\}} \diam{a}\top \wedge \diam{\{a\}} \neg \diam{a}\top$
and neither satisfies $\diam{\emptyset} \textcolor{red}(\diam{a}\top \wedge \diam{\tau}\diam{b}\top\textcolor{red})$ or 
$\diam{\{a\}} \textcolor{red}(\diam{a}\top \wedge \diam{\tau}\diam{b}\top\textcolor{red})$.
\end{example}

\begin{theorem}{modal char}
Let $P,Q\in\IP$ and $X\subseteq A$. Then
$P \bis{r} Q ~~\Leftrightarrow~~ \forall \phi\in\IO.~(P \models \phi \Leftrightarrow Q \models \phi)$ \\
and\hspace{180pt}
$P \rbis{X}{r} Q ~~\Leftrightarrow~~ \forall \phi\in\IO.~(P \models_X \phi \Leftrightarrow Q \models_X \phi)$.
\end{theorem}
\begin{proof}
``$\Rightarrow$'': I prove by simultaneous structural induction on $\phi \in \IO$ that,
for all $P,Q\in\IP$ and $X\subseteq A$,
$P \bis{r} Q \wedge P \models \phi ~\Rightarrow~ Q \models \phi$ and
$P \rbis{X}{r} Q \wedge P \models_X \phi ~\Rightarrow~ Q \models_X \phi$.
For each $\phi$, the converse implications ($Q \models \phi \Rightarrow P \models \phi$ and
$Q \models_X \phi \Rightarrow P \models_X \phi$) follow by symmetry.
In particular, these converse directions may be used when invoking the induction hypothesis.
\begin{itemize}
\item Let $P \bis{r} Q \wedge P \models \phi$.
\begin{itemize}
\item Let $\phi = \bigwedge_{i\in I}\phi_i$. Then $P \models \phi_i$ for all $i\mathbin\in I$. By induction 
  $Q \models \phi_i$ for all $i$, so $Q \models \bigwedge_{i\in I}\phi_i$.
\item Let $\phi = \neg \psi$. Then $P \not\models \psi$. By induction $Q \not\models \psi$, so $Q \models \neg \psi$.
\item Let $\phi = \diam{\alpha}\psi$ with $\alpha\mathbin\in A\cup\{\tau\}$.
  Then $P \goesto{\alpha} P'$ for some $P'$ with $P' \models \psi$.
  By \df{reactive bisimulation}, $Q \goesto{\alpha} Q'$ for some $Q'$ with $P' \bis{r} Q'$.
  So by induction $Q' \models \psi$, and thus $Q \models \diam{\alpha} \psi$.
\item Let $\phi = \diam{X}\psi$ for some $X\subseteq A$. Then $\I(P)\cap (X\cup\{\tau\}) = \emptyset$
  and $P \goesto{\rt} P'$ for some $P'$ with $P' \models_X \psi$. By \df{reactive bisimulation},
  $Q \goesto{\rt} Q'$ for some $Q'$ with $P' \rbis{X}{r} Q'$.
  So by induction $Q' \models_X \psi$. Moreover, $\I(Q)\mathbin=\I(P)$, as $P \bis{r} Q$, so $\I(Q)\cap (X\cup\{\tau\}) \mathbin= \emptyset$.
  Thus $Q \models \diam{X}\psi$.
\end{itemize}
\item Let $P \rbis{X}{r} Q \wedge P \models_X \phi$.
\begin{itemize}
\item Let $\phi \mathbin= \bigwedge_{i\in I}\phi_i$, and $P \models_X \phi_i$ for all $i\in I$. By induction 
  $Q \models_X \phi_i$ for all $i\in I$, so $q \models_X \bigwedge_{i\in I}\phi_i$.
\item Let $\phi = \neg \psi$, and $P \not\models_X \psi$. By induction $Q \not\models_X \psi$, so $Q \models_X \neg \psi$.
\item Let $\phi = \diam{a}\psi$ with
  $a\mathbin\in X$ and $P \goesto{a} P'$ for some $P'$ with $P' \models \psi$.
  By \df{reactive bisimilarity}, $Q \goesto{a} Q'$ for some $Q'$ with $P' \bis{r} Q'$.
  By induction $Q' \models \psi$, so $Q \models_X \diam{a} \psi$.
\item Let $\phi = \diam{\tau}\psi$, and
  $P \goesto{\tau} P'$ for some $P'$ with $P' \models_X \psi$.
  By \df{reactive bisimilarity}, $Q \goesto{\tau} Q'$ for some $Q'$ with $P' \rbis{X}{r} Q'$.
  By induction $Q' \models_X \psi$, so $Q \models_X \diam{\tau} \psi$.
\item Let $\I(P)\cap (X\cup\{\tau\}) = \emptyset$ and $P \models \phi$.
  By the fifth clause of \df{reactive bisimilarity}, $P \bis{r} Q$.
  Hence, by the previous case in this proof, $Q \models \phi$. Moreover,
  $\I(Q)\cap (X\cup\{\tau\}) =\I(P) \cap (X\cup\{\tau\})$, since $P \rbis{X}{r} Q$.  Thus $Q \models_X \phi$.
\end{itemize}
\end{itemize}
``$\Leftarrow$'': Write $P \equiv Q$ for $\forall \phi\mathbin\in\IO.~(P \models \phi \Leftrightarrow Q \models \phi)$,
and $P \equiv_X Q$ for $\forall \phi\mathbin\in\IO.~(P \models_X \phi \Leftrightarrow Q \models_X \phi)$.
I show that the family of relations  $\equiv$, $\equiv_X$ for $X\subseteq A$ constitutes a gsrb.
\begin{itemize}
\item Suppose $P \equiv Q$ and $P \goesto\alpha P'$ with $\alpha\in A \cup \{\tau\}$. Let
$\mathcal{Q}^\dagger := \{Q^\dagger \in \IP \mid Q \goesto\alpha Q^\dagger \wedge P' \not\equiv Q^\dagger\}$.
For each $Q^\dagger \in \mathcal{Q}^\dagger$, let $\phi_{Q^\dagger}\in\IO$ be a formula such that $P'\models \phi_{Q^\dagger}$
and $Q^\dagger\not\models \phi_{Q^\dagger}$. (Such a formula always exists because $\IO$ is closed under negation.)
Define \plat{$\phi := \bigwedge_{Q^\dagger \in \mathcal{Q}^\dagger} \phi_{Q^\dagger}$}. Then $P' \models \phi$, so $P \models \diam{a}\phi$.
Consequently, also $Q \models \diam{a}\phi$. Hence there is a $Q'$ with $Q \goesto\alpha Q'$ and
$Q' \models \phi$. Since none of the $Q^\dagger\in\mathcal{Q}^\dagger$ satisfies $\phi$, one obtains
$Q' \notin \mathcal{Q}^\dagger$ and thus $P' \equiv Q'$.
\item Suppose $P \equiv Q$,~ $X\subseteq A$,~ $\I(P)\cap (X\cup\{\tau\}) = \emptyset$ and $P \goesto\rt P'$.
Let $$\mathcal{Q}^\dagger := \{Q^\dagger \in \IP \mid Q \goesto\rt Q^\dagger \wedge P' \not\equiv_X Q^\dagger\}.$$
For each $Q^\dagger \in \mathcal{Q}^\dagger$, let $\phi_{Q^\dagger}\in\IO$ be a formula such that $P'\models_X \phi_{Q^\dagger}$
and $Q^\dagger\not\models_X \phi_{Q^\dagger}$.
Define $\phi := \bigwedge_{Q^\dagger \in \mathcal{Q}^\dagger} \phi_{Q^\dagger}$. Then $P' \models_X \phi$, so $P \models \diam{X}\phi$.
Consequently, also $Q \models \diam{X}\phi$. Hence there is a $Q'$ with $Q \goesto\rt Q'$ and
$Q' \models_X \phi$. Again $Q' \notin \mathcal{Q}^\dagger$ and thus $P' \equiv_X Q'$.
\item Suppose $P \equiv_Y Q$ and $P \goesto\alpha P'$ with $a\in A$ and either $a\mathbin\in Y$ or
$\I(P)\cap (Y\cup\{\tau\})\mathbin=\emptyset$.

Let $\mathcal{Q}^\dagger := \{Q^\dagger \in \IP \mid Q \goesto\alpha Q^\dagger \wedge P' \not\equiv Q^\dagger\}$.
For each $Q^\dagger \in \mathcal{Q}^\dagger$, let $\phi_{Q^\dagger}\in\IO$ be a formula such that $P'\models \phi_{Q^\dagger}$
and $Q^\dagger\not\models \phi_{Q^\dagger}$.
Define \plat{$\phi := \bigwedge_{Q^\dagger \in \mathcal{Q}^\dagger} \phi_{Q^\dagger}$}. Then $P' \models \phi$, so $P \models \diam{a}\phi$,
and also $P \models_Y \diam{a}\phi$, using either the third or last clause in the definition of $\models_X$.
Hence also $Q \models_Y \diam{a}\phi$. Therefore there is a $Q'$ with $Q \goesto\alpha Q'$ and
$Q' \models \phi$, using the third clause of either $\models_X$ or $\models$.
Since none of the $Q^\dagger\in\mathcal{Q}^\dagger$ satisfies $\phi$, one obtains
$Q' \notin \mathcal{Q}^\dagger$ and thus $P' \equiv Q'$.
\item The fourth clause of \df{reactive bisimulation} is obtained exactly like the first,
  but using $\models_Y$ instead of $\models$.
\item Suppose $P \equiv_Y Q$, $P \goesto\rt P'$ and $\I(P)\cap (X \cup Y\cup\{\tau\}) = \emptyset$, with $X \subseteq A$.
Let $$\mathcal{Q}^\dagger := \{Q^\dagger \in \IP \mid Q \goesto\rt Q^\dagger \wedge P' \not\equiv_X Q^\dagger\}.$$
For each $Q^\dagger \in \mathcal{Q}^\dagger$, let $\phi_{Q^\dagger}\in\IO$ be a formula such that $P'\models_X \phi_{Q^\dagger}$
and $Q^\dagger\not\models_X \phi_{Q^\dagger}$.
Define $\phi := \bigwedge_{Q^\dagger \in \mathcal{Q}^\dagger} \phi_{Q^\dagger}$. Then $P' \models_X \phi$, so $P \models \diam{X}\phi$,
and thus $P \models_Y \diam{X}\phi$.
Consequently, also $Q \models_Y \diam{X}\phi$ and therefore $Q \models \diam{X}\phi$.
Hence there is a $Q'$ with $Q \goesto\rt Q'$ and
$Q' \models_X \phi$. Again $Q' \notin \mathcal{Q}^\dagger$ and thus $P' \equiv_X Q'$.
\qed
\end{itemize}
\end{proof}

\section{Time-out bisimulations}\label{sec:timeout bisimulations}

I will now present a characterisation of strong reactive bisimilarity in terms of a binary relation $\B$ on
processes---a \emph{strong time-out bisimulation}---not parametrised by the set of allowed actions $X$.
To this end I need a family of unary operators $\theta_X$ on processes, for $X\subseteq A$.
These \emph{environment} operators place a process in an environment that allows exactly the actions
in $X$ to occur.
They are defined by the following structural operational rules.
$$\frac{x \goesto{\tau} y}{\theta_X(x) \goesto{\tau} \theta_X(y)} \qquad
\frac{x \goesto{a} y}{\theta_X(x) \goesto{a} y}~(a\in X) \qquad
\frac{x \goesto{\alpha} y \quad x {\ngoesto\beta}~\mbox{for all}~\beta\in X\cup\{\tau\}}
{\theta_X(x) \goesto{\alpha} y}~(\alpha\in A\cup\{\rt\})$$
The operator $\theta_X$ modifies its argument by inhibiting all initial transitions (here
including also those that occur after a $\tau$-transition) that cannot occur in the specified
environment. When an observable transition does occur, the environment may be triggered to change,
and the inhibiting effect of the $\theta_X$-operator comes to an end. The premises
$x {\ngoesto\beta}~\mbox{for all}~\beta\in X\cup\{\tau\}$ in the third rule guarantee that the
process $x$ will idle for a positive amount of time in its current state. During this time,
the environment may be triggered to change, and again the inhibiting effect of the
$\theta_X$-operator comes to an end.

Below I assume that $\IP$ is closed under $\theta$, that is, if $P\in \IP$ and $X\subseteq A$ then $\theta_X(P)\in\IP$.
\begin{definition}{time-out bisimulation}
A \emph{strong time-out bisimulation} is a symmetric relation ${\B} \subseteq \IP \times \IP$, such
that, for $P\B Q$,
\begin{itemize}\itemsep 0pt \parsep 0pt
\item if $P \goesto{\alpha} P'$ with $\alpha\mathbin\in A\cup\{\tau\}$, then $\exists Q'$ such that
  $Q\goesto{\alpha} Q'$ and $P'\B Q'$,
\item if $\I(P)\cap (X\cup \{\tau\})=\emptyset$ and $P \goesto{\rt} P'$, then $\exists Q'$ such that $Q\goesto{\rt} Q'$
  and $\theta_{X}(P')\B \theta_X(Q')$.
\end{itemize}
\end{definition}

\begin{proposition}{time-out bisimulation}
$P \bis{r} Q$ iff there exists a strong time-out bisimulation $\B$ with $P \B Q$.
\end{proposition}

\begin{proof}
Let $\R$ be a gsrb on $\IP$. Define ${\B} \subseteq \IP \times \IP$ by
$P \B Q$ iff either $(P,Q)\mathbin\in \R$ or $P=\theta_{X}(P^\dagger)$, $Q=\theta_{X}(Q^\dagger)$ and $(P^\dagger,X,Q^\dagger)\in\R$.
I show that $\B$ is a strong time-out bisimulation.
\begin{itemize}\itemsep 0pt \parsep 0pt
\item Let $P \B Q$ and $P \goesto{a} P'$ with $a\mathbin\in A$.
  First suppose $(P,Q)\in\R$. Then, by the first clause of \df{reactive bisimulation},
  there exists a $Q'$ such that $Q\goesto{a} Q'$ and $(P', Q') \mathbin\in\R$.  So $P' \B Q'$.

  Next suppose $P=\theta_{X}(P^\dagger)$, $Q=\theta_{X}(Q^\dagger)$ and $(P^\dagger,X,Q^\dagger)\in\R$.
  Since $\theta_{X}(P^\dagger)\goesto{a} P'$ it must be that $P^\dagger\goesto{a} P'$ and 
  either $a\in X$ or $P^\dagger {\ngoesto\beta}$ for all $\beta\in X\cup\{\tau\}$.
  Hence there exists a $Q'$ such that $Q^\dagger\goesto{a} Q'$ and $(P', Q') \mathbin\in\R$,
  using the third clause of \df{reactive bisimulation}.
  Recall that $P^\dagger \rbis{X}r Q^\dagger$ implies $I(P^\dagger)\cap(X\cup\{\tau\}) \mathbin= I(Q^\dagger)\cap(X\cup\{\tau\})$,
  and thus either $a\mathbin\in X$ or $Q^\dagger {\ngoesto\beta}$ for all $\beta\mathbin\in X\cup\{\tau\}$.
  It follows that $Q=\theta_{X}(Q^\dagger)\goesto{a} Q'$ and $P' \B Q'$.

\item Let $P \B Q$ and $P \goesto{\tau} P'$.
  First suppose $(P,Q)\in\R$. Then, using the first clause of \df{reactive bisimulation},
  there is a $Q'$ with $Q\goesto{\tau} Q'$ and $(P', Q') \mathbin\in\R$.  So $P' \B Q'$.

  Next suppose $P=\theta_{X}(P^\dagger)$, $Q=\theta_{X}(Q^\dagger)$ and $(P^\dagger,X,Q^\dagger)\in\R$.
  Since $\theta_{X}(P^\dagger)\goesto{\tau} P'$, it must be that $P'$ has the form $\theta_{X}(P^\ddagger)$,
  and $P^\dagger\goesto{\tau} P^\ddagger$.
  Thus, by the fourth clause of  \df{reactive bisimulation},
  there is a $Q^\ddagger$ with $Q^\dagger\goesto{\tau} Q^\ddagger$ and $(P^\ddagger,X, Q^\ddagger) \mathbin\in\R$.
  Now $Q = \theta_{X}(Q^\dagger) \goesto{\tau} \theta_{X}(Q^\ddagger) =: Q'$ and $P' \B Q'$.

\item Let $P\B Q$, $\I(P)\cap (X\cup\{\tau\})=\emptyset$ and $P \goesto{\rt} P'$.
  First suppose $(P,Q)\in\R$. Then, by the second clause of \df{reactive bisimulation},
  there is a $Q'$ with $Q\goesto{\rt} Q'$ and $(P',X,Q')\in\R$.
  So $\theta_{X}(P')\B \theta_X(Q')$.

  Next suppose $P=\theta_{Y}(P^\dagger)$, $Q=\theta_{Y}(Q^\dagger)$ and $(P^\dagger,Y,Q^\dagger)\in\R$.
  Since $\theta_{Y}(P^\dagger)\goesto{\rt} P'$, it must be that $P^\dagger\goesto{\rt} P'$
  and $P^\dagger {\ngoesto\beta}$ for all $\beta\in Y\cup\{\tau\}$.
  Consequently, $\I(P^\dagger)=\I(P)$ and thus $\I(P^\dagger) \cap (X\cup Y\cup\{\tau\}) = \emptyset$.
  By the last clause of \df{reactive bisimulation} there is a $Q'$ such that $Q^\dagger\goesto{\rt} Q'$ and
  $(P,X,Q')\in\R$. So
  $\theta_{X}(P')\B \theta_X(Q')$. From $(P^\dagger,Y,Q^\dagger)\in\R$ and $\I(P^\dagger)\cap (Y\cup\{\tau\})=\emptyset$,
  I infer $\I(Q^\dagger)\cap (Y\cup\{\tau\})=\emptyset$. 
  So $Q^\dagger {\ngoesto\beta}$ for all $\beta\in Y\cup\{\tau\}$.
  This yields $Q = \theta_{Y}(Q^\dagger) \goesto{\rt} Q'$.
\end{itemize}
Now let $\B$ be a time-out bisimulation. Define $\R\subseteq \IP \times \Pow(A) \times \IP$ by
$(P,Q)\in \R$ iff $P \B Q$, and
$(P,X,Q)\in \R$ iff $\theta_X(P) \B \theta_X(Q)$.
I need to show that $\R$ is a gsrb.
\begin{itemize}\itemsep 0pt \parsep 0pt
\item Suppose $(P,Q)\in\R$ and $P \goesto{\alpha} P'$ with $\alpha \in A\cup\{\tau\}$.
  Then $P \B Q$, so there is a $Q'$ such that $Q\goesto{\alpha} Q'$ and $P' \B Q'$. Hence $(P',Q')\in\R$.
\item Suppose $(P,Q)\in\R$, $X\subseteq A$, $\I(P)\cap (X\cup\{\tau\})=\emptyset$ and $P \goesto{\rt} P'$.
  Then $P \B Q$, so $\exists Q'$ such that $Q\goesto{\rt} Q'$ and $\theta_{X}(P')\B \theta_X(Q')$.
  Thus $(P',X,Q')\in\R$.
\item Suppose $(P,X,Q)\in\R$ and $P \goesto{a} P'$ with either $a\mathbin\in X$ or $\I(P)\cap (X\cup\{\tau\})=\emptyset$.
  Then $\theta_X(P) \B \theta_X(Q)$.
  Moreover, $\theta_X(P) \goesto{a} P'$.
  Hence there is a $Q'$ such that $\theta_X(Q)\goesto{a} Q'$ and $P'\B Q'$.
  It must be that $Q\goesto a Q'$.
  Moreover, $(P',Q') \in\R$.
\item Suppose $(P,X,Q)\in\R$ and $P \goesto{\tau} P'$.
  Then $\theta_X(P) \B \theta_X(Q)$.
  Since $P \goesto{\tau} P'$, one has $\theta_X(P) \goesto{\tau} \theta_X(P')$.
  Hence there is an $R$ such that $\theta_X(Q)\goesto{\tau} R$ and $\theta_X(P')\B R$.
  The process $R$ must have the form $\theta_X(Q')$ for some $Q'$ with $Q\goesto\tau Q'$.
  It follows that $(P',X,Q')\in\R$.
\item Suppose $(P,Y,Q)\in\R$, $X\subseteq A$, $\I(P)\cap (X\cup Y \cup \{\tau\})=\emptyset$ and $P \goesto{\rt} P'$.
  Then $\theta_Y(P) \B \theta_Y(Q)$ and $\theta_Y(P) \goesto{\rt} P'$.
  Moreover, $\I(\theta_Y(P))=\I(P)$, so by the second clause of \df{time-out bisimulation}
  there exists a $Q'$ such that $\theta_Y(Q)\goesto{\rt} Q'$ and $\theta_{X}(P')\B \theta_X(Q')$.
  So $Q \goesto{\rt} Q'$ and $(P',X,Q')\mathbin\in\R$. 
\qed
\end{itemize}
\end{proof}
%
Note that the union of arbitrarily many strong time-out bisimulations is itself a strong time-out bisimulation.
Consequently, the relation $\bis{r}$ is a strong time-out bisimulation.

\section[The process algebra CCSP]{The process algebra $\CCSP$}\label{sec:ccsp}

Let $A$ be a set of \emph{visible actions} and $\Var$ an infinite set of \emph{variables}.
The syntax of $\CCSP$ is given by\vspace{-1pt}
$$E ::= 0 ~\mbox{\Large $\,\mid\,$}~ \alpha.E ~\mbox{\Large $\,\mid\,$}~ E+E
~\mbox{\Large $\,\mid\,$}~ E \spar{S} E ~\mbox{\Large $\,\mid\,$}~ \tau_I(E) ~\mbox{\Large $\,\mid\,$}~\Rn(E) \mbox{\Large
~$\,\mid\,$}~ \theta_L^U(E) ~\mbox{\Large $\,\mid\,$}~ \psi_X(E) ~\mbox{\Large $\,\mid\,$}~ x ~\mbox{\Large $\,\mid\,$}~\rec{x|\RS}\mbox{ (with }x \mathbin\in V_\RS)$$
with $\alpha \mathbin\in Act := A \uplus\{\tau,\rt\}$, $S,I,U,L,X\mathbin\subseteq A$, $L \subseteq U$,
$\Rn \mathbin\subseteq A \mathop\times A$,
$x \mathbin\in \Var$ and $\RS$ a {\em recursive specification}: a set of equations
$\{y = \RS_{y} \mid y \mathbin\in V_\RS\}$ with $V_\RS \subseteq \Var$
(the {\em bound variables} of $\RS$) and each $\RS_{y}$ a $\CCSP$ expression.
I require that all sets ${\{b\mid (a,b)\in \Rn\}}$ are finite.

The constant $0$ represents a process that is unable to perform any
action. The process $\alpha.E$ first performs the action $\alpha$ and then
proceeds as $E$. The process $E+F$ behaves as either $E$ or $F$.
$\spar{S}$ is a partially synchronous parallel composition operator; 
actions $a\in S$ must synchronise---they can occur only when both arguments
are ready to perform them---whereas actions $\alpha\notin S$ from both arguments are interleaved.
$\tau_I$ is an abstraction operator; it conceals the actions in $I$ by renaming them into the hidden
action $\tau$.
The operator $\Rn$ is a relational renaming: it renames a given action $a\in A$ into a choice between
all actions $b$ with $(a,b)\mathbin\in \Rn$.
The \emph{environment operators} $\theta_L^U$ and $\psi_X$ are new in this paper and explained below.
Finally, $\rec{x|\RS}$ represents the $x$-component of a solution of the system of recursive equations $\RS$.

The language CCSP is a common mix of the process algebras CCS \cite{Mi90ccs} and CSP \cite{BHR84,Ho85}.
It first appeared in \cite{Ol87}, where it was named following a suggestion by M. Nielsen.
The family of parallel composition operators $\|_S$ stems from \cite{OH86}, and incorporates
the two CSP parallel composition operators from \cite{BHR84}.
The relation renaming operators $\Rn(\_\!\_)$ stem from \cite{Va93}; they combine both the
(functional) renaming operators that are common to CCS and CSP, and the inverse image operators of CSP\@.
The choice operator $+$ stems from CCS, and the abstraction operator from CSP, while the inaction
constant $0$, action prefixing operators $a.\_\!\_\,$ for $a\in A$, and the recursion construct are common to CCS and CSP\@.
The time-out prefixing operator $\rt.\_\!\_\,$ was added by me in \cite{vG21}.
The syntactic form of inaction $0$, action prefixing $\alpha.E$ and choice $E+F$ follows CCS,
whereas the syntax of abstraction $\tau_I(\_\!\_)$ and recursion $\rec{x|\RS}$ follows ACP
\cite{BW90,Fok00}.
The fragment of $\CCSP$ without $\theta_L^U$ and $\psi_X$ is called CCSP$_\rt$ \cite{vG21}.

An occurrence of a variable $x$ in a $\CCSP$ expression $E$ is \emph{bound} iff it occurs in a
subexpression $\rec{y|\RS}$ of $E$ with $x \mathbin\in V_\RS$; otherwise it is \emph{free}.
Here each $\RS_y$ for $y \mathbin\in V_\RS$ counts as a subexpression of $\rec{x|\RS}$.
An expression $E$ is \emph{invalid} if it has a subexpression $\theta_L^U(F)$ or $\psi_X(F)$
such that a variable occurrence in $F$ is free in $F$ but bound in $E$.
Let $\IT$ be the set of valid $\CCSP$ expressions.
Furthermore, $\IP\subseteq\IT$ is the set of {\em closed} valid $\CCSP$ expressions,
or \emph{processes}; those in which every variable occurrence is bound.

A substitution is a partial function $\rho\!:\!\Var \mathbin\rightharpoonup \IT$.
The application $E[\rho]$ of a substitution $\rho$ to an expression $E\mathbin\in\IT$ is the result of
simultaneous replacement, for all $x\mathbin\in \textrm{dom}(\rho)$, of each free occurrence of $x$\linebreak
in $E$ by the expression $\rho(x)$, while renaming bound variables in $E$ if necessary to prevent name clashes.

The semantics of $\CCSP$ is given by the labelled transition relation
$\mathord\rightarrow \subseteq \IP\times Act \times\IP$, where the transitions 
{$P\goesto{\alpha}Q$} are derived from the rules of \tab{sos CCSP}.
Here $\rec{E|\RS}$ for $E \in \IT$ and $\RS$ a recursive specification
denotes the result of substituting $\rec{y|\RS}$ for $y$ in $E$, for all $y \mathbin\in V_\RS$.

\begin{table}[t]
\vspace{-6pt}
\caption{Structural operational interleaving semantics of $\CCSP$}
\label{tab:sos CCSP}
\begin{center}
\framebox{$\begin{array}{c@{\qquad}c@{\qquad}c}
\multicolumn{3}{c}{
\alpha.x \goesto{\alpha} x \qquad
\displaystyle\frac{x \goesto{\alpha} x'}{x+y \goesto{\alpha} x'} \qquad
\displaystyle\frac{y \goesto{\alpha} y'}{x+y \goesto{\alpha} y'} \qquad
\displaystyle\frac{x \goesto{\alpha} x'} {\Rn(x) \goesto{\beta} \Rn(x')}
~\left( \begin{array}{@{}r@{}} \scriptstyle\alpha=\beta=\tau \\[-3pt]
\scriptstyle \vee~~ \alpha=\beta=\rt\\[-3pt]\scriptstyle \vee~(\alpha,\beta)\in \Rn\!\end{array}\right)}\\[1.5em]

\displaystyle\frac{x \goesto{\alpha} x'}{x\spar{S} y \goesto{\alpha} x'\spar{S} y}~(\alpha\not\in S) &
\displaystyle\frac{x \goesto{a} x'\quad y \goesto{a} y'}{x\spar{S} y \goesto{a} x'\spar{S} y'} ~(a\in S) &
\displaystyle\frac{y \goesto{\alpha} y'}{x\spar{S} y \goesto{\alpha} x\spar{S} y'}~(\alpha\not\in S) \\[1.5em]

\displaystyle\frac{x \goesto{\alpha} x'}{\tau_I(x) \goesto{\alpha} \tau_I(x')}~(\alpha\not\in I) &
\displaystyle\frac{x \goesto{a} x'}{\tau_I(x) \goesto{\tau} \tau_I(x')}~(a\in I) &
\displaystyle\frac{\rec{\RS_{x}|\RS} \goesto{\alpha} y}{\rec{x|\RS}\goesto{\alpha}y} \\[1.5em]

\multicolumn{3}{c@{}}{\displaystyle
\frac{x \goesto{\tau} y}{\theta_L^U(x) \goesto{\tau} \theta_L^U(y)} \qquad
\frac{x \goesto{a} y}{\theta_L^U(x) \goesto{a} y}~(a\mathbin{\in} U) \qquad
\frac{x \goesto{\alpha} y \quad x {\ngoesto\beta}~\mbox{for all}~\beta\mathbin\in L\mathord\cup\{\tau\}}
{\theta_L^U(x) \goesto{\alpha} y}~(\alpha\mathbin\in A\mathop{\cup}\{\rt\})}\\[1.5em]

\multicolumn{3}{c@{}}{\displaystyle
\frac{x \goesto{\alpha} y}{\psi_X(x) \goesto{\alpha} y}~(\alpha\in A\cup\{\tau\}) \qquad
\frac{x \goesto{\rt} y \quad x {\ngoesto\beta}~\mbox{for all}~\beta\mathbin{\in} X \cup\{\tau\}}
{\psi_X(x) \goesto{\rt} \theta_X(y)}}\\[1em]

\end{array}$}
\vspace{-3pt}
\end{center}
\end{table}

The auxiliary operators $\theta_L^U$ and $\psi_X$ are added here to facilitate complete
axiomatisation, similar to the left merge and communication merge of ACP \cite{BW90,Fok00}.  The operator
$\theta_X^X$ is the same as what was called $\theta_X$ in \Sec{timeout bisimulations}.  It inhibits
those transitions of its argument that are blocked in the environment $X$, allowing only the
actions from $X\subseteq A$. It stops inhibiting as soon as the system performs a visible action or
takes a break, as this may trigger a change in the environment.  The operator $\theta_L^U$ preserves
those transitions that are allowed in some environment $X$ with $L\subseteq X \subseteq U$. The
letters $L$ and $U$ stand for \emph{lower} and \emph{upper} bound.
The operator $\psi_X$ places a process in the environment $X$ when
a time-out transition occurs; it is inert if any other transition occurs.
If $P {\goesto\beta}$ for $\beta\in A\cup\{\tau\}$, then a time-out transition $P \goesto\rt Q$ cannot
occur in an environment that allows $\beta$. Thus the transition $P \goesto\rt Q$ survives only when
considering an environments that blocks $\beta$, meaning $\beta \notin X\cup\{\tau\}$. Taking the
contrapositive, $\beta \in X\cup\{\tau\}$ implies $P {\ngoesto\beta}$.

The operator $\theta^U_\emptyset$ features in the forthcoming law \hyperlink{L3}{L3}, which is a convenient addition 
to my axiomatisation, although only $\psi_X$ and $\theta_X$ ($= \theta_X^X$) are necessary for completeness.

\paragraph{Stratification.}\label{sec:stratification}

Even though negative premises occur in \tab{sos CCSP}, the meaning of this transition system
specification is well-defined, for instance by the method of \emph{stratification} explained in \cite{Gr93,vG04}.
Assign inductively to each expression $E\in \IT$ an ordinal $\lambda_E$ that counts the nesting depth of recursive
specifications: if $E = \rec{x|\RS}$ then $\lambda_E$ is 1 more than the supremum of the
$\lambda_{S_y}$ for $y \in V_\RS$; otherwise $\lambda_E$ is the supremum of $\lambda_{\rec{x|\RS}}$
for all subterms $\rec{x|\RS}$ of $E$. Moreover $\kappa_E \in \IN$ is the nesting depth of $\theta_L^U$ and
$\psi_X$ operators in $E$ that remain after replacing any subterm $F$ of $E$ with $\lambda_F < \lambda_E$ by $0$.
Now the ordered pair $(\lambda_P,\kappa_P)$ constitutes a valid stratification for closed literals
$P \goesto\alpha P'$. Namely, whenever a transition $P \goesto\alpha P'$ depends on a transition $Q\goesto \beta Q'$,
in the sense that that there is a closed substitution instance $\mathfrak{r}$ of a rule from \tab{sos CCSP}
with conclusion $P \goesto\alpha P'$, and $Q\goesto \beta Q'$ occurring in its premises,
then either $\lambda_Q < \lambda_P$, or $\lambda_Q = \lambda_P$ and $\kappa_Q \mathbin\leq \kappa_P$.
Moreover, when $P \goesto\alpha P'$ depends on a negative literal $Q {\ngoesto\beta}$, then 
$\lambda_Q \mathbin= \lambda_P$ and $\kappa_Q \mathbin< \kappa_P$.

The above argument hinges on the exclusion of invalid $\CCSP$ expressions.
The invalid expression $P:= \rec{x \mid \{x = \theta^{\{a\}}_{\{a\}} (b.0 + \Rn(x))\}}$ for instance, with
$\Rn = \{(b,a)\}$, does not have a well-defined meaning, since the transition $P \goesto{b} 0$ is derivable
iff one has the premise $P {\ngoesto{b}}$:\vspace{-5pt}
\[
\frac{\displaystyle
 \frac{\displaystyle
  \frac{}{b.0\goesto{b} 0\phantom{()}}
 }{b.0 + \Rn(P) \goesto{b} 0}
 \qquad
 \frac{\displaystyle
  \frac{P {\ngoesto{b}}}{\Rn(P){\ngoesto{a}}}
 }{b.0 + \Rn(P) {\ngoesto{a}}
 }
 \qquad
 \frac{\displaystyle
  \frac{P {\ngoesto{\tau} \makebox[0pt][l]{\quad \scriptsize(OK)}}}{\Rn(P){\ngoesto{\tau}}}
 }{b.0 + \Rn(P) {\ngoesto{\tau}}
 }
}{\displaystyle
 \frac{ \theta^{\{a\}}_{\{a\}} (b.0 + \Rn(P)) \goesto{b} 0}{P \goesto{b} 0}
}
\]
However, the meaning of the valid expression
$\rec{x \mid \{x = \theta^{\{a\}}_{\{a\}} (\rec{y|\{y= b.y\}}) \spar{\emptyset} \Rn(x)\}}$, for
instance, is entirely unproblematic.

\section{Guarded recursion and finitely branching processes}\label{sec:guarded}

In many process algebraic specification approaches, only guarded recursive specifications are allowed.
\begin{definition}{guarded}
An occurrence of a variable $x$ in an expression $E$ is \emph{guarded} if $x$ occurs in
a subexpression $\alpha.F$ of $E$, with $\alpha\mathbin\in Act$.
An expression $E$ is \emph{guarded} if all free occurrences of variables in $E$ are guarded.
A recursive specification $\RS$ is \emph{manifestly guarded} if all expressions $\RS_y$ for $y\in V_\RS$ are guarded.
It is \emph{guarded} if it can be converted into a manifestly guarded recursive specification by
repeated substitution of expressions $\RS_y$ for variables $y\in V_\RS$ occurring in the expressions
$\RS_z$ for $z\in V_\RS$.
Let \emph{guarded} $\CCSP$ be the fragment of $\CCSP$ allowing only guarded recursion.
\end{definition}

\begin{definition}{finitely branching}
The set of processes \emph{reachable} from a given process $P\in\IP$ is inductively defined by
\begin{enumerate}[(i)]
\item $P$ is reachable from $P$, and
\item if $Q$ is reachable from $P$ and $Q \goesto\alpha R$ for
some $\alpha\in Act$ then $R$ is reachable from $P$.
\end{enumerate}
A process $P$ is \emph{finitely branching} if for all $Q\in\IP$ reachable from $P$ there are only
finitely many pairs $(\alpha,R)$ such that $Q \goesto\alpha R$.
Likewise, $P$ is \emph{countably branching} if there are countably many such pairs.
A process is \emph{finite} iff it is finitely branching, has finitely many reachable states, and is loop-free,
in the sense that there are no $Q_0 \goesto{\alpha_1} Q_1 \goesto{\alpha_2} \cdots \goesto{\alpha_n} Q_n$
with $n>0$ and $Q_0=Q_n$ reachable from $P$.
\end{definition}

\begin{proposition}{countably branching}
Each $\CCSP$ process is countably branching.
\end{proposition}

\begin{proof}
I show that for each $\CCSP$ process $Q$ there are only countably many transitions $Q \goesto\alpha R$.
Each such transition must be derivable from the rules of \tab{sos CCSP}.  So it suffices to
show that for each $Q$ there are only countably many derivations of transitions $Q \goesto\alpha R$.

A derivation of a transition is a well-founded, upwardly branching tree, in which each node
models an application of one of the rules of \tab{sos CCSP}. Since each of these rules has finitely
many positive premises, such a proof tree is finitely branching, and thus finite.
Let $d(\pi)$, the \emph{depth} of $\pi$, be the length of the longest branch in a derivation $\pi$.
If $\pi$ derives a transition $Q \goesto\alpha R$, then I call $Q$ the \emph{source} of $\pi$.

It suffices to show that for each $n\in\IN$
there are only finitely many derivations of depth $n$ with a given source.
This I do by induction on $n$. 

In case $Q=f(Q_1,\dots,Q_k)$, with $f$ an $k$-ary $\CCSP$ operator, a derivation $\pi$ of depth $n$
is completely determined by the concluding rule from \tab{sos CCSP}, deriving a transition
$Q\goesto\beta R$, the subderivations of $\pi$ with source $Q_i$ for some of the $i\in\{1,\dots,k\}$,
and the transition label $\beta$. (For the purposes of this proof, \tab{sos CCSP} is understood to
have only 15 rules, even if each of them can be seen as a template, with an instance for each
choice of $\Rn$, $S$, $I$, $\RS$ etc., and for each fitting choice of a transition labels $a$,
$\alpha$ and/or $\beta$.) The choice of the concluding rule depends on $f$, and for each $f$
there are at most three choices. The subderivations of $\pi$ with source $Q_i$ have depth $< n$, so
by induction there are only finitely many. When $f$ is not a renaming operator $\Rn$, there is no further
choice for the transition label $\beta$, as it is completely determined by the premises of the rule, and
thus by the subderivations of those premises. In case $f=\Rn$, there are finitely many choices for $\beta$
when faced with a given transition label $\alpha$ contributed by the premise of the rule for renaming.
Here I use the requirement of \Sec{ccsp} that all sets ${\{b\mid (a,b)\in \Rn\}}$ are finite.
This shows there are only finitely many choices for $\pi$.

In case $Q=\rec{x|\RS}$, the last step in $\pi$ must be application of the rule for recursion, so
$\pi$ is completely determined by a subderivation $\pi'$ of a transition with source $\rec{\RS_x|\RS}$.
By induction there are only finitely many choices for $\pi'$, and hence also for $\pi$.
\end{proof}

\begin{proposition}{finitely branching}
Each $\CCSP$ process with guarded recursion is finitely branching.
\end{proposition}

\begin{proof}
A trivial structural induction shows that if $P$ is a $\CCSP$ process with guarded recursion and $Q$
is reachable from $P$, then also $Q$ has guarded recursion.  Hence it suffices to show that for
each $\CCSP$ process $Q$ with guarded recursion there are only finitely many derivations with source $Q$.

Let $\rightsquigarrow$ be the smallest binary relation on $\IP$ such that (i) $f(P_1,\dots,P_k) \rightsquigarrow P_i$
for each $k$-ary $\CCSP$ operator $f$ except action prefixing, and each $i\in\{1,\dots,k\}$, and
(ii) $\rec{x|\RS} \rightsquigarrow \rec{\RS_x|\RS}$. This relation is finitely branching.
Moreover, on processes with guarded recursion, $\rightsquigarrow$ has no forward infinite chains
$P_0 \rightsquigarrow P_1 \rightsquigarrow \dots$.
In fact, this could have been used as an alternative definition of guarded recursion.
Let, for any process $Q$ with guarded recursion, $e(Q)$ be the length of the longest forward chain
$Q \rightsquigarrow P_1 \rightsquigarrow \dots \rightsquigarrow P_{e(Q)}$.
I show with induction on $e(Q)$ that there are only finitely many derivations with source $Q$.
In fact, this proceeds exactly as in the previous proof.
\end{proof}

\begin{propositionq}{absolute expressiveness guarded}{\cite{vG94a}}
Each finitely branching processes in an LTS can be denoted by a closed CCSP$_\rt$ expression with guarded recursion.
Here I only need the operations inaction ($0$), action prefixing ($\alpha.\_\!\_\,$) and choice ($+$),
as well as recursion $\rec{x|\RS}$.
\end{propositionq}

\begin{proof}
Let $P$ be a finitely branching process in an LTS $(\IP',Act,\rightarrow)$.
Let $$V_\RS := \{x_Q \mid Q\in\IP' \mbox{~is reachable from~} P\} \subseteq\Var.$$
For each $Q$ reachable from $P$, let $\textit{next}(Q)$ be the finite set of pairs
$(\alpha,R)\in Act \times \IP'$ such that there is a transition $Q \goesto\alpha R$.
Define the recursive specification $\RS$ as
$\{x_Q = \sum_{(\alpha,R)\in\textit{next}(Q)} \alpha.x_R \mid  x_Q \in V_\RS\}$.
Here the finite choice operator $\sum_{i\in I}\alpha_i.P_i$ can easily be expressed in terms of
inaction, action prefixing and choice. Now the CCSP$_\rt$ process $\rec{x_P|\RS}$ denotes $P$.
\end{proof}
In fact, $\rec{x_P|\RS} \bis{\,} P$, where $\bis{\,}$ denotes strong bisimilarity \cite{Mi90ccs},
formally defined in the next section.

Likewise, recursion-free $\CCSP$ processes are finite, and, up to strong bisimilarity, each finite process is
denotable by a closed recursion-free $\CCSP$ expression, using only $0$, $\alpha.\_\!\_\,$ and $+$.

\begin{propositionq}{absolute expressiveness}{\cite{vG94a}}
Each countably branching processes in an LTS can be denoted by a closed CCSP$_\rt$ expression.
Again I only need the CCSP$_\rt$ operations inaction, action prefixing, choice and recursion.
\end{propositionq}
\begin{proof}
The proof is the same as the previous one, except that $\textit{next}(Q)$ now is a countable set,
rather than a finite one, and consequently I need a countable choice operator $\sum_{i\in \INs}\alpha_i.P_i$.
The latter can be expressed in CCSP$_\rt$ with unguarded recursion by
$\sum_{i\in \INs}\alpha_i.P_i := \rec{z_o | \{z_i = \alpha_i.P_i + z_{i+1} \mid i\in\IN\}}$.
\end{proof}

\section{Congruence}\label{sec:congruence}

Given an arbitrary process algebra with a collection of operators $f$, each with an arity $n$,
and a recursion construct $\rec{x|\RS}$ as in \Sec{ccsp}, let $\IP$ and $\IT$ be the sets of [closed] valid
expressions, and let a substitution instance $E[\rho]\in \IT$ for $E\in \IT$ and $\rho:\Var\rightharpoonup \IT$
be defined as in \Sec{ccsp}.
Any semantic equivalence ${\sim} \subseteq \IP\times\IP$ extends to 
${\sim} \subseteq \IT\times\IT$ by defining $E \sim F$ iff $E[\rho] \sim F[\rho]$ for each closed
substitution $\rho:\Var \rightarrow\IP$.
It extends to substitutions $\rho,\nu:\Var\rightharpoonup \IT$
by $\rho\sim \nu$ iff $\textrm{dom}(\rho) = \textrm{dom}(\nu)$ and
$\rho(x) \sim \nu(x)$ for each $x\in \textrm{dom}(\rho)$.

\begin{definitionq}{congruence}{\cite{vG17b}}
A semantic equivalence ${\sim}$ is a \emph{lean congruence} if
$E[\rho] \sim E[\nu]$ for any expression $E\in\IT$ and any substitutions $\rho$ and $\nu$
with $\rho \sim \nu$.
It is a \emph{full congruence} if it satisfies\vspace{-1ex}%
\begin{equation}\label{comp-operators-closed}
P_i \sim Q_i ~\mbox{for all}~i=1,...,n ~~\Rightarrow~~
 f(P_1,...,P_n) \sim f(Q_1,...,Q_n)
\vspace{-4ex}
\end{equation}
\begin{equation}\label{comp-recursion-closed}
\RS_y \sim \RS_y' ~\mbox{for all}~y\mathbin\in V_\RS
~~\Rightarrow~~ \rec{x|\RS} \sim \rec{x|\RS'}
\end{equation}
for all functions $f$ of arity $n$, processes $P_i,Q_i\mathbin\in \IP$, and
recursive specifications $\RS,\RS'$ with $x \mathbin\in V_\RS \mathbin= V_{\RS'}$
and $\rec{x|\RS},\rec{x|\RS'}\in\IP$.
\end{definitionq}
Clearly, each full congruence is also a lean congruence, and each lean congruence satisfies
(\ref{comp-operators-closed}) above. Both implications are strict, as illustrated in \cite{vG17b}.

A main result of the present paper will be that strong reactive bisimilarity is a full congruence
for the process algebra $\CCSP$. To achieve it I need to establish first that
strong bisimilarity \cite{Mi90ccs}, $\bis{}$, and initials equivalence
\cite[Section~16]{vG01}, $=_\I$, are full congruences for $\CCSP$.

\subsection{Initials equivalence}

\begin{definition}{initials equivalence}
Two $\CCSP$ processes $P$ and $Q$ are \emph{initials equivalent}, denoted $P =_\I Q$, if $\I(P)=\I(Q)$.
\end{definition}

\begin{theorem}{initials congruence}
Initials equivalence is a full congruence for $\CCSP$.
\end{theorem}
\begin{proof}
In Appendix \ref{initials congruence}.
\end{proof}

\subsection{Strong bisimilarity}\label{sec:strong}

\begin{definition}{bisimulation}
A \emph{strong bisimulation} is a symmetric relation $\B$ on $\IP$, such that, whenever
$P \B Q$,\vspace{-2pt}
\begin{itemize}
\item if $P \goesto{\alpha} P'$ with $\alpha\mathbin\in Act$ then $Q\goesto{\alpha} Q'$ for some
  $Q'$ with $P'\B Q'$.\vspace{-2pt}
\end{itemize}
Two processes $P,Q\mathbin\in\IP$ are \emph{strongly bisimilar}, $P \bis{} Q$, if $P \mathrel\R Q$ for some strong bisimulation $\B$.
\end{definition}
Contrary to reactive bisimilarity, strong bisimilarity treats the time-out action $\rt$, as well as
the hidden action $\tau$, just like any visible action. In the absence of time-out actions, there is
no difference between a strong bisimulation and a time-out bisimulation, so $\bis{r}$ and $\bis{}\,$ coincide.
In general, strong bisimulation is a finer equivalence relation than strong reactive bisimilarity and initials equivalence:
$P \bis{\,} Q \Rightarrow P \bis{r} Q\linebreak[3] \Rightarrow P =_\I Q$, and both implications are strict.

\begin{lemma}{absolute expressiveness}
For each $\CCSP$ process $P$ there exists a CCSP$_\rt$ process $Q$ only built using
inaction, action prefixing, choice and recursion, such that $P \bis{\,} Q$.
\end{lemma}
\begin{proof}
Immediately from Propositions~\ref{pr:countably branching} and~\ref{pr:absolute expressiveness}.
\pagebreak[3]
\end{proof}

\begin{theorem}{bisimilarity congruence}
Strong bisimilarity is a full congruence for $\CCSP$.
\end{theorem}

\begin{proof}
The structural operational rules for CCSP$_\rt$ (that is, $\CCSP$ without the operators $\theta_L^U$ and $\psi_X$)
fit the \emph{tyft/tyxt format with recursion} of \cite{vG17b}. By \cite[Theorem~3]{vG17b} this
implies that $\bis{\,}$ is a full congruence for CCSP$_\rt$. (In fact, when omitting the recursion
construct, the operational rules for CCSP$_\rt$ fit the  \emph{tyft/tyxt format} of \cite{GrV92}, and by 
the main theorem of \cite{GrV92}, $\bis{\,}$ is a congruence \emph{for the operators of} CCSP$_\rt$, that
is, it satisfies (\ref{comp-operators-closed}) in \df{congruence}. The work of \cite{vG17b} extends
this result of \cite{GrV92} with recursion.)

The structural operational rules for all of $\CCSP$ fit the \emph{ntyft/ntyxt format with recursion} of \cite{vG17b}.
By \cite[Theorem~2]{vG17b} this implies that $\bis{\,}$ is a lean congruence for $\CCSP$.
(In fact, when omitting the recursion construct, the operational rules for $\CCSP$ fit the
\emph{ntyft/ntyxt format} of \cite{Gr93}, and by the main theorem of \cite{Gr93}, $\bis{\,}$ is a
congruence for the operators of $\CCSP$. The work of \cite{vG17b} extends this result of \cite{Gr93} with recursion.)

To verify (\ref{comp-recursion-closed}) for the whole language $\CCSP$, let $\RS$ and $\RS'$ be
recursive specifications with $x \in V_\RS = V_{\RS'}$, such that $\rec{x|\RS},\rec{x|\RS'}\in\IP$
and $\RS_y \bis{\,} \RS_y'$ for all $y\mathbin\in V_\RS$.
Let $\{P_i \mid i \in I\}$ be the collection of processes of the form \plat{$\theta_L^U(Q)$} or
\plat{$\psi_X(Q)$}, for some $L$, $U$, $X$, that occur as a closed subexpression of $\RS_y$ or $\RS'_y$
for one of the $y\in V_\RS$, not counting strict subexpressions of a closed subexpression $R$ of $\RS_y$ or
$\RS'_y$ that is itself of the form \plat{$\theta_L^U(Q)$} or \plat{$\psi_X(Q)$}.
Pick a fresh variable $z_i\notin V_\RS$ for each $i\in I$, and let, for $y\in V_\RS$, \plat{$\widehat \RS_y$}
be the result of replacing each occurrence of $P_i$ in $\RS_y$ by $z_i$.
Then \plat{$\widehat \RS_y$} does not contain the operators $\theta_L^U(Q)$ or $\psi_X(Q)$.
In deriving this conclusion it is essential that $\rec{x|\RS}$ is a valid expression,
for this implies that the term $\RS_y\in \IT$, which may contain free occurrences of the variables $y\in V_\RS$,
does not have a subterm of the form \plat{$\theta_L^U(F)$} or \plat{$\psi_X(F)$} that contains 
free occurrences of these variables. Let \plat{$\widehat \RS := \{y = \widehat \RS_y \mid y \in V_\RS\}$};
it is a recursive specification in the language CCSP$_\rt$.
The recursive specification \plat{$\widehat \RS'$} is defined in the same way.

For each $i\in I$ there is, by \lem{absolute expressiveness}, a process $Q_i$ in the language CCSP$_\rt$
such that $P_i \bis{\,} Q_i$.
Now let $\rho,\eta:\{z_i \mid i \in I\} \rightarrow \IP$ be the substitutions defined by $\rho(z_i)=P_i$
and $\eta(z_i) = Q_i$ for all $i \in I$. Then $\rho \bis{\,} \eta$.
Since $\bis{\,}$ is a lean congruence for $\CCSP$, one has
$\rec{x| \widehat \RS\,}[\rho] \bis{\,} \rec{x| \widehat \RS\,}[\eta]$ and likewise
$\rec{x| \widehat \RS'\,}[\rho] \bis{\,} \rec{x| \widehat \RS'\,}[\eta]$.
For the same reason one has $\widehat\RS_y[\eta] \bis{\,} \widehat\RS_y[\rho] = \RS_y \bis{\,} \RS'_y 
 \bis{\,} \widehat\RS'_y[\rho] \bis{\,} \widehat\RS'_y[\eta]$ for all $y \in V_\RS$.
Since $\widehat\RS[\eta]$ and $\widehat\RS'[\eta]$ are recursive specifications over CCSP$_\rt$,
$\rec{x| \widehat \RS[\eta]} \bis{\,} \rec{x| \widehat \RS'[\eta]}$.
Hence\mbox{}\hfill
$\rec{x|\RS} = \rec{x|\widehat\RS[\rho]} = \rec{x|\widehat\RS}[\rho]  \bis{\,} \rec{x|\widehat\RS}[\eta]
 = \rec{x|\widehat\RS[\eta]} \bis{\,} \rec{x|\widehat\RS'[\eta]} \bis{\,} \rec{x|\widehat\RS'[\rho]} = \rec{x|\RS'}$.
\end{proof}

\noindent
The following lemmas on the relation between $\theta_X$ and the other operators of $\CCSP$ deal with
strong bisimilarity, but are needed in the congruence proof for strong reactive bisimilarity.
Their proofs can be found in Appendix~\ref{strong proofs}.

\begin{lemma}{theta-Par}
If $P{\ngoesto\tau}$, $\I(P) \cap X \subseteq S$ and $Y\mathbin=X\setminus(S\setminus \I(P))$,
then $\theta_X(P \spar{S} Q) \bis{} \theta_X(P \spar{S}\theta_Y(Q))$.
\end{lemma}

\begin{lemma}{theta-tau}
$\theta_X(\tau_I(P)) \bis{} \theta_X(\tau_I(\theta_{X\cup I}(P)))$.
\end{lemma}

\begin{lemma}{theta-R}
$\theta_X(\Rn(P)) \bis{} \theta_X(\Rn(\theta_{\Rn^{-1}(X)}(P)))$.
\end{lemma}

\section[Strong reactive bisimilarity is a full congruence for CCSP]
        {Strong reactive bisimilarity is a full congruence for $\CCSP$}\label{sec:full congruence}

The forthcoming proofs showing that $\bis{r}$ is a full congruence for $\CCSP$ follow the lines of
Milner \cite{Mi90ccs}, but are more complicated due to the nature of reactive bisimilarity.
A crucial tool is Milner's notion of \emph{bisimilarity up-to}.
The above three lemmas play an essential r\^ole. Even if we would not be interested in the
operators $\theta_L^U$ and $\psi_X$, the proof needs to take the operator $\theta_X$ ($=
\theta_X^X$) along in order to deal with the other operators. This is a consequence of the
occurrence of $\theta_X$ in \df{time-out bisimulation}.

\begin{definition}{upto}
Given a relation ${\sim} \subseteq \IP \times \IP$, a \emph{strong time-out bisimulation up to $\sim$}
is a symmetric relation ${\B} \subseteq \IP \times \IP$, such that, for $P\B Q$,
\begin{itemize}\itemsep 0pt \parsep 0pt
\item if $P \goesto{\alpha} P'$ with $\alpha\mathbin\in A\cup\{\tau\}$, then $\exists Q'$ such that
  $Q\goesto{\alpha} Q'$ and $P' \sim \B \sim Q'$,
\item if $\I(P)\cap (X\cup\{\tau\})=\emptyset$ and $P \goesto{\rt} P'$, then $\exists Q'$ with $Q\goesto{\rt} Q'$
  and $\theta_{X}(P') \sim \B \sim \theta_X(Q')$.
\end{itemize}
Here ${\sim \B \sim}  := \{(R,T) \mid \exists R',T'.~ R \sim R' \B T' \sim T\}$.
\end{definition}

\begin{proposition}{upto}
If $P \B Q$ for some strong time-out bisimulation $\B$ up to $\bis{}$\,, then $P \bis{r} Q$.
\end{proposition}
\begin{proof}
Using the reflexivity of $\bis{}$ it suffices to show that $\bis{}\, \B \,\bis{}$ is a strong time-out bisimulation.
Clearly this relation is symmetric, and that it satisfies the first clause of \df{time-out bisimulation} is
straightforward, using transitivity of $\bis{}$\,.
So assume $P \bis{} R \B T \bis{} Q$, $\I(P)\cap (X\cup\{\tau\})=\emptyset$ and $P \goesto{\rt} P'$.
Then $\I(R)\cap (X\cup\{\tau\})=\emptyset$. By the transfer property of $\bis{}$\,,
there exists an $R'$ with $R\goesto{\rt} R'$ and $P' \bis{} R'$. Since $\bis{}$ is a congruence for
$\theta_X$ it follows that $\theta_{X}(P') \bis{} \theta_X(R')$. By \df{upto}, 
there exists a $T'$ with $T\goesto{\rt} T'$ and $\theta_{X}(R') \bis{}\, \B \,\bis{} \theta_X(T')$.
Again using the transfer property of $\bis{}$\,,
there exists a $Q'$ with $Q\goesto{\rt} Q'$ and $\theta_{X}(T') \bis{} \theta_{X}(Q')$.
Thus, $\theta_{X}(P') \bis{}\, \B \,\bis{} \theta_X(Q')$.
\end{proof}

\begin{theorem}{congruence}
Strong reactive bisimilarity is a lean congruence for $\CCSP$.
In other words, if $\rho,\nu\!:\!\Var\mathbin\rightharpoonup\IT$ are substitutions with
$\rho \rbis{}{r} \nu$, then $E[\rho] \rbis{}{r} E[\nu]$ for any expression $E\in\IT$.
\end{theorem}

\begin{proof}
It suffices to prove this theorem for the special case that $\rho,\nu\!:\!\Var\rightarrow\IP$ are
closed substitutions; the general case then follows by means of composition of substitutions.
Let $\mathord{\B}\subseteq \IP \times \IP$ be the smallest relation satisfying
\begin{itemize} 
\item if $P \bis{r} Q$, then $P \B Q$,
\item if $P \B Q$ and $\alpha\in A\cup\{\tau,\rt\}$, then $\alpha.P \B \alpha.Q$,
\item if $P_1 \B Q_1$ and $P_2 \B Q_2$, then $P_1 + P_2 \B Q_1 + Q_2$,
\item if $P \B Q$, $L\subseteq U \subseteq A$ and $X \subseteq A$,
  then $\theta_L^U(P) \B \theta_L^U(Q)$ and $\psi_X(P) \B \psi_X(Q)$,
\item if $P_1 \B Q_1$, $P_2 \B Q_2$ and $S\subseteq A$, then $P_1\spar{S}P_2 \B Q_1\spar{S}Q_2$,
\item if $P \B Q$ and $I\subseteq A$, then $\tau_I(P) \B \tau_I(Q)$,
\item if $P \B Q$ and $\Rn \subseteq A \times A$, then $\Rn(P) \B \Rn(Q)$,
\item if $\RS$ is a recursive specification with $z \in V_\RS$,
      and $\rho,\nu:\Var\setminus V_\RS \rightarrow\IP$ are substitutions satisfying $\rho(x) \B \nu(x)$ for all
      $x\in\Var\setminus V_\RS$, then $\rec{z|\RS}[\rho] \B \rec{z|\RS}[\nu]$.
\end{itemize}
A straightforward induction on the derivation of $P \B Q$, employing \thm{initials congruence}, yields that\\[5pt]
\mbox{}\hfill if $P \B Q$ then $\I(P)=\I(Q)$, i.e., $P =_\I Q$.\hfill(\verb#@#)\\[5pt]
(For the last case, the assumption that $\rho(x) \B \nu(x)$ for all $x\in\Var\setminus V_\RS$
implies $\rho =_\I \nu$ by induction. Since $=_\I$ is a lean congruence by \thm{initials congruence}, this implies 
$\rec{z|\RS}[\rho] =_\I \rec{z|\RS}[\nu]$.)
\\
A trivial structural induction on $E\in\IT$ shows that\\[5pt]\mbox{}\hfill if $\rho,\nu:\Var \rightarrow\IP$ satisfy
$\rho(x) \mathrel\B \nu(x)$ for all $x\in\Var$, then $E[\rho] \mathbin\B E[\nu]$.\hfill ({\color{red}*})\\[5pt]
For $\RS$ a recursive specification and $\rho:\Var\setminus V_\RS \rightarrow\IP$, let $\rho_\RS: \Var\rightarrow\IP$ be the
closed substitution given by $\rho_\RS(x):= \rec{x|\RS}[\rho]$ if $x\in V_\RS$ and $\rho_\RS(x):=\rho(x)$ otherwise.
Then $\rec{E|\RS}[\rho] = E[\rho_\RS]$ for all $E\in\IT$.\linebreak[3]
Hence an application of ({\color{red}*}) with $\rho_\RS$ and $\nu_\RS$ yields that under the
conditions of the last clause for $\B$ above one even has 
$\rec{E|\RS}[\rho] \mathrel\B \rec{E|\RS}[\nu]$ for all expressions $E\in\IT$.\hfill (\$)

\noindent
It suffices to show that $\B$ is a strong time-out bisimulation up to $\bis{}$\,, because then
$P \rbis{}r Q \Leftrightarrow P \mathrel\B Q$, and ({\color{red}*}) implies that $\B$ is a lean congruence.
Because $\bis{r}$ is symmetric, so is $\B$. So I need to show that
$\B$ satisfies the two clauses of \df{upto}.
\begin{itemize}
\item Let $P \B Q$ and $P \goesto{\alpha} P'$ with $\alpha\mathbin\in A \cup \{\tau\}$.
I have to find a $Q'$ with $Q \goesto{\alpha} Q'$ and $P' \bis{}\, \B \,\bis{} Q'\!$.
In fact, I show that even $P' \B Q'$.
This I will do by structural induction on the proof $\pi$ of $P \goesto{\alpha} P'$ from the rules of
\tab{sos CCSP}. I make a case distinction based on the derivation of $P \B Q$.

\begin{itemize}
\item
Let $P\bis{r} Q$. Using that the relation $\bis{r}$ is a strong time-out bisimulation, there must
be a process $Q'$ such that $Q \goesto{\alpha}Q'$ and $P' \bis{r} Q'$. Hence $P'\BR Q'$.

\item Let $P = \beta.P^\dagger$ and $Q = \beta.Q^\dagger$ with $\beta\in A\cup\{\tau,\rt\}$ and $P^\dagger \B Q^\dagger$.
  Then $\alpha=\beta$ and $P'=P^\dagger$. Take $Q':=Q^\dagger$. Then $Q \goesto{\alpha}Q'$ and $P'\BR Q'$.

\item Let $P = P_1 + P_2$ and $Q = Q_1 + Q_2$ with $P_1 \B Q_1$ and $P_2 \B Q_2$.
  I consider the first rule from \tab{sos CCSP} that could have been responsible for the derivation
  of $P \goesto{\alpha} P'$; the other proceeds symmetrically.
  So suppose that $P_1 \goesto{\alpha} P'$. Then by induction $Q_1 \goesto{\alpha} Q'$
  for some $Q'$ with $P'\BR Q'$. By the same rule from \tab{sos CCSP}, $Q \goesto{\alpha}Q'$.

\item Let $P=\theta_L^U(P^\dagger)$, $Q=\theta_L^U(Q^\dagger)$ and $P^\dagger \B Q^\dagger$.
  First suppose $\alpha\in A$.
  Since $\theta_L^U(P^\dagger)\goesto{\alpha} P'$, it must be that $P^\dagger\goesto{\alpha} P'$ and
  either $\alpha\in U$ or $P^\dagger {\ngoesto\beta}$ for all $\beta\in L\cup\{\tau\}$.
  In the latter case, (\verb#@#) yields $\I(P^\dagger) \mathbin= \I(Q^\dagger)$,
  and thus $Q^\dagger {\ngoesto\beta}$ for all $\beta\mathbin\in L\mathop\cup\{\tau\}$.\linebreak[3]
  By induction there exists a $Q'$ such that $Q^\dagger\goesto{\alpha} Q'$ and $P'\BR Q'$.
  So, in both cases, $Q=\theta_L^U(Q^\dagger)\goesto{\alpha} Q'$.

  Now suppose $\alpha=\tau$.
  Since $\theta_L^U(P^\dagger)\goesto{\tau} P'$ it must be that $P'$ has the form $\theta_L^U(P^\ddagger)$,
  and $P^\dagger\goesto{\tau} P^\ddagger$.
  By induction, there exists a $Q^\ddagger$ such that $Q^\dagger\goesto{\tau} Q^\ddagger$ and $P^\ddagger \BR Q^\ddagger$.
  Now $Q = \theta_L^U(Q^\dagger) \goesto{\tau} \theta_L^U(Q^\ddagger) =: Q'$ and $P' \BR Q'$.

\item Let $P=\psi_X(P^\dagger)$, $Q=\psi_X(Q^\dagger)$ and $P^\dagger \B Q^\dagger$.
  Since $\psi_X(P^\dagger)\goesto{\alpha} P'$, one has $P^\dagger\goesto{\alpha} P'$.
  By induction there exists a $Q'$ with $Q^\dagger\goesto{\alpha} Q'$ and $P'\BR Q'$.
  So $Q=\psi_X(Q^\dagger)\goesto{\alpha} Q'$.

\item Let $P = P_1 \spar{S} P_2$ and $Q = Q_1 \spar{S} Q_2$ with $P_1 \B Q_1$ and $P_2 \B Q_2$.
  I consider the three rules from \tab{sos CCSP} that could have been responsible for the derivation
  of $P \goesto{\alpha} P'$.

First suppose that $\alpha \notin S$, $P_1 \goesto{\alpha} P'_1$ and $P' =  P'_1 \spar{S} P_2$.
By induction, $Q_1 \goesto{\alpha} Q'_1$ for some $Q'_1$ with $P'_1 \BR Q'_1$.
Consequently, $Q_1 \spar{S} Q_2 \goesto{\alpha} Q'_1 \spar{S} Q_2$, and $P' =  P'_1 \spar{S} P_2 \BR  Q'_1 \spar{S} Q_2$.

Next suppose that $\alpha \in S$, $P_1 \goesto{\alpha} P'_1$, $P_2 \goesto{\alpha} P'_2$ and $P' =  P'_1 \spar{S} P'_2$.
By induction, $Q_1 \goesto{\alpha} Q'_1$ for some $Q'_1$ with $P'_1 \BR Q'_1$, and
$Q_2 \goesto{\alpha} Q'_2$ for some $Q'_2$ with $P'_2 \BR Q'_2$.
Consequently, $Q_1 \spar{S} Q_2 \goesto{\alpha} Q'_1 \spar{S} Q'_2$, and $P' =  P'_1 \spar{S} P'_2 \BR  Q'_1 \spar{S} Q'_2$.

The remaining case proceeds symmetrically to the first.

\item Let $P = \tau_I(P^\dagger)$ and $Q= \tau_I(Q^\dagger)$ with $I\subseteq A$ and $P^\dagger \B Q^\dagger$.
Then $P^\dagger\goesto{\beta} P^\ddagger$ for some $P^\ddagger$ with $P'= \tau_I(P^\ddagger)$, and
either $\beta = \alpha \notin I$, or $\beta\in I$ and $\alpha=\tau$.
By induction, $Q^\dagger \goesto{\beta} Q^\ddagger$ for some $Q^\ddagger$ with $P^\ddagger \BR Q^\ddagger$.
Consequently, $Q = \tau_I(Q^\dagger)\goesto{\alpha}\tau_I(Q^\ddagger) =: Q'$ and $P' \BR Q'$.

\item Let $P = \Rn(P^\dagger)$ and $Q= \Rn(Q^\dagger)$ with $\Rn \subseteq A \times A$ and $P^\dagger \B Q^\dagger$.
Then $P^\dagger\goesto{\beta} P^\ddagger$ for some $P^\ddagger$ with $P'= \Rn(P^\ddagger)$, and
either $(\beta,\alpha) \in \Rn$ or $\beta=\alpha = \tau$.
By induction, $Q^\dagger \goesto{\beta} Q^\ddagger$ for some $Q^\ddagger$ with $P^\ddagger \BR Q^\ddagger$.
Consequently, $Q = \Rn(Q^\dagger)\goesto{\alpha}\Rn(Q^\ddagger) =: Q'$ and $P' \BR Q'$.

\item
Let $P\mathbin=\rec{z|\RS}[\rho]\mathbin=\rec{z|\RS[\rho]}$ and
$Q\mathbin=\rec{z|\RS}[\nu]\mathbin=\rec{z|\RS[\nu]}$ where
$\RS$ is a recursive specification with $z \in V_\RS$,
and $\rho,\nu:\Var\setminus V_\RS \rightarrow\IP$ satisfy $\rho(x) \B \nu(x)$ for all $x\in\Var\setminus V_\RS$.
By \tab{sos CCSP} the transition $\rec{\RS_z|\RS[\rho]} \goesto{\alpha} P'$ is provable by means of a strict subproof of the
proof $\pi$ of $\rec{z|\RS}[\rho]\mathbin{\goesto{\alpha}}P'$.
By (\$) above one has $\rec{\RS_z|\RS[\rho]} \B \rec{\RS_z|\RS[\nu]}$.
So by induction there is a $Q'$ such that $\rec{\RS_z|\RS[\nu]} \goesto{a} Q'$ and $P' \BR Q'$.
By \tab{sos CCSP}, $Q = \rec{z|\RS[\nu]} \goesto{\alpha} Q'$.
\end{itemize}

\item Let $P \B Q$, $\I(P)\cap (X\cup\{\tau\})=\emptyset$ and $P \goesto{\rt} P'$.
I have to find a $Q'$ such that $Q\goesto{\rt} Q'$ and $\theta_{X}(P') \bis{}\, \B \,\bis{} \theta_X(Q')$.
This I will do by structural induction on the proof $\pi$ of $P \goesto{\rt} P'$ from the rules of
\tab{sos CCSP}. I make a case distinction based on the derivation of $P \B Q$.

\begin{itemize}
\item
Let $P\bis{r} Q$. Using that the relation $\bis{r}$ is a strong time-out bisimulation, there must
be a process $Q'$ such that $Q \goesto{\rt}Q'$ and $\theta_{X}(P') \bis{r} \theta_X(Q')$.
Thus $\theta_{X}(P') \bis{}\, \BRB \theta_X(Q')$.

\item Let $P = \beta.P^\dagger$ and $Q = \beta.Q^\dagger$ with $\beta\in A\cup\{\tau,\rt\}$ and $P^\dagger \B Q^\dagger$.
  Then $\beta =\rt$ and $P'=P^\dagger$. Take $Q':=Q^\dagger$. Then $Q \goesto{\rt}Q'$ and $P'\B Q'$.
  Thus $\theta_{X}(P') \bis{}\, \BRB \theta_X(Q')$.

\item Let $P = P_1 + P_2$ and $Q = Q_1 + Q_2$ with $P_1 \B Q_1$ and $P_2 \B Q_2$.
  I consider the first rule from \tab{sos CCSP} that could have been responsible for the derivation
  of $P \goesto{\rt} P'$; the other proceeds symmetrically.
  So suppose that $P_1 \goesto{\rt} P'$. Since $\I(P_1)\cap (X\cup\{\tau\})\subseteq \I(P)\cap (X\cup\{\tau\})=\emptyset$,
  by induction $Q_1 \goesto{\rt} Q'$ for some $Q'$ with $P' \bis{}\, \BRB Q'$. Hence $Q \goesto{\rt}Q'$.

\item Let $P=\theta_L^U(P^\dagger)$, $Q=\theta_L^U(Q^\dagger)$ and $P^\dagger \B Q^\dagger$.
  Since $\theta_L^U(P^\dagger)\goesto{\rt} P'$ it must be that $P^\dagger\goesto{\rt} P'$ and
  $P^\dagger {\ngoesto\beta}$ for all $\beta\in L\cup\{\tau\}$.
  Consequently, $P\goesto{\alpha} P^\ddagger$ iff $P^\dagger\goesto{\alpha} P^\ddagger$, for all
  $\alpha\in A \cup\{\rt\}$. So $\I(P^\dagger)\cap (X\cup\{\tau\}) = \emptyset$.
  By induction, $Q^\dagger\goesto{\rt} Q'$ for some $Q'$ with $\theta_{X}(P') \bis{\,}\, \BRB \theta_X(Q')$.
  By (\verb#@#), $Q^\dagger {\ngoesto\beta}$ for all $\beta\in L\cup\{\tau\}$.
  Hence $Q = \theta_L^U(Q^\dagger)\goesto{\rt} Q'$.

\item Let $P=\psi_Y(P^\dagger)$, $Q=\psi_Y(Q^\dagger)$ and $P^\dagger \B Q^\dagger$.
  Since $\psi_Y(P^\dagger)\goesto{\rt} P'$ one has $P^\dagger\goesto{\rt} P^\ddagger$ for
  some $P^\ddagger$ with $P'=\theta_Y(P^\ddagger)$, and $P^\dagger {\ngoesto\beta}$ for all $\beta\in Y\cup\{\tau\}$,
  i.e., $\I(P^\dagger) \cap (Y\cup\{\tau\}) = \emptyset$.
  By induction, $Q^\dagger\goesto{\rt} Q^\ddagger$ for a $Q^\dagger$ with
  $\theta_{Y}(P^\ddagger) \bis{\,}\, \BRB \theta_{Y}(Q^\ddagger)$.
  By (\verb#@#), $\I(P^\dagger)=\I(Q^\dagger)$, so $Q^\dagger {\ngoesto\beta}$ for all $\beta\in Y\cup\{\tau\}$.
  Let $Q' := \theta_Y(Q^\ddagger)$, so that $Q = \psi_Y(Q^\dagger)\goesto{\rt} \theta_Y(Q^\ddagger) =Q'$.
  From $\theta_{Y}(P^\ddagger) \bis{}\, P'' \BR Q'' \mathrel{\color{blue}\bis{\,}} \theta_{Y}(Q^\ddagger)$ one obtains
  $$\theta_X(\theta_{Y}(P^\ddagger))
  \bis{\,} \theta_X(P'') \BR \theta_X(Q'') \mathrel{\color{blue}\bis{\,}} \theta_X(\theta_Y(Q^\ddagger)),$$ 
  using that $\mathrel{\color{blue}\bis{}}$ is a congruence for $\theta_X$ ($= \theta_X^X$).
  Thus $\theta_X(P') \bis{}\, \BRB \theta_X(Q')$.

\item Let $P = P_1 \spar{S} P_2$ and $Q = Q_1 \spar{S} Q_2$ with $P_1 \B Q_1$ and $P_2 \B Q_2$.
  I consider the last rule from \tab{sos CCSP} that could have been responsible for the derivation of $P \goesto{\rt} P'$.
  The other proceeds symmetrically. So suppose that $P_2 \goesto{\rt} P'_2$ and $P' =  P_1 \spar{S} P'_2$.
  Let $Y:=X\setminus(S\setminus \I(P_1))=(X\setminus S) \cup (X \cap S \cap \I(P_1))$.
  Then $\I(P_2)\cap (Y\cup\{\tau\})=\emptyset$.
  By induction, $Q_2 \goesto{\rt} Q'_2$ for some $Q'_2$ with $\theta_{Y}(P'_2) \bis{}\, \BRB \theta_Y(Q'_2)$.
  Let $Q':= Q_1 \spar{S} Q'_2$, so that $Q = Q_1 \spar{S} Q_2 \goesto{\rt} Q_1 \spar{S} Q'_2 = Q'$.
  From $\theta_{Y}(P'_2) \bis{}\, P''_2 \BR Q''_2 \,\mathrel{\color{blue}\bis{}} \theta_Y(Q'_2)$ and $P_1 \B Q_1$ one obtains
  $P_1 \spar{S} \theta_{Y}(P'_2) \bis{} P_1 \spar{S}P''_2 \BR Q_1 \spar{S}Q''_2 \mathrel{\color{blue}\bis{}} Q_1 \spar{S} \theta_{Y}(Q'_2)$,
  using that $\mathrel{\color{blue}\bis{}}$ is a congruence for $\spar{S}$.
  Therefore, since $\color{blue}\bis{}$ is also a congruence for $\theta_X$ ($= \theta_X^X$),
  $$\theta_X(P_1 \spar{S} \theta_{Y}(P'_2)) \bis{\,} \theta_X(P_1 \spar{S}P''_2)
  \BR \theta_X(Q_1 \spar{S}Q''_2) \mathrel{\color{blue}\bis{}}\, \theta_X(Q_1 \spar{S} \theta_{Y}(Q'_2)).$$
  Since $\I(P_1\spar{S}P_2)\cap (X\cup\{\tau\})=\emptyset$,
  one has $P_1 {\ngoesto\tau}$ and $\I(P_1)\cap X \subseteq S$.
  Moreover, since $P_1 \B Q_1$, one has $\I(P_1) = \I(Q_1)$.
  Hence $\theta_X(P') = \theta_X(P_1 \spar{S} P'_2)\bis{} \theta_X(P_1 \spar{S} \theta_Y(P'_2)) 
  \bis{}\, \BRB \theta_X(Q_1 \spar{S} \theta_Y(Q'_2)) \bis{} \theta_X(Q_1 \spar{S} Q'_2)  =
  \theta_X(Q')$ by \lem{theta-Par}.

\item Let $P = \tau_I(P^\dagger)$ and $Q= \tau_I(Q^\dagger)$ with $I\subseteq A$ and $P^\dagger \B Q^\dagger$.
Then $P^\dagger\goesto{\rt} P^\ddagger$ for some $P^\ddagger$ with $P'= \tau_I(P^\ddagger)$.
Moreover, $\I(P^\dagger) \cap (X \cup I \cup\{\tau\}) = \emptyset$.
By induction, $Q^\dagger \goesto{\rt} Q^\ddagger$ for some $Q^\ddagger$ with
$\theta_{X\cup I}(P^\ddagger) \bis{}\, \BRB \theta_{X\cup I}(Q^\ddagger)$.
Let $Q' := \tau_I(Q^\ddagger)$, so that $Q = \tau_I(Q^\dagger)\goesto{\rt}\tau_I(Q^\ddagger) = Q'$.
  From $\theta_{X\cup I}(P^\ddagger) \bis{}\, P'' \BR Q'' \mathrel{\color{blue}\bis{\,}} \theta_{X\cup I}(Q^\ddagger)$ one obtains
  $$\theta_X(\tau_I(P^\ddagger)) \bis{\,} \theta_X(\tau_I(\theta_{X\cup I}(P^\ddagger))) \bis{} \theta_X(\tau_I(P'')) \BR
  \theta_X(\tau_I(Q'')) \mathrel{\color{blue}\bis{\,}} \dots
  \bis{} \theta_X(\tau_I(Q^\ddagger)),$$ 
  using \lem{theta-tau} and that $\mathrel{\color{blue}\bis{}}$ is a congruence for $\tau_I$ and $\theta_X$.
  Thus $\theta_X(P') \bis{}\, \BRB \theta_X(Q')$.

\item Let $P = \Rn(P^\dagger)$ and $Q= \Rn(Q^\dagger)$ with $\Rn\subseteq A \times A$ and $P^\dagger \B Q^\dagger$.
Then $P^\dagger\goesto{\rt} P^\ddagger$ for some $P^\ddagger$ with $P'= \Rn(P^\ddagger)$.
Moreover, $\I(P^\dagger) \cap (\Rn^{-1}(X) \cup\{\tau\}) = \emptyset$.
By induction, $Q^\dagger \goesto{\rt} Q^\ddagger$ for some $Q^\ddagger$ with
$\theta_{\Rn^{-1}(X)}(P^\ddagger) \bis{}\, \BRB \theta_{\Rn^{-1}(X)}(Q^\ddagger)$.
Let $Q' := \Rn(Q^\ddagger)$, so that $Q \mathbin= \Rn(Q^\dagger)\goesto{\rt}\Rn(Q^\ddagger) \mathbin= Q'$.
  From $\theta_{\Rn^{\mbox{\scriptsize-}1}(X)}(P^\ddagger) \bis{} P'' \BR Q'' \mathrel{\color{blue}\bis{\,}} \theta_{\Rn^{\mbox{\scriptsize-}1}(X)}(Q^\ddagger)$ one obtains
  $$\theta_X(\Rn(P^\ddagger)) \bis{} \theta_X(\Rn(\theta_{\Rn^{\mbox{\scriptsize-}1}(X)}(P^\ddagger))) \bis{} \theta_X(\Rn(P'')) \BR
  \theta_X(\Rn(Q'')) \mathrel{\color{blue}\bis{}} \dots \bis{} \theta_X(\Rn(Q^\ddagger)),$$ 
  using \lem{theta-R} and that $\mathrel{\color{blue}\bis{}}$ is a congruence for $\Rn$ and $\theta_X$.
  Thus $\theta_X(P') \bis{}\, \BRB\theta_X(Q')$.

\item
Let $P\mathbin=\rec{z|\RS}[\rho]\mathbin=\rec{z|\RS[\rho]}$ and
$Q\mathbin=\rec{z|\RS}[\nu]\mathbin=\rec{z|\RS[\nu]}$ where
$\RS$ is a recursive specification with $z \in V_\RS$,
and $\rho,\nu:\Var\setminus V_\RS \rightarrow\IP$ satisfy $\rho(x) \B \nu(x)$ for all $x\in\Var\setminus V_\RS$.
By \tab{sos CCSP} the transition $\rec{\RS_z|\RS[\rho]} \goesto{\rt} P'$ is provable by means of a strict subproof of the
proof $\pi$ of $\rec{z|\RS}[\rho]\mathbin{\goesto{\rt}}P'$.
The rule for recursion in \tab{sos CCSP} also implies that
$\I(\rec{z|\RS}[\rho]) = \I(\rec{\RS_z|\RS}[\rho])$.
Therefore, $\I(\rec{\RS_z|\RS}[\rho])\cap (X\cup\{\tau\}) = \emptyset$.
By (\$) above one has $\rec{\RS_z|\RS[\rho]} \B \rec{\RS_z|\RS[\nu]}$.\linebreak[3]
So by induction there is a $Q'$ such that $\rec{\RS_z|\RS[\nu]} \goesto{\rt} Q'$ and $\theta_X(P') \bis{}\, \BRB \theta_X(Q')$.
By \tab{sos CCSP}, $Q = \rec{z|\RS[\nu]} \goesto{\rt} Q'$.
\qed
\end{itemize}
\end{itemize}
\end{proof}

\begin{proposition}{upto r}
If $P \B Q$ for some strong time-out bisimulation $\B$ up to $\bis{r}$\,, then $P \bis{r} Q$.
\end{proposition}
\begin{proof}
Exactly as the proof of \pr{upto}, now using that $\bis{r}$ is a congruence for $\theta_X$.
\end{proof}

\begin{theorem}{full congruence}
Strong reactive bisimilarity is a full congruence for $\CCSP$.
\end{theorem}

\begin{proof}
Let $\mathord{\B}\subseteq \IP \times \IP$ be the smallest relation satisfying
\begin{itemize} 
\item if $\RS$ and $\RS'$ are recursive specifications with $x \in V_\RS = V_{\RS'}$ and
      $\rec{x|\RS},\rec{x|\RS'}\in\IP$, such that $\RS_y \bis{\,} \RS'_y$ for all $y\in V_\RS$,
      then $\rec{x|\RS} \B \rec{x|\RS'}$,
\end{itemize}
in addition to the eight or nine clauses listed in the proof of \thm{congruence}.
Again, a straightforward induction on the derivation of $P \B Q$, employing \thm{initials congruence}, yields that\\[5pt]
\mbox{}\hfill if $P \B Q$ then $\I(P)=\I(Q)$, i.e., $P =_\I Q$.\hfill(\verb#@#)\\[5pt]
(For the new case, the assumption that $\RS_y \bis{\,} \RS'_y$ for all $y\in V_\RS$
implies  $\RS_y =_\I \RS'_y$ for all $y\in V_\RS$. So by \thm{initials congruence},  $\rec{x|\RS} =_\I \rec{x|\RS'}$.)
A trivial structural induction on $E\in\IT$ shows again that\\[5pt]\mbox{}\hfill if $\rho,\nu:\Var \rightarrow\IP$ satisfy
$\rho(x) \mathrel\B \nu(x)$ for all $x\in\Var$, then $E[\rho] \mathbin\B E[\nu]$.\hfill ({\color{red}*})\\[5pt]
This again implies that in the last clause for $\B$ one even has
$\rec{E|\RS}[\rho] \B \rec{E|\RS'}[\nu]$ for all $E\in\IT$,\hfill (\$)\\
and likewise, in the new clause, $\rec{E|\RS} \B \rec{E|\RS'}$ for all $E\in\IT$ with variables from $V_\RS$.\hfill (\#)

It suffices to show that $\B$ is a strong time-out bisimulation up to $\bis{r}$\,, because then
${\B}\subseteq{\bis{r}}\,$ with \pr{upto r}, and the new clause for $\B$ implies (\ref{comp-recursion-closed}).
By construction $\B$ is symmetric.
\begin{itemize}
\item Let $P \B Q$ and $P \goesto{\alpha} P'$ with $\alpha\mathbin\in A \cup \{\tau\}$.
I have to find a $Q'$ with $Q \goesto{\alpha} Q'$ and $P' \bis{r}\, \B \,\bis{r} Q'\!$.
In fact, I show that even $P' \B \,\bis{r} Q'$.
This I will do by structural induction on the proof $\pi$ of $P \goesto{\alpha} P'$ from the rules of
\tab{sos CCSP}. I make a case distinction based on the derivation of $P \B Q$.

\begin{itemize}
\item
Let $P=\rec{x|\RS}\in\IP$ and $Q=\rec{x|\RS'}\in \IP$ where
$\RS$ and $\RS'$ are recursive specifications with $x \in V_\RS = V_{\RS'}$,
such that $\RS_y \bis{\,} \RS'_y$ for all $y\in V_\RS$, meaning that
for all $y\in W$ and $\sigma:V_\RS\rightarrow \IP$ one has
$\RS_y[\sigma] \bis{r} \RS'_y[\sigma]$.
\\
By \tab{sos CCSP} the transition $\rec{\RS_x|\RS} \goesto{\alpha} P'$ is provable by means of a strict subproof of $\pi$.
By (\#) above one has $\rec{\RS_x|\RS} \B \rec{\RS_x|\RS'}$.
So by induction there is an $R'\in\IP$ such that $\rec{\RS_x|\RS'} \goesto{\alpha} R'$ and $P' \B\, \bis{r} R'$.
Since $\rec{\_\!\_\, | \RS'}$ is the application of a substitution of the form $\sigma:V_{\RS'} \rightarrow\IP$,
one has $\rec{\RS_x|\RS'} \bis{r} \rec{\RS'_x|\RS'}$.
Hence there is a $Q'$ with $P\vdash \rec{\RS'_x|\RS'} \mathbin{\goesto{\alpha}} Q'$ and $R' \bis{r} Q'$.
So $P' \B \,\bis{r} Q'$.
By \tab{sos CCSP}, $Q \mathbin= \rec{x|\RS'} \mathbin{\goesto{\alpha}} Q'$.
\item
The remaining nine cases proceed just as in the proof of \thm{congruence}, but with $\B\,\bis{r}$ substituted
for the blue occurrences of $\B$. In the case for $\theta_L^U$ with $\alpha=\tau$,
I conclude from $P^\ddagger \B\,\bis{r} Q^\ddagger$ that $\theta_L^U(P^\ddagger) \B\,\bis{r} \theta_L^U(Q^\ddagger)$.
Besides applying the definition of $\B$, this also involves the application of \thm{congruence} that
$\bis{r}$ is already known to be a congruence for $\theta_L^U$. The same reasoning applies in the
cases for $\spar{S}$, $\tau_I$ and $\Rn$.
\end{itemize}

\item Let $P \B Q$, $\I(P)\cap (X\cup\{\tau\})=\emptyset$ and $P \goesto{\rt} P'$.
I will find a $Q'$ such that $Q\goesto{\rt} Q'$ and $\theta_{X}(P') \bis{\,}\, \B \,\bis{r} \theta_X(Q')$.
This I will do by structural induction on the proof $\pi$ of $P \goesto{\rt} P'$ from the rules of
\tab{sos CCSP}. I make a case distinction based on the derivation of $P \B Q$.

\begin{itemize}
\item
Let $P=\rec{x|\RS}\in\IP$ and $Q=\rec{x|\RS'}\in \IP$ where
$\RS$ and $\RS'$ are recursive specifications with $x \in V_\RS = V_{\RS'}$,
such that 
for all $y\in W$ and $\sigma:V_\RS\rightarrow \IP$ one has
$\RS_y[\sigma] \bis{r} \RS'_y[\sigma]$.
By \tab{sos CCSP} the transition \plat{$\rec{\RS_x|\RS} \goesto{\rt} P'$} is provable by means of a strict subproof of the
proof $\pi$ of $\rec{x|\RS}\goesto{\rt}P'$.
The rule for recursion in \tab{sos CCSP} also implies that
$\I(\rec{x|\RS}) = \I(\rec{\RS_x|\RS})$.
Therefore, $\I(\rec{\RS_x|\RS})\cap (X\cup\{\tau\}) = \emptyset$.
By (\#) above one has $\rec{\RS_x|\RS} \B \rec{\RS_x|\RS'}$.
So by induction there is an $R'\in\IP$ such that $\rec{\RS_x|\RS'} \goesto{\rt} R'$ and $\theta_X(P') \bis{\,}\,\B\, \bis{r} \theta_X(R')$.
Since $\rec{\_\!\_\, | \RS'}$ is the application of a substitution of the form $\sigma:V_{\RS'} \rightarrow\IP$,
$\rec{\RS_x|\RS'} \bis{r} \rec{\RS'_x|\RS'}$.
Using (\verb#@#), $\I(\rec{\RS_x|\RS'})\cap (X\cup\{\tau\})=\emptyset$.
Hence $\exists Q'$ with $P\vdash \rec{\RS'_x|\RS'} \mathbin{\goesto{\rt}} Q'$ and $R' \bis{r} Q'$,
and thus $\theta_X(R') \bis{r} \theta_X(Q')$, using \thm{congruence}. So $\theta_X(P') \bis{\,}\, \B \,\bis{r} \theta_X(Q')$.
By \tab{sos CCSP}, $Q \mathbin= \rec{x|\RS'} \mathbin{\goesto{\rt}} Q'$.
\item
The remaining eight cases proceed just as in the proof of \thm{congruence}, but with $\B\,\bis{r}$ substituted
for the blue occurrences of $\B\,\bis{\,}$.
\qed
\end{itemize}
\end{itemize}
\end{proof}

\section{The Recursive Specification Principle}\label{sec:RSP}

For $W \subseteq \Var$ a set of variables, a $W$-tuple of expressions is a function $\vec
E\in\IT^{W}$.  It has a component $\vec{E}(x)$ for each variable $x\in W$. Note that a $W$-tuple of
expressions is nothing else than a substitution. Let $\textit{id}_W$ be the identity function, given
by $\textit{id}_W(x)\mathbin=x$ for all $x\mathbin\in W$.
If $G\in\IT$ and $\vec{E}\in\IT^W$ then $G[\vec{E}]$ denotes the result of simultaneous substitution of $\vec{E}(x)$ for
$x$ in $G$, for all $x\in W$. Likewise, if $\vec{G}\in\IT^{V}$ and $\vec{E}\in\IT^W$ then
$\vec{G}[\vec{E}]\in\IT^V$ denotes the $V$-tuple with components $G(y)[\vec{E}]$ for $y\in V$.
Henceforth, I regard a recursive specification $\RS$ as a $V_\RS$-tuple with components
$\RS(y)=\RS_y$ for $y\in V_\RS$.
If $\vec{E}\mathbin\in\IT^W$ and $\RS\mathbin\in\IT^V$, then $\rec{\vec{E}|\RS}\mathbin\in \IT^W$
is the $W$-tuple with components $\rec{\vec{E}(x)|\RS}\mathbin\in \IT^W$ for $x\mathbin\in W$.

For $\RS$ a recursive specification and $\vec E\mathbin\in\IT^{V_\RS}\!$ a $V_\RS$-tuple of
expressions, $\vec{E} \bis{r} \RS[\vec{E}]$ states that $\vec E$ is a \emph{solution} of $\RS$, up
to strong reactive bisimilarity. The tuple $\rec{\textit{id}_{V_\RS}|\RS} \in \IT^{V_\RS}$ is
called the \emph{default solution}.

In \cite{BW90,Fok00} two requirements occur for process algebras with recursion.
The \emph{recursive definition principle} (RDP) says that each recursive specification must have a solution,
and the \emph{recursive specification principle} (RSP) says that guarded recursive specifications
have at most one solution. When dealing with process algebras where the meaning of a closed expression is
a semantic equivalence class of processes, these principles become requirements on the semantic equivalence employed.

\begin{proposition}{RDP}
Let $\RS$ be a recursive specification, and $x\in V_\RS$. Then $\rec{x|\RS} \bis{r} \rec{\RS_x|\RS}$.
\end{proposition}
\begin{proof}
Let $\sigma:\Var\rightarrow\IP$ be a closed substitution. I have to show that
$\rec{x|\RS}[\sigma] \bis{r} \rec{\RS_x|\RS}[\sigma]$.
Equivalently I may show this for $\sigma\!:\!\Var{\setminus} V_\RS\rightarrow\IP\!$.
Now $\rec{x|\RS}[\sigma] = \rec{x|\RS[\sigma]} \mathbin\in\IP$
and $\rec{\RS_x|\RS}[\sigma] = \rec{\RS_x[\sigma]|\RS[\sigma]}\linebreak\in\IP$.
Consequently, it suffices to prove the proposition under the assumption that $\rec{x|\RS}, \rec{\RS_x|\RS} \in\IP$.
This follows immediately from the rule for recursion in \tab{sos CCSP} and \df{time-out bisimulation}.
\end{proof}
\pr{RDP} says that the recursive definition principle holds for strong reactive bisimulation semantics.
The ``default solution'' of a recursive specification is in fact a solution.
Note that the conclusion of \pr{RDP} can be restated as $\rec{\textit{id}_{V_\RS}|\RS} \bis{r} \rec{\RS|\RS}$,
and that $\RS[\rec{\textit{id}_{V_\RS}|\RS}] = \rec{\RS|\RS}$.

The following theorem establishes the recursive specification principle for strong reactive bisimulation semantics.
Some aspects of the proof that are independent of the notion of bisimilarity employed are delegated
to the following two lemmas.

\begin{lemma}{8}
Let $H\in\IT$ be guarded and have free variables from $W\subseteq \Var$ only, and let $\vec{P},\vec{Q}\in\IP^W$.
Then $\I(H[\vec{P}]) = \I(H[\vec{Q}])$.
\end{lemma}
\begin{proof}
In Appendix \ref{initials congruence}.
\end{proof}

\begin{lemma}{9}
Let $H\in\IT$ be guarded and have free variables from $W\subseteq \Var$ only, and let $\vec{P},\vec{Q}\in\IP^W$.
If $H[\vec{P}] \goesto\alpha R'$ with $\alpha\in Act$, then $R'$ has the form $H'[\vec{P}]$ for some term
$H'\in\IT$ with free variables in $W$ only. Moreover $H[\vec{Q}] \goesto\alpha H'[\vec{Q}]$.
\end{lemma}
\begin{proof}
By induction on the derivation of $H[\vec{P}] \goesto\alpha R'$, making a case distinction on the shape of $H$.

Let $H=\alpha.G$, so that $H[\vec{P}] = \alpha.G[\vec{P}]$.
Then $R' = G[\vec{P}]$ and $H[\vec{Q}] \goesto\alpha G[\vec{Q}]$.

The case $H=0$ cannot occur. Nor can the case $H=x\in \Var$, as $H$ is guarded.

Let $H = H_1 \spar{S} H_2$, so that $H[\vec{P}] = H_1[\vec{P}]\spar{S} H_2[\vec{P}]$.
Note that $H_1$ and $H_2$ are guarded and have free variables in $W$ only.
One possibility is that $a\notin S$, $H_1[\vec{P}]\goesto\alpha R_1$ and $R'= R_1 \spar{S} H_2[\vec{P}]$.
By induction, $R'_1$ has the form $H'_1[\vec{P}]$ for some term
$H'_1\in\IT$ with free variables in $W$ only. Moreover, $H_1[\vec{Q}] \goesto\alpha H'_1[\vec{Q}]$.
Thus $R' = (H'_1 \spar{S} H_2)[\vec{P}]$, and $H':= H'_1 \spar{S} H_2$ has free variables in $W$ only.
Moreover, $H[\vec{Q}] =  H_1[\vec{Q}]\spar{S} H_2[\vec{Q}] \goesto\alpha  H'_1[\vec{Q}]\spar{S} H_2[\vec{Q}] = H'[\vec{Q}]$.

The other two cases for $\spar{S}$, and the cases for the operators $+$, $\tau_I$ and $\Rn$, are equally trivial.

Let $H= \theta_L^U(H^\dagger)$, so that $H[\vec{P}] = \theta_L^U(H^\dagger[\vec{P}])$.
Note that $H^\dagger$ is guarded and has free variables in $W$ only.
The case $\alpha = \tau$ is again trivial, so assume $\alpha\neq\tau$.
Then \plat{$H^\dagger[\vec{P}] \goesto\alpha R'$} and either $\alpha\in X$ or
$H^\dagger[\vec{P}] {\ngoesto\beta}$ for all $\beta\in L\cup\{\tau\}$.
By induction, $R'$ has the form $H'[\vec{P}]$ for some term
$H'\in\IT$ with free variables in $W$ only. Moreover, $H^\dagger[\vec{Q}] \goesto\alpha H'[\vec{Q}]$.
Since $\I(H^\dagger[\vec{P}]) = \I(H^\dagger[\vec{Q}])$ by \lem{8}, either $\alpha\in X$ or
$H^\dagger[\vec{Q}] {\ngoesto\beta}$ for all $\beta\in L\cup\{\tau\}$.
Consequently, $H[\vec{Q}] = \theta_L^U(H^\dagger[\vec{Q}])\goesto\alpha H'[\vec{Q}]$.

Let $H= \psi_X(H^\dagger)$, so that $H[\vec{P}] = \psi_X(H^\dagger[\vec{P}])$.
Note that $H^\dagger$ is guarded and has free variables in $W$ only.
The case $\alpha \in A\cup\{\tau\}$ is trivial, so assume $\alpha=\rt$.
Then \plat{$H^\dagger[\vec{P}] \goesto\rt R^\dagger$} for some $R^\dagger$ such that $R'=\theta_X(R^\dagger)$.
Moreover, $H^\dagger[\vec{P}] {\ngoesto\beta}$ for all $\beta\in X\cup\{\tau\}$.
By induction, $R^\dagger$ has the form $H'[\vec{P}]$ for some term
$H'\in\IT$ with free variables in $W$ only. Moreover, $H^\dagger[\vec{Q}] \goesto\rt H'[\vec{Q}]$.
Since $\I(H^\dagger[\vec{P}]) = \I(H^\dagger[\vec{Q}])$ by \lem{8},
$H^\dagger[\vec{Q}] {\ngoesto\beta}$ for all $\beta\in X\cup\{\tau\}$.
Consequently, $H[\vec{Q}] = \psi_X(H^\dagger[\vec{Q}])\goesto\rt \theta_X(H'[\vec{Q}])$.

Finally, let $H = \rec{x|\RS}$, so that $H[\vec{P}] = \rec{x|\RS[\vec{P}^\dagger]}$, where $\vec{P}^\dagger$ is the
$W {\setminus} V_\RS$-tuple that is left of $\vec{P}$ after deleting the $y$-components, for $y\in V_\RS$.
The transition $\rec{\RS_x[\vec{P}^\dagger]|\RS[\vec{P}^\dagger]} \goesto\alpha R'$ is derivable through a
subderivation of the one for $\rec{x|\RS[\vec{P}^\dagger]}\goesto\alpha R'$.
Moreover, $\rec{\RS_x[\vec{P}^\dagger]|\RS[\vec{P}^\dagger]} = \rec{\RS_x|\RS}[\vec{P}]$.
So by induction, $R'$ has the form $H'[\vec{P}]$ for some term $H'\mathbin\in\IT$ with free variables in $W$
only, and $\rec{\RS_x|\RS}[\vec{Q}] \goesto\alpha H'[\vec{Q}]$. 
Since $\rec{\RS_x|\RS}[\vec{Q}] = \rec{\RS_x[\vec{Q}^\dagger]|\RS[\vec{Q}^\dagger]}$, it follows that 
$H[\vec{Q}] = \rec{x|\RS}[\vec{Q}]= \rec{x|\RS[\vec{Q}^\dagger]}\goesto\alpha H'[\vec{Q}]$. 
\end{proof}
 
\begin{theorem}{RSP}
Let $\RS$ be a guarded recursive specification.
If $\vec{E} \bis{r} \RS[\vec{E}]$ and $\vec{F} \bis{r} \RS[\vec{F}]$ with $\vec E, \vec F\mathbin\in\IT^{V_\RS}$, then 
$\vec{E} \bis{r} \vec{F}$.
\end{theorem}

\begin{proof}
It suffices to prove \thm{RSP} under the assumptions that $\vec E, \vec F\mathbin\in\IP^{V_\RS}$ and
only the variables from $V_\RS$ occur free in the expressions $\RS_x$ for $x \in V_\RS$.
For in the general case I have to establish that $\vec{E}[\sigma] \bis{r} \vec{F}[\sigma]$
for an arbitrary closed substitution $\sigma\!:\!\Var\rightarrow\IP$.
Let $\hat\sigma:\Var{\setminus}V_\RS\rightarrow\IP$
be given by $\hat\sigma(x)=\sigma(x)$ for all $x\in \Var{\setminus}V_\RS$.
Then $\vec{E} \bis{r} \RS[\vec{E}]$ implies $\vec{E}[\sigma] \bis{r} \RS[\vec{E}][\sigma] = \RS[\hat\sigma][\vec{E}[\sigma]]$.
Hence, I merely have to prove the theorem with $\vec{E}[\sigma]$, $\vec{F}[\sigma]$and
$\RS[\hat\sigma]$ in place of $\vec{E}$, $\vec{F}$ and $\RS$.

It also suffices to prove \thm{RSP} under the assumption that $\RS$ is a manifestly guarded recursive specification.
Namely, for a general guarded recursive specification $\RS$, let $\RS'$ be the manifestly guarded 
specification into which $\RS$ can be converted. Then $\vec{E} \bis{r} \RS[\vec{E}]$ implies $\vec{E} \bis{r} \RS'[\vec{E}]$
by \thm{congruence}.

So let $\RS$ be manifestly guarded with free variables from $V_\RS$ only,
and let $\vec{P},\vec{Q}\in \IP^{V_\RS}$ be two of its solutions, that
is, $\vec{P} \bis{r} \RS[\vec{P}]$ and $\vec{Q} \bis{r} \RS[\vec{Q}]$.
I will show that the symmetric closure of $${\B} := \{H [\RS[\vec{P}]], H[\RS[\vec{Q}]] \mid
H\mathbin\in\IT \mbox{~has free variables in $V_\RS$ only}\}$$ is
a strong time-out bisimulation up to $\bis{r}\,$.
Once I have that, taking $H:= x\in V_\RS$ yields $\RS_x[\vec{P}] \bis{r} \RS_x[\vec{Q}]$ by \pr{upto r},
and thus $P(x) \bis{r} \RS_x[\vec{P}] \bis{r} \RS_x[\vec{Q}] \bis{r} Q(x)$ for all $x\in V_\RS$. So $\vec{P} \bis{r} \vec{Q}$.

\begin{itemize}
\item Let $R \B T$ and $R \goesto{\alpha} R'$ with $\alpha\mathbin\in A \cup \{\tau\}$.
I have to find a $T'$ with $T \goesto{\alpha} T'$ and $P' \bis{r}\, \B \,\bis{r} Q'\!$.
Assume that $R \mathbin= H[\RS[\vec{P}]]$ and $T\mathbin=H[\RS[\vec{Q}]]$---the case
that $R \mathbin= H[\RS[\vec{Q}]]$ will follow by symmetry.

Note that $H[\RS[\vec{P}]]$ can also be written as $H[\RS][\vec{P}]$.
Since the expressions $\RS_x$ for $x\in V_\RS$ have free variables from $V_\RS$ only, so does $H[\RS]$.
Moreover, since $\RS$ is manifestly guarded, the expression $H[\RS]$ must be guarded.
By \lem{9}, $R'$ must have the form $H'[\vec{P}]$, where $H'\in\IT$ has free variables in $V_\RS$ only.
Moreover, $T = H [\RS[\vec{Q}]] = H[\RS][\vec{Q}] \goesto\alpha H'[\vec{Q}] =: T'$. 
Furthermore, by \thm{congruence}, $H'[\vec{P}] \bis{r} H'[\RS[\vec{P}]]$ and $H'[\RS[\vec{Q}]] \bis{r} H'[\vec{Q}]$.
Thus, $R'\mathbin= H'[\vec{P}]\bis{r}\,\B\,\bis{r}H'[\vec{Q}] \mathbin= T'\!$.

\item Let $R \B T$, $\I(R)\cap (X\cup\{\tau\})=\emptyset$ and $R \goesto{\rt} R'$.
I have to find a $T'$ such that $T\goesto{\rt} T'$ and $\theta_{X}(R') \bis{\,}\, \B \,\bis{r} \theta_X(T')$.
The proof for this case proceeds exactly as that of the previous case, up to the last sentence; the condition
$\I(R)\cap (X\cup\{\tau\})=\emptyset$ is not even used.
Now from $R' = H'[\vec{P}] \bis{r}\, H'[\RS[\vec{P}]] \B  H'[\RS[\vec{Q}]]\,\bis{r} H'[\vec{Q}] = T'$
it follows that $$\theta_X(R') \bis{r}\, \theta_X(H'[\RS][\vec{P}]) \B \theta_X(H'[\RS][\vec{Q}]) \,\bis{r} \theta_X(T')$$
using \thm{congruence} and the observation that $\theta_X(H'[\RS[\vec{P}]]) = \theta_X(H')[\RS[\vec{P}]]$.
\qed
\end{itemize}
\end{proof}

\section{Complete axiomatisations}\label{sec:axioms}

Let $\textit{Ax}$ denote the collection of axioms from Tables~\ref{tab:axioms CCSP},~\ref{tab:axioms thetaX}
and~\ref{tab:reactive axioms}, $\textit{Ax}'$ the
ones from Tables~\ref{tab:axioms CCSP} and~\ref{tab:axioms thetaX}, and
$\textit{Ax}''$ merely the ones from \tab{axioms CCSP}.
Moreover, let $\textit{Ax}_f$, resp.\ $\textit{Ax}'_f$
and $\textit{Ax}''_f$, be same collections without the two
axioms using the recursion construct $\rec{x|\RS}$, RDP and RSP\@.
In this section I establish the following.
\\[-2pt]
Let $P$ and $Q$ be recursion-free CCSP$_\rt$ processes. Then
\begin{minipage}[b]{3in}
\begin{equation}\label{bisSCf}
P \bis{} Q ~~\Leftrightarrow~~ \textit{Ax}''_f \vdash P = Q .
\end{equation}
\end{minipage}\\[-2pt]
Let $P$ and $Q$ be CCSP$_\rt$ processes with guarded recursion. Then\hspace{-40pt}
\begin{minipage}[b]{3in}
\begin{equation}\label{bisSC}
P \bis{} Q ~~\Leftrightarrow~~ \textit{Ax}'' \vdash P = Q .
\end{equation}
\end{minipage}\\[-2pt]
Let $P$ and $Q$ be recursion-free $\CCSP$ processes. Then
\begin{minipage}[b]{3in}
\begin{equation}\label{bisSCthetaf}
P \bis{} Q ~~\Leftrightarrow~~ \textit{Ax}'_f \vdash P = Q .
\end{equation}
\end{minipage}\\[-2pt]
Let $P$ and $Q$ be $\CCSP$ processes with guarded recursion. Then\hspace{-40pt}
\begin{minipage}[b]{3in}
\begin{equation}\label{bisSCtheta}
P \bis{} Q ~~\Leftrightarrow~~ \textit{Ax}' \vdash P = Q .
\end{equation}
\end{minipage}\\[-2pt]
Let $P$ and $Q$ be recursion-free $\CCSP$ processes. Then
\begin{minipage}[b]{3in}
\begin{equation}\label{reacSCf}
P \bis{r} Q ~~\Leftrightarrow~~ \textit{Ax}_f \vdash P = Q .
\end{equation}
\end{minipage}\\[-2pt]
Let $P$ and $Q$ be $\CCSP$ processes with guarded recursion. Then\hspace{-40pt}
\begin{minipage}[b]{3in}
\begin{equation}\label{reacSC}
P \bis{r} Q ~~\Leftrightarrow~~ \textit{Ax} \vdash P = Q .
\end{equation}
\end{minipage}\\
In each of these cases ``$\Leftarrow$'' states the soundness of the axiomatisation and ``$\Rightarrow$'' completeness.

\Sec{axioms strong CCSP} recalls (\ref{bisSC}), which stems from \cite{GM20}, and (\ref{bisSCf}), which is folklore.
Then \Sec{axioms strong} extends the existing proofs of (\ref{bisSC}) and (\ref{bisSCf}) to obtain
(\ref{bisSCtheta}) and (\ref{bisSCthetaf}).
In \Sec{axioms reactive} I move from strong bisimilarity to strong reactive bisimilarity;
I discuss the merits of the axiom RA from \tab{reactive axioms}, and
establish its soundness, thereby obtaining direction ``$\Leftarrow$'' of (\ref{reacSC}) and (\ref{reacSCf}).
I prove the completeness of $\textit{Ax}_f$ for recursion-free processes---direction ``$\Rightarrow$'' of
(\ref{reacSCf})---in \Sec{finite}. Sections~\ref{sec:method}--\ref{sec:infinite} deal with
the completeness of $\textit{Ax}$ for guarded $\CCSP$---direction ``$\Rightarrow$'' of (\ref{reacSC}).
\Sec{choice} explains why I need the axiom of choice for the latter result.

\subsection[A complete axiomatisation of strong bisimilarity on guarded CCSP]
           {A complete axiomatisation of strong bisimilarity on guarded CCSP$_\rt$}\label{sec:axioms strong CCSP}

\begin{table}[t]
\vspace{-6pt}
\caption{A complete axiomatisation of strong bisimilarity on guarded CCSP$_\rt$}
\label{tab:axioms CCSP}
\begin{center}
\framebox{$\begin{array}{r@{~=~}l@{\qquad}r@{~=~}l@{}l r@{~=~}l@{}}
    x+(y+z) & (x+y)+z  & \tau_I(x+y) & \tau_I(x) + \tau_I(y)   && \Rn(x+y) & \Rn(x) + \Rn(y)
\\
x+y & y+x          & \tau_I(\alpha.x) & \alpha.\tau_I(x) & \mbox{\small if $\alpha\mathbin{\notin} I$}
                   & \Rn(\tau.x) & \tau.\Rn(x)
\\
x+x & x            & \tau_I(\alpha.x) & \tau.\tau_I(x) & \mbox{\small if $\alpha\mathbin\in I$}
                   & \Rn(\rt.x) & \rt.\Rn(x)
\\
x+0 & 0            & \rec{x|\RS} & \rec{\RS_x | \RS} & \mbox{(RDP)}
                   & \Rn(a.x) & \plat{\hspace{-1em}$\displaystyle\sum_{\{b\mid (a,b)\in \Rn\}}\hspace{-1em} b.\Rn(x)$}
\\
\multicolumn{7}{l}{\mbox{If $\rule{0pt}{15pt}\displaystyle P= \sum_{i\in I}\alpha_i.P_i$ and $\displaystyle Q= \sum_{j\in J}\beta_j.Q_j$ then}}\\
\multicolumn{7}{c}{\displaystyle P \spar{S} Q = \sum_{i\in I,~\alpha_i \notin S}(\alpha_i.P_i \spar{S} Q) + \sum_{j\in J,~\beta_j\notin S}(P \spar{S} \beta_j.Q_j)
+ \!\!\!\sum_{i\in I,~j\in J,~ \alpha_i=\beta_j\in S}\!\!\! \alpha_i.(P_i \spar{S} Q_j)}\\
\hline
\multicolumn{7}{l}{\mbox{Recursive Specification Principle (RSP)} \qquad\qquad \rule{0pt}{15pt}\RS
  \Rightarrow x = \rec{x|\RS}} \qquad\qquad \mbox{($\RS$ guarded)}
\end{array}$}
\end{center}
\end{table}
The well-known axioms of \tab{axioms CCSP} are \emph{sound} for strong bisimilarity, meaning that
writing $\bis{\,}$ for $=$, and substituting arbitrary expressions for the free variables $x,y,z$,
or the meta-variables $P_i$ and $Q_j$, turns them into true statements. In these axioms
$\alpha,\beta$ range over $Act$ and $a,b$ over $A$.  All axioms involving variables are
equations. The axiom involving $P$ and $Q$ is a template that stands for a family of equations, one
for each fitting choice of $P$ and $Q$. This is the CCSP$_\rt$ version of the \emph{expansion law}
from \cite{Mi90ccs}. The axiom RDP ($\rec{x|\RS} = \rec{\RS_x | \RS}$) says that recursively defined
processes $\rec{x|\RS}$ satisfy their set of defining equations $\RS$. As discussed in the previous
section, this entails that each recursive specification has a solution. The axiom RSP \cite{BW90,Fok00} is a
conditional equation with the equations of a guarded recursive specification $\RS$ as antecedents.
It says that the $x$-component of any solution of $\RS$---a vector of processes substituted
for the variables $V_\RS$---equals $\rec{x | \RS}$. In other words, each solution of $\RS$ equals
the default solution. This is a compact way of saying that solutions of guarded recursive
specifications are unique.

\begin{theorem}{completeness}
For CCSP$_\rt$ processes $P,Q\in\IP$ with guarded recursion, one has $P \bis{\,} Q$, that is, $P$ and $Q$ are
strongly bisimilar, iff $P=Q$ is derivable from the axioms of \tab{axioms CCSP}.
\end{theorem}
In this theorem, ``if'', the \emph{soundness} of the axiomatisation of \tab{axioms CCSP}, is an immediate
consequence of the soundness of the individual axioms.
``Only if'' states the \emph{completeness} of the axiomatisation.

A crucial tool in its proof is the simple observation that the axioms from the first box of
\tab{axioms CCSP} allow any CCSP$_\rt$ process with guarded recursion to be brought in the form
$\sum_{i\in I}\alpha_i.P_i$---a \emph{head normal form}.
Using this, the rest of the proof is a standard argument employing RSP, independent of the choice of
the specific process algebra. It can be found in \cite{Mi84,Mi90ccs}, \cite{BW90}, \cite{Fok00} and many other places.
However, in the literature this completeness theorem was always stated and proved for a small
fragment of the process algebra, allowing only guarded recursive specifications with a finite number
of equations, and whose right-hand sides $\RS_y$ involve only the basic operators inaction, action
prefixing and choice. Since the set of true statements $P \bis{\,} Q$, with $P$ and $Q$ processes
in a process algebra like guarded CCSP$_\rt$, is well-known to be undecidable, and even not recursively
enumerable, it was widely believed that no sound and complete finitely presented
axiomatisation of strong bisimilarity could exist.
Only in March 2017, Kees Middelburg observed (in the setting of the process algebra ACP \cite{BW90,Fok00})
that the standard proof applies almost verbatim to arbitrary processes with guarded recursion, although one has to be a
bit careful in dealing with the infinite nature of recursive specifications.
The argument has been carefully documented in \cite{GM20}, in the setting of the process algebra ACP\@.
This result does not contradict the non-enumerability of the set of true statements $P \bis{} Q$,
due to the fact that RSP is a proof rule with infinitely many premises.

A well-known simplification of \thm{completeness} and its proof also yields completeness without recursion:

\begin{theorem}{completeness for finite processes}
For CCSP$_\rt$ processes $P,Q\in\IP$ without recursion, one has $P \bis{\,} Q$ iff $P=Q$ is
derivable from the axioms of \tab{axioms CCSP} minus RDP and RSP\@.
\end{theorem}

\subsection[A complete axiomatisation of strong bisimilarity]
           {A complete axiomatisation of strong bisimilarity on guarded $\CCSP$}\label{sec:axioms strong}

\begin{table}[t]
\vspace{-6pt}
\caption{A complete axiomatisation of strong bisimilarity on guarded $\CCSP$}
\label{tab:axioms thetaX}
\begin{center}
\framebox{$\begin{array}{ll}
\theta_L^U(\sum_{i\in I} \alpha_i.x_i) = \sum_{i\in I} \alpha_i.x_i &
              (\alpha_i\notin L\cup\{\tau\} \mbox{~for all~}i\mathbin\in I)\\[2pt]
\theta_L^U(x+\alpha.y+\beta.z) = \theta_L^U(x + \alpha.y) &(\alpha \in L\cup\{\tau\} \wedge \beta\notin U\cup\{\tau\})\\[2pt]
\theta_L^U(x+\alpha.y+\beta.z) = \theta_L^U(x + \alpha.y) + \theta_L^U(\beta.z) &
            (\alpha \in L\cup\{\tau\} \wedge \beta\in U\cup\{\tau\})\\[2pt]
\theta_L^U(\beta.x) = \beta.x  &(\beta\neq \tau)\\[2pt]
\theta_L^U(\tau.x) = \tau.\theta_L^U(x) \\[8pt]
\psi_X(x+\alpha.z) = \psi_X(x) + \alpha.z & (\alpha \notin X\cup\{\tau,\rt\})\\[2pt]
\psi_X(x+\alpha.y+\rt.z) = \psi_X(x + \alpha.y) &(\alpha \in X\cup\{\tau\})\\[2pt]
\psi_X(x+\alpha.y+\beta.z) = \psi_X(x + \alpha.y) + \beta.z &
            (\alpha,\beta \in X\cup\{\tau\})\\[2pt]
\psi_X(\alpha.x) = \alpha.x  &(\alpha\neq \rt)\\[2pt]
\psi_X(\sum_{j\in I}\rt.y_i) = \sum_{j\in I}\rt.\theta_X(y_j) \\[2pt]
\end{array}$}
\end{center}
\end{table}

\tab{axioms thetaX} extends \tab{axioms CCSP} with axioms for the auxiliary operators $\theta_L^U$ and $\psi_X$.
With \tab{sos CCSP} it is straightforward to check the soundness of these axioms.
The fourth axiom, for instance, follows from the second or third rule for $\theta_L^U$ in \tab{sos CCSP},
depending on whether $\beta \in L\cup\{\rt\}$.
Moreover, a straightforward induction shows that these axioms suffice to convert each $\CCSP$ process with
guarded recursion into the form $\sum_{I\in I}\!\alpha_i.P_i$---a head normal form.
The below proposition sharpens this observation by pointing out that one can take the processes $P_i$
for $i\mathbin\in I$ to be exactly the ones that are reachable by one $\alpha_i$-labelled transition from $P$.
\begin{definition}{hnf}
Given a $\CCSP$ process $P \in \IP$, let \plat{$\widehat P := \sum_{\{(\alpha,Q) \mid P \goesto{\scriptscriptstyle\alpha} Q\}}\alpha.Q$}.
\end{definition}
By \pr{countably branching}, $P$ is countably branching, so using \pr{absolute expressiveness}
$\widehat P$ is a valid $\CCSP$ process.
In case $P \mathbin\in \IP$ is a process with only guarded recursion, then $P$ is finitely branching by
\pr{finitely branching}, so also $\widehat P$ is a valid $\CCSP$ process with only guarded recursion.

\begin{proposition}{hnf}
Let $P \in \IP$ have guarded recursion only. Then $\textit{Ax}' \vdash P = \widehat P$.
The conditional equation RSP is not even needed here.
\end{proposition}

\begin{proof}
The proof is by induction on the measure $e(P)$, defined in the proof of \pr{finitely branching}.

Let $P = \rec{x|\RS}$. Axiom RDP yields $\textit{Ax} \vdash P = \rec{x|\RS} = \rec{\RS_x|\RS}$.
Moreover, $e(\rec{\RS_x|\RS}) < e(\rec{x|\RS})$. So by induction, 
\plat{$\textit{Ax} \vdash \rec{\RS_x|\RS} = \widehat{\raisebox{0pt}[8pt]{$\rec{\RS_x|\RS}$}}$}.
Moreover, $\{(\alpha,Q) \mid \rec{\RS_x|\RS} \goesto{\scriptscriptstyle\alpha} Q\} =
\{(\alpha,Q) \mid \rec{x|\RS} \goesto{\scriptscriptstyle\alpha} Q\}$,\vspace{1pt} so
\plat{$\widehat{\raisebox{0pt}[8pt]{$\rec{\RS_x|\RS}$}} = \widehat{\raisebox{0pt}[8pt]{$\rec{x|\RS}$}} = \widehat P$}.
Thus $\textit{Ax} \vdash P = \widehat P$.

Let $P= \theta_L^U(P')$. Using that $e(P')<e(P)$, by induction
$\textit{Ax} \vdash P' = \widehat {P'}$ so  $\textit{Ax} \vdash P = \theta_L^U(\widehat {P'})$.
Let $$\widehat {P'} = \sum_{h\in H} \tau.P_h +
\sum_{i\in I} a_i.Q_i + \sum_{j\in J} b_j.R_j + \sum_{k\in K} \gamma_k.T_k\;,$$
where $a_i\in L$ for all $i \mathbin\in I$, $b_j\in U{\setminus}L$ for all $j \mathbin\in J$,
and $\gamma_k\notin U\cup\{\tau\}$ for all $k \mathbin\in K$. (So $\gamma_k$ may be $\rt$.)

\noindent
In case $H\cup I=\emptyset$, one has $\textit{Ax} \vdash P =  \theta_L^U(\widehat {P'}) = \widehat {P'} = \widehat P$, using the first
axiom for $\theta_L^U$.
Otherwise $$\textit{Ax} \vdash P = \sum_{h\in H} \tau.\theta_L^U(P_h) + \sum_{i\in I} a_i.Q_i + \sum_{j\in J} b_j.R_j$$
by the remaining four axioms for $\theta_L^U$. The right-hand side is $\widehat P$.

The cases for the remaining operators are equally straightforward.
\end{proof}
In the special case that $P$ is a recursion-free process, also the axiom RDP is not needed for this result.

Once we have head normalisation, the proofs of Theorems~\ref{thm:completeness}
and~\ref{thm:completeness for finite processes} are independent of the precise syntax of the process
algebra in question. Using \pr{hnf} we immediately obtain (\ref{bisSCtheta}) and (\ref{bisSCthetaf}):
\begin{theorem}{completeness thetaX}
For $\CCSP$ processes $P,Q\in\IP$ with guarded recursion, one has $P \bis{\,} Q$ iff $P=Q$ is derivable from the
axioms of Tables~\ref{tab:axioms CCSP} and~\ref{tab:axioms thetaX}.
\qed
\end{theorem}

\begin{theorem}{completeness thetaX finite}
For $\CCSP$ processes $P,Q\in\IP$ without recursion, one has $P \bis{\,} Q$ iff $P=Q$ is
derivable from the axioms of Tables~\ref{tab:axioms CCSP} and~\ref{tab:axioms thetaX} minus RDP and RSP\@.
\end{theorem}
A law that turns out to be particularly useful in verifications modulo strong reactive bisimilarity is
\begin{center}
 \framebox{$\theta_K^V(\theta_L^U(x)) \bis{} \theta_{K\cup L}^{V\cap U}(x)$\qquad 
 {\small provided $U\mathbin=V$ or $K\mathbin=L$ or $K\mathbin\subseteq L \mathbin\subseteq U \mathbin\subseteq V$
 or $L\mathbin\subseteq K \mathbin\subseteq V \mathbin\subseteq U$}\qquad  (L1)}\;.
\end{center}
Note that the right-hand side only exists if $(K\cup L)\subseteq(V\cap U)$.
This law is sound for strong bisimilarity, as demonstrated by the following proposition.
Yet it is not needed to add it to \tab{axioms thetaX}, as all its closed instances are derivable.
In fact, this is a consequence of the above completeness theorems.

\begin{proposition}{theta collapse}
$\theta_K^V(\theta_L^U(P)) \bis{} \theta_{K\cup L}^{V\cap U}(P)$, provided $(K\cup L)\subseteq(V\cap U)$
  and either $U\mathbin=V$ or $K\mathbin=L$ or $K\mathbin\subseteq L \mathbin\subseteq U \mathbin\subseteq V$
 or $L\mathbin\subseteq K \mathbin\subseteq V \mathbin\subseteq U$.
\end{proposition}
\begin{proof}
For given $K,L,U,V\subseteq A$ with $(K\cup L)\subseteq(V\cap U)$ and either $U\mathbin=V$ or
$K\mathbin=L$ or $K\mathbin\subseteq L \mathbin\subseteq U \mathbin\subseteq V$
or $L\mathbin\subseteq K \mathbin\subseteq V \mathbin\subseteq U$, let\vspace{-2.5ex}
  $${\B} :=\textit{Id} \cup \left\{\big(\theta_K^V(\theta_L^U(P)), \theta_{K\cup L}^{V\cap U}(P)\big) \mid P \in \IP\right\}\vspace{-1ex}.$$
  It suffices to show that the symmetric closure $\widetilde\B$ of $\B$ is a strong bisimulation.
  So let $R \mathrel{\widetilde\B} T$ and $R \goesto{\alpha} R'$ with $\alpha\in A \cup \{\tau,\rt\}$.
  I have to find a $T'$ with $T \goesto{\alpha} T'$ and $R' \mathrel{\widetilde\B} T'$.
\begin{itemize}
\item
The case that $R = T$ is trivial.
\item
Let $R = \theta_K^V(\theta_L^U(P))$ and $T = \theta_{K\cup L}^{V\cap U}(P)$.

First assume $\alpha=\tau$. Then $P\goesto\tau P'$ for some $P'$ such that $R' = \theta_K^V(\theta_L^U(P'))$.\\
Hence $T \mathbin= \theta_{K\cup L}^{V\cap U}(P) \goesto\tau \theta_{K\cup L}^{V\cap U}(P') =: T'$, and $R' \B T'$.

Now assume $\alpha \mathbin\in A \cup\{\rt\}$.
Then $\theta_L^U(P) \goesto{\alpha} R'$ and either $\alpha \in V$ or $\theta_L^U(P){\ngoesto\beta}$
for all $\beta\in K\cup\{\tau\}$. Using that $K\subseteq U$, this implies that either $\alpha \in V$ or $P{\ngoesto\beta}$
for all $\beta\in K\cup\{\tau\}$. Moreover, $P \goesto{\alpha} R'$ and either $\alpha \in U$ or
$P{\ngoesto\beta}$ for all $\beta\in L\cup\{\tau\}$. It follows that either $\alpha \in V\cap U$
or $P{\ngoesto\beta}$ for all $\beta\in K\cup L\cup\{\tau\}$. (Here I use that either $U\mathbin=V$
or $K\mathbin=L$ or $K\mathbin\subseteq L \mathbin\subseteq U \mathbin\subseteq V$
or $L\mathbin\subseteq K \mathbin\subseteq V \mathbin\subseteq U$.)\,
Consequently, $T = \theta_{K\cup L}^{V\cap U}(P) \goesto\alpha R'$.

\item
Let $R = \theta_{K\cup L}^{V\cap U}(P)$ and $T = \theta_K^V(\theta_L^U(P))$.

First assume $\alpha=\tau$. Then $P\goesto\tau P'$ for some $P'$ such that $R' =\theta_{K\cup L}^{V\cap U}(P')$.\\
Hence $T \mathbin= \theta_K^V(\theta_L^U(P)) \goesto\tau \theta_K^V(\theta_L^U(P')) =: T'$,
and \plat{$R' \mathrel{\widetilde\B} T'$}.

Now assume $\alpha \mathbin\in A \cup\{\rt\}$.
Then $P\goesto{\alpha} R'$ and either $\alpha \in V\cap U$
or $P{\ngoesto\beta}$ for all $\beta\in K\cup L\cup\{\tau\}$.
Consequently, $\theta_L^U(P) \goesto\alpha R'$ and thus $T = \theta_K^V(\theta_L^U(P)) \goesto\alpha R'$.
\qed
\end{itemize}
\end{proof}
The side condition to L1 cannot be dropped, for
$\theta_{\{c\}}^{\{a,c\}}\theta_\emptyset^{\{c\}}(a.0+c.0) \goesto{a} 0$, yet $\theta_{\{c\}}^{\{c\}}(a.0+c.0) \ngoesto{a}$.

\subsection[A complete axiomatisation of strong reactive bisimilarity]
           {A complete axiomatisation of strong reactive bisimilarity on guarded $\CCSP$}\label{sec:axioms reactive}

To obtain a sound and complete axiomatisation of strong reactive bisimilarity for $\CCSP$ with guarded recursion,
one needs to combine the axioms of Tables~\ref{tab:axioms CCSP},~\ref{tab:axioms thetaX} and~\ref{tab:reactive axioms}.
\begin{table}[t]
\vspace{-6pt}
\caption{A complete axiomatisation of strong reactive bisimilarity on guarded $\CCSP$}
\label{tab:reactive axioms}
\begin{center}
\framebox{$\begin{array}{c@{\qquad}c}
\displaystyle\frac{\psi_X(x) = \psi_X(y) \mbox{~for all~} X\subseteq A}{x=y} & \mbox{(RA)}
\end{array}$}
\end{center}
\end{table}
These axioms are useful only in combination with the full congruence
property of strong reactive bisimilarity, \thm{full congruence}. This is what allows us to apply
these axioms within subexpressions of a given expression.
Since ${\bis{\,}} \subseteq {\bis{r}}$\,, the soundness of all equational axioms for strong reactive
bisimilarity follows from their soundness for strong bisimilarity. The soundness of RSP has been
established as \thm{RSP}.
The soundness of RA, the \emph{reactive approximation axiom}, is contributed by the following proposition.

\begin{proposition}{reactive approximation}
Let $P,Q\in\IP$. If $\psi_X(P) \bis{r} \psi_X(Q)$ for all
$X\subseteq A$, then $P \bis{r} Q$.
\end{proposition}
\begin{proof}
Given $P,Q\mathbin\in\IP$ with $\psi_X(P) \bis{r} \psi_X(Q)$ for all $X\mathbin\subseteq A$,
I show that ${\B} := {\bis{r}}\, \cup \{(P,Q),(Q,P)\}$ is a strong time-out bisimulation.

Let $P\goesto\alpha P'$ with $\alpha\in A\cup\{\tau\}$.
Take any $X\subseteq A$. Then $\psi_X(P)\goesto\alpha P'$.
Since $\psi_X(P) \bis{r} \psi_X(Q)$, this implies $\psi_X(Q)\goesto\alpha Q'$ for some $Q'$ with
$P' \bis{r} Q'$, and hence $Q\goesto\alpha Q'$.

Let $P\goesto\rt P'$ and $\I(P)\cap (X\cup \{\tau\}) = \emptyset$.
Then $\psi_X(P)\goesto\rt \theta_X(P')$ and $\I(\psi_X(P)) \cap\linebreak[2] (X\cup \{\tau\}) \mathbin= \emptyset$.
Since $\psi_X(P) \bis{r} \psi_X(Q)$, this implies $\psi_X(Q)\goesto\rt Q''$ for some $Q''$ with
$\theta_X(\theta_X(P')) \bis{r} \theta_X(Q'')$. It must be that $Q \goesto\rt Q'$ for some $Q'$
with $Q'' = \theta_X(Q')$. By \pr{theta collapse}, $\theta_X(\theta_X(R)) \bis{\,} \theta_X(R)$
for all $R\mathbin\in\IP$. Thus
$\theta_X(P') \bis{\,} \theta_X(\theta_X(P')) \bis{r} \theta_X(\theta_X(Q')) \bis{\,} \theta_X(Q')$,
which had to be shown.
\end{proof}

\noindent
At first sight it appears that axiom RA is not very handy, as, in case the alphabet $A$ of
visible actions is finite, the number of premises to verify is exponential in the size $A$.
In case $A$ is infinite, there are even
uncountably many premises. However, in practical verifications this is hardly an issue, as one uses
a partition of the premises into a small number of equivalence classes, each of which
requires only one common proof. This technique will be illustrated on three examples below.
Furthermore, one could calculate the set of visible actions $\mathcal{J}(P)$ of a process $P$ that can
be encountered as initial actions after one $\rt$-transition followed by a sequence of
$\tau$-transitions. For large classes of processes, $\mathcal{J}(P)$ will be a finite set.
Now axiom RA can be modified by changing $X \subseteq A$ into $X \subseteq \mathcal{J}(P)\cup \mathcal{J}(Q)$.
This preserves the soundness of the axiom, because only the actions in $\mathcal{J}(P)$ play any r\^ole in evaluating $\psi_X(P)$.

A crucial property of strong reactive bisimilarity was mentioned in the introduction:
\begin{center}
 \framebox{$\tau.P + \rt.Q = \tau.P$ \qquad (L2)}\;.
\end{center}
It is an immediate consequence of RA, since
$\psi_X(\tau.P + \rt.Q) = \psi_X(\tau.P)$ for any $X\subseteq A$, by \tab{axioms thetaX}.
\hypertarget{L3}{Another useful law in verifications modulo strong reactive bisimilarity is
\begin{center}
 \framebox{$\sum_{i\in I}a_i. x_i + \rt.y = \sum_{i\in I}a_i .x_i + \rt.\theta_\emptyset^{A\setminus \it In}(y)$,
 where ${\it In} = \{a_i\mid i\in I\}$.
 \qquad (L3)}
\end{center}
Its soundness is intuitively obvious: the $\rt$-transition to $y$ will be taken only in an
environment $X$ with $X \cap {\it In} = \emptyset$. Hence one can just as well restrict the behaviour of $y$
to those transitions that are allowed in one such environment. This law was one of the prime reasons
for extending the family of operators $\theta_X$ ($= \theta_X^X$), which were needed to establish the
key theorems of this paper, to the larger family $\theta_L^U$.
Law L3 for finite $I$ is effortlessly derivable from its simple instance
\begin{center}
 \framebox{$a.x + \rt. y = a.x + \rt.\theta_\emptyset^{A\setminus \{a\}}(y)$. \qquad (L3$'$)}
\end{center}
in combination with L1.} I now show how to derive L3 from RA. For this proof I need to partition the set of premises
of RA in only two equivalence classes.

First let $X\cap {\it In} \neq \emptyset$. Then
$\psi_X(\sum_{i\in I}a_i. x_i + \rt.y) = \sum_{i\in I}a_i. x_i = \psi_X(\sum_{i\in I}a_i .x_i + \rt.\theta_\emptyset^{A\setminus \it In}(y))$.

Next let $X\cap {\it In} = \emptyset$. Then
$\psi_X(\sum_{i\in I}a_i. x_i + \rt.y) \begin{array}[t]{@{~=~}l}\sum_{i\in I}a_i. x_i + \rt.\theta_X(y) \\
\sum_{i\in I}a_i .x_i + \rt.\theta_X(\theta_\emptyset^{A\setminus \it In}(y)) \\
\psi_X(\sum_{i\in I}a_i .x_i + \rt.\theta_\emptyset^{A\setminus \it In}(y))\;,
\end{array}$\\
where the second step is an application of L1.

As an application of L3$'$ one obtains the law from \cite{vG21} that was justified in the
introduction:
\[a.P + \rt.(Q + \tau.R + a.S) \begin{array}[t]{@{~=~}l}
                              a.P + \rt.\theta_\emptyset^{A\setminus \{a\}}(Q + \tau.R + a.S) \\
                              a.P + \rt.\theta_\emptyset^{A\setminus \{a\}}(Q + \tau.R) \\
                              a.P + \rt.(Q + \tau.R)\;.
\end{array}\]

As a third illustration of the use of RA I derive an equational law that does not follow from
L1, L2 and L3, namely
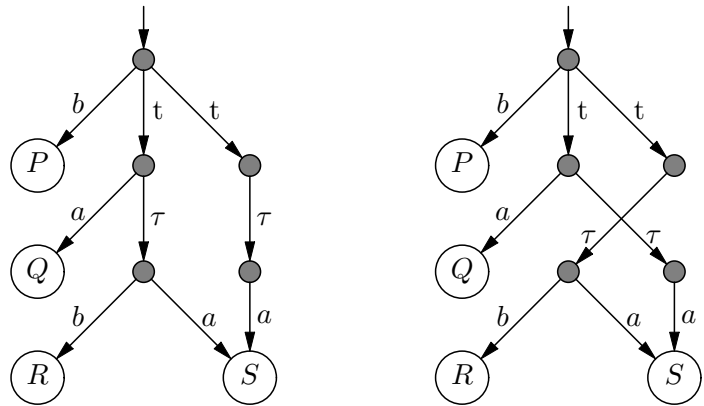
\vspace{-1ex}
\[b.P + \rt.\textcolor{red}(a.Q+\tau.(b.R+a.S)\textcolor{red}) + \rt.\tau.a.S  =
  b.P + \rt.\textcolor{blue}(a.Q+\tau.a.S\textcolor{blue}) + \rt.\tau.\textcolor{magenta}(b.R+a.S\textcolor{magenta})\]
These are the systems depicted in \fig{L3}.
These systems are surely not strongly bisimilar.
Moreover, L3 does not help in proving them equivalent, as applying
\plat{$\theta_\emptyset^{A\setminus \{b\}}$} to any of the four targets of a $\rt$-transition does not
kill any of the transitions of those processes. In particular, \plat{$\theta_\emptyset^{A\setminus \{b\}}(b.R+a.S)=b.R+a.S$}.
To derive this law from RA, I partition $\Pow(A)$ into three equivalence classes.

First let $b \in X$. Then $\begin{array}[t]{@{~=~}l}
\multicolumn{1}{l}{\psi_X(b.P + \rt.\textcolor{red}(a.Q+\tau.(b.R+a.S)\textcolor{red}) + \rt.\tau.a.S)} \\
b.P \\
\psi_X(b.P + \rt.\textcolor{blue}(a.Q+\tau.a.S\textcolor{blue}) + \rt.\tau.\textcolor{magenta}(b.R+a.S\textcolor{magenta})).
\end{array}$

Next let $b\notin X$ and $a \in X$. Then
\[
\begin{array}{c@{~=~}l}
\multicolumn{2}{l}{\psi_X\big(b.P + \rt.\textcolor{red}(a.Q+\tau.(b.R+a.S)\textcolor{red}) + \rt.\tau.a.S\big)}\\
& b.P + \rt.\theta_X\textcolor{red}{\big(}a.Q+\tau.(b.R+a.S)\textcolor{red}{\big)} + \rt.\theta_X\big(\tau.a.S\big) \\
& b.P + \rt.\textcolor{red}(a.Q+\tau.\theta_X(b.R+a.S)\textcolor{red}) + \rt.\tau.\theta_X(a.S)  \\
& b.P + \rt.(a.Q+\tau.a.S) + \rt.\tau.a.S \\
& b.P + \rt.\textcolor{blue}(a.Q+\tau.\theta_X(a.S)\textcolor{blue}) +
       \rt.\tau.\theta_X\textcolor{magenta}{\big(}b.R+a.S\textcolor{magenta}{\big)} \\
& b.P + \rt.\theta_X\textcolor{blue}{\big(}a.Q+\tau.a.S\textcolor{blue}{\big)} +
       \rt.\theta_X\big(\tau.\textcolor{magenta}(b.R+a.S\textcolor{magenta})\big) \\
& \psi_X\big( b.P + \rt.\textcolor{blue}(a.Q+\tau.a.S\textcolor{blue}) + \rt.\tau.\textcolor{magenta}(b.R+a.S\textcolor{magenta})\big)\;.
\end{array}\]

Finally let $a,b \notin X$. Then
\[
\begin{array}{c@{~=~}l}
\multicolumn{2}{l}{\psi_X\big(b.P + \rt.\textcolor{red}(a.Q+\tau.(b.R+a.S)\textcolor{red}) + \rt.\tau.a.S\big)}\\
& b.P + \rt.\theta_X\textcolor{red}{\big(}a.Q+\tau.(b.R+a.S)\textcolor{red}{\big)} + \rt.\theta_X\big(\tau.a.S\big) \\
& b.P + \rt.\tau.\theta_X(b.R+a.S) + \rt.\tau.\theta_X(a.S) \\
& b.P + \rt.\tau.(b.R+a.S) + \rt.\tau.a.S \\
& b.P + \rt.\tau.a.S + \rt.\tau.\textcolor{magenta}(b.R+a.S\textcolor{magenta}) \\
& b.P + \rt.\tau.\theta_X(a.S) + \rt.\tau.\theta_X\textcolor{magenta}{\big(}b.R+a.S\textcolor{magenta}{\big)} \\
& b.P + \rt.\theta_X\textcolor{blue}{\big(}a.Q+\tau.a.S\textcolor{blue}{\big)} +
       \rt.\theta_X\big(\tau.\textcolor{magenta}(b.R+a.S\textcolor{magenta})\big) \\
& \psi_X\big( b.P + \rt.\textcolor{blue}(a.Q+\tau.a.S\textcolor{blue}) + \rt.\tau.\textcolor{magenta}(b.R+a.S\textcolor{magenta})\big)\;.
\end{array}\]

\subsection{Completeness for finite processes}\label{sec:finite}

\begin{theorem}{finite}
Let $P$ and $Q$ be closed recursion-free $\CCSP$ expressions. Then
$P \bis{r} Q \Rightarrow \textit{Ax}_f \vdash P \mathbin= Q$.
\end{theorem}

\begin{proof}
Let the \emph{length} of a path $P \goesto{\alpha_1} P_1 \goesto{\alpha_2} \dots \goesto{\alpha_n} P_n$
of a processes $P$ be $n$. Let $d(P)$, the \emph{depth} of $P$, be the length of its longest path;
it is guaranteed to exists when $P$ is a closed recursion-free $\CCSP$ expression.
I prove the theorem with induction on $\max(d(P),d(Q))$.

Suppose $P \bis{r} Q$. By \pr{hnf} one has $\textit{Ax}_f \vdash P = \widehat P$ and $\textit{Ax}_f \vdash Q = \widehat Q$.
I will show that $\textit{Ax}_f \vdash \psi_X(\widehat P) = \psi_X(\widehat Q)$ for all $X\subseteq A$.
This will suffice, as then Axiom RA yields $\textit{Ax}_f \vdash \widehat P = \widehat Q$ and
thus $\textit{Ax}_f \vdash P = Q$. So pick $X \subseteq A$. Let
$$\widehat P = \sum_{i\in I} \alpha_i.P'_i + \sum_{j\in J} \rt.P''_j
\qquad\mbox{and}\qquad\widehat Q = \sum_{k\in K} \beta_k.Q'_k + \sum_{h\in H} \rt.Q''_h$$
with $\alpha_j,\beta_k \mathbin\in A\cup\{\tau\}$ for all $i\mathbin\in I$ and $k\mathbin\in K$.
The following two claims are the crucial part of the proof.%
\vspace{1ex}

\noindent
\textit{Claim 1:} For each $i \mathbin\in I$ there is a $k \mathbin\in K$ with $\alpha_i=\beta_k$
and $\textit{Ax}_f \vdash P'_i = Q'_k$.

\noindent
\textit{Claim 2:} If $\I(P) \cap (X \cup\{\tau\}) = \emptyset$, then for each $j \mathbin\in J$ there is a $h \mathbin\in H$ with
$\textit{Ax}_f \vdash \theta_X(P''_j) = \theta_X(Q''_h)$.
\vspace{1ex}

\noindent
With these claims, the rest of the proof is straightforward.
Since $P \bis{r} Q$, one has $\I(P) = \I(\widehat P) = \{\alpha_i \mid i \mathbin\in I\} =\
\{\beta_k \mid k \mathbin\in K\} = \I(\widehat Q) = \I(Q)$.
First suppose that $\I(P) \cap (X \cup\{\tau\}) = \emptyset$. Then
$$\psi_X(\widehat P) = \sum_{i\in I} \alpha_i.P'_i + \sum_{j\in J} \rt.\theta_X(P''_j)
\qquad\mbox{and}\qquad \psi_X(\widehat Q) = \sum_{k\in K} \beta_k.Q'_k + \sum_{h\in H} \rt.\theta_X(Q''_h)\;.$$
Claim 1 yields
$\textit{Ax}_f \vdash \psi_X(\widehat Q) = \psi_X(\widehat Q) + \alpha_i.P'_i$ for each $i \mathbin\in I$.
Likewise, Claim 2 yields
$\textit{Ax}_f \vdash \psi_X(\widehat Q) = \psi_X(\widehat Q) + \rt.\theta_X(P''_j)$ for each $j \mathbin\in J$.
Together this yields  $\textit{Ax}_f \vdash \psi_X(\widehat Q) = \psi_X(\widehat Q) + \psi_X(\widehat P)$.
By symmetry one obtains $\textit{Ax}_f \vdash \psi_X(\widehat P) = \psi_X(\widehat P) + \psi_X(\widehat Q)$
and thus $\textit{Ax}_f \vdash \psi_X(\widehat P) = \psi_X(\widehat Q)$.
\vspace{1ex}

\noindent
Next suppose $\I(P) \cap (X \cup\{\tau\}) \neq \emptyset$. Then
$\psi_X(\widehat P) = \sum_{i\in I} \alpha_i.P'_i$ and $\psi_X(\widehat Q) = \sum_{k\in K} \beta_k.Q'_k$.
The proof proceeds just as above, but without the need for Claim 2.
\vspace{1ex}

\noindent
\textit{Proof of Claim 1:} Pick $i \mathbin\in I$. Then $\widehat P \goesto{\alpha_i} P'_i$.
So  $\widehat Q \goesto{\alpha_i} Q'$ for some $Q'$ with $P'_i \bis{r} Q'$.
Hence there is a $k \mathbin\in K$ with $\alpha_i=\beta_k$ and $Q'=Q'_k$.
Using that $d(P'_i) < d(P)$ and $d(Q'_i) < d(Q)$, by induction $\textit{Ax}_f \vdash P'_i = Q'_k$.
\vspace{1ex}

\noindent
\textit{Proof of Claim 2:} Pick $j \mathbin\in J$. Then \plat{$\widehat P \goesto{\rt} P''_j$}.
Since $\I(\widehat P) \cap (X \cup\{\tau\}) = \emptyset$, there is a $Q''$ such that 
$\widehat Q \goesto{\rt} Q''$ and $\theta_X(P''_j) \bis{r} \theta_X(Q'')$.
Hence there is a $h \mathbin\in H$ with $Q''=Q''_h$.
Using that $d(\theta_X(P''_j)) \leq d(P''_j) < d(P)$ and $d(\theta_X(Q''_h)) \leq d(Q''_h) < d(Q)$,
by induction $\textit{Ax}_f \vdash \theta_X(P''_j) = \theta_X(Q''_h)$.
\end{proof}

\subsection{The method of canonical representatives}\label{sec:method}

The classic technique of proving completeness of axiomatisations for process algebras with recursion
involves \emph{merging guarded recursive equations} \cite{Mi84,Mi89a,Wa90,vG93a,LDH05}.
In essence it proves two bisimilar systems $P$ and $Q$ equivalent by equating both to an
intermediate variant that is essentially a \emph{product} of $P$ and $Q$.
I tried so hard, and in vain, to apply this technique to obtain (\ref{reacSC}), that I came to
believe that it fundamentally does not work for this axiomatisation.
\begin{figure}[t]
\input{product}
\centerline{\box\graph}\vspace{1ex}
\caption{A failed product construction}
\label{fig:product}\vspace{-1ex}
\end{figure}
The problem is illustrated in \fig{product}. Here, similar to the example of \fig{L3},
the processes $1$ and $6$ are strongly reactive bisimilar. The merging technique constructs a
transition system whose states are pairs of states reachable from $1$ and $6$.
There is a transition $(s,t) \goesto\alpha (s',t')$ iff both
$s \goesto\alpha s'$ and $t \goesto\alpha t'$.
Normally, only those pairs $(s,t)$ satisfying $s \bis{} t$ are included.
Here the requirement $s \bis{r} t$ would be to strong. Namely, although $1 \bis{r} 6$, one has neither
$2 \bis{r} 7$ nor
$2 \bis{r} 8$ nor
$3 \bis{r} 7$ nor
$3 \bis{r} 8$, so there would be no outgoing $\rt$-transitions from $(1,6)$.
Hence one has to include states $(s,t)$ with $s \rbis{X}{r} t$ for some set $X$.
Note that $2 \rbis{X}{r} 7$ and $3 \rbis{X}{r} 8$ when $a\in X$ and $b \notin X$,
whereas $2 \rbis{X}{r} 8$ and $3 \rbis{X}{r} 7$ when $a\notin X$. This yields the product depicted in
\fig{product}.

In the reactive bisimulation game, the transition $1 \goesto\rt 2$ will be matched by $6 \goesto\rt 8$
only in an environment $X$ with $a \not\in X$. Hence intuitively the state $(2,8)$ in the product
should only be visited in such an environment. Yet, when aiming to show that $1 \bis{r} (1,6) \bis{r} 6$,
one cannot prevent taking the transition $(1,6) \goesto\rt (2,8)$ in an environment $X$ with $a\in X$
and $b \notin X$. However, since $(2,8) \ngoesto{a}$, this $\rt$-transition cannot be simulated by process $2$.

It may be possible to repair the construction, for instance by adding a transition $(2,8) \goesto{a} Q$
or $(2,8) \goesto{a} T$ after all, but not both. However, each such ad hoc repair that I tried gave
raise to further problems, making the solution more and more complicated without sight on success.

Therefore, I here employ the novel method of \emph{canonical solutions} \cite{GF20,LY20}, which equates
both $P$ and $Q$ to a canonical representative within the bisimulation equivalence class of $P$ and $Q$---one
that has only one reachable state for each bisimulation equivalence class of states of $P$ and $Q$.
Moreover, my proof employs the axiom of choice \cite{Zermelo08} in defining the transition
relation on my canonical representative, in order to keep this process finitely branching.

To illustrate his technique on the example from \fig{product}, the states $1$ and $6$, being strongly
reactive bisimilar, form one new state $\{1,6\}$ of the canonical representative.
Likewise, there will be states $\{4,9\}$ and $\{5,10\}$. However, the states $2$, $3$, $7$ and $8$ remain separate.
Within the new state $\{1,6\}$ my construction chooses an arbitrary element, say $1$. Based on this
choice, the outgoing transitions of $\{1,6\}$ are dictated by $1$, and thus go to $P$, $\{2\}$ and $\{3\}$.
As a result, the canonical representative will look just like the left-hand process.
It could however be the case that $S \bis{r} P$, in which case the initial states of these
subprocesses are merged in the canonical representative, and again an element in the resulting
equivalence class will be chosen that dictates its outgoing transitions.

\subsection{The canonical representative}\label{sec:minimisation}

Let $\IP^g$ denote the set of $\CCSP$ processes with guarded recursion.
Let $[P] := \{Q \in \IP^g \mid Q \bis{r} P\}$\linebreak[4] be the strong reactive bisimulation equivalence class
of a process $P\in\IP^g$.
Below, by ``abstract process'' I will mean such an equivalence class.
Choose a function $\chi$ that selects an element out of each $\bis{r}$\,-equivalence class of $\CCSP$
processes with guarded recursion---this is possible by the axiom of choice~\cite{Zermelo08}.
Define the transition relations $\goesto{\alpha}$, for $\alpha\in Act$, between abstract processes by
\begin{equation}\label{normal form}
R \goesto\alpha R' ~~\Leftrightarrow~~ \exists P'\in R'.~ \chi(R) \goesto\alpha P'\;.
\end{equation}
I will show that $P \bis{r} [P]$ for all $P \in \IP^g$.
Formally, $\bis{r}$ has been defined only between processes belonging to the same LTS
$\IP$, and here $[P]\notin\IP$. However, this restriction is not material: two processes $P \in \IP$
and $Q \in \IQ$ from different LTSs can be compared by considering $\bis{r}$ on the disjoint union $\IP \uplus \IQ$.

\begin{lemma}{minimisation activities}
Let $\alpha\in A \cup\{\tau\}$. Then $[P] \goesto\alpha R'$ iff $P \goesto\alpha P'$ for some $P'$ with $R'=[P']$.
\end{lemma}
\begin{proof}
Let $P \goesto\alpha P'$ with $\alpha \mathbin\in A \cup \{\tau\}$.
Since $P \bis{r} \chi([P])$, by \df{time-out bisimulation} there is a $Q'$ such that
$\chi([P]) \goesto\alpha Q'$ and $P' \bis{r} Q'$. Hence $[P] \goesto\alpha [Q']$ by (\ref{normal form}).
Moreover, $P' \mathbin\in [Q']{=}[P']$.

Let $[P] \goesto\alpha R'$ with $\alpha \in A \cup \{\tau\}$.
Then $\chi([P]) \goesto\alpha Q'$ for some $Q'\in R'$.
Since $\chi([P]) \bis{r} P$, there is a $P'$ such that $P \goesto\alpha P'$ and $Q' \bis{r} P'$.
Hence $P'\in R'$ and thus $R'\mathbin=[P']$.
\end{proof}

\begin{corollary}{initials minimisation}
$\I([P]) = \I(P)$ for all $P \in \IP^g$.
\qed
\end{corollary}

\begin{lemma}{minimisation timeout}
If $\I(P)\cap(X\cup\{\tau\})=\emptyset$ and $P \goesto\rt P'$
then $[P] \goesto\rt [Q']$ for a $Q'$ with $\theta_X(P') \bis{r} \theta_X(Q')$.
Moreover, if $\I(P)\cap(X\cup\{\tau\})=\emptyset$ and $[P] \goesto\rt [Q']$
then $P \goesto\rt P'$ for a $P'$ with $\theta_X(Q') \bis{r} \theta_X(P')$.
\end{lemma}
\begin{proof}
Let $\I(P)\cap(X\cup\{\tau\})=\emptyset$ and $P \goesto\rt P'$.
Since $P \bis{r} \chi([P])$, by \df{time-out bisimulation} there is a $Q'$ such that
$\chi([P]) \goesto\rt Q'$ and $\theta_X(P') \bis{r} \theta_X(Q')$.
Hence $[P] \goesto\rt [Q']$ by (\ref{normal form}).

Let $\I(P)\cap(X\cup\{\tau\})=\emptyset$ and $[P] \goesto\rt [Q']$.
Then $\chi([P]) \goesto\rt R'$ for some $R'\in [Q']$.
Since $\chi([P]) \bis{r} P$ (so $\I(\chi([P])) = \I(P)$), there is a $P'$ such that $P \goesto\rt P'$
and $\theta_X(R') \bis{r} \theta_X(P')$. 
As $\bis{r}$ is a congruence for $\theta_X$, one has $\theta_X(Q') \bis{r} \theta_X(R')$,
and thus $\theta_X(Q') \bis{r} \theta_X(P')$.
\end{proof}

\begin{definition}{uptoRT}
Let $\B^*  := \{(R,T) \mid \exists n\geq 0. ~ \exists R_0,\dots, R_n.~ R = R_0 \B R_1 \B \dots \B R_n = T\}$
denote the reflexive and transitive closure of a binary relation $\B$.
A \emph{strong time-out bisimulation up to reflexivity and transitivity}
is a symmetric relation ${\B} \subseteq \IP \times \IP$, such that, for $P\B Q$,
\begin{itemize}\itemsep 0pt \parsep 0pt
\item if $P \goesto{\alpha} P'$ with $\alpha\mathbin\in A\cup\{\tau\}$, then $\exists Q'$ such that
  $Q\goesto{\alpha} Q'$ and $P' \B^* Q'$,
\item if $\I(P)\cap (X\cup\{\tau\})=\emptyset$ and $P \goesto{\rt} P'$, then $\exists Q'$ with $Q\goesto{\rt} Q'$
  and $\theta_{X}(P') \B^* \theta_X(Q')$.
\end{itemize}
\end{definition}

\begin{proposition}{uptoRT}
If $P \B Q$ for a strong time-out bisimulation $\B$ up to reflexivity and transitivity, then $P \bis{r} Q$.
\end{proposition}
\begin{proof}
It suffices to show that $\B^*$ is a strong time-out bisimulation.
Clearly this relation is symmetric.
\begin{itemize}
\item
Suppose $R_0 \B R_1 \B \dots \B R_n$ for some $n\geq 0$ and $R_0 \goesto\alpha R'_0$ with $\alpha\in A\cup\{\tau\}$.
I have to find an $R'_n$ such that $R_n \goesto\alpha R'_n$ and $R'_0 \B^* R'_n$.
I proceed with induction on $n$. The case $n=0$ is trivial.
Fixing an $n>0$, by \df{uptoRT} there is an $R'_1$ such that $R_1 \goesto\alpha R'_1$ and $R'_0 \B^* R'_1$.
Now by induction there is an $R'_n$ such that $R_n \goesto\alpha R'_n$ and $R'_1 \B^* R'_n$. Hence $R'_0 \B^* R'_n$.
\item
Suppose $R_0 \B R_1 \B \dots \B R_n$ for some $n\geq 0$, $\I(R_0)\cap(X\cup\{\tau\})=\emptyset$
and $R_0 \goesto\rt R'_0$. By \df{uptoRT} $\I(R_0) = \I(R_1) = \dots = \I(R_n)$.
I have to find an $R'_n$ such that $R_n \goesto\rt R'_n$ and $\theta_X(R'_0) \B^* \theta_X(R'_n)$.
This proceeds exactly as for the case above.
\qed
\end{itemize}
\end{proof}

\begin{lemma}{minimisation}
$\theta_X([P]) \bis{r} [\theta_X(P)]$  for all $P \in \IP^g$ and $X \subseteq A$.
\end{lemma}
\begin{proof}
I show that the symmetric closure of
${\B} := \{(\theta_X([P]),[\theta_X(P)])\mid P \in \IP^g \wedge X \subseteq A\}$
is a strong time-out bisimulation up to reflexivity and transitivity.
\begin{itemize}
\item
Let $\theta_X([P]) \goesto\tau R'$.
Then $[P]\goesto \tau Q'$ for some $Q'$ with $R'=\theta_X(Q')$.
By \lem{minimisation activities},
$P \goesto \tau P'$ for some $P'$ with $Q'=[P']$.
Hence $\theta_X(P) \goesto\tau \theta_X(P')$ and thus $[\theta_X(P)] \goesto\tau [\theta_X(P')]$ by
\lem{minimisation activities}.
Moreover, $R' = \theta_X([P']) \B [\theta_X(P')]$.
\item
Let $[\theta_X(P)] \goesto\tau R'$.
By \lem{minimisation activities},
$\theta_X(P) \goesto\tau Q'$ for some $Q'$ with $R'=[Q']$.
Thus $P \goesto \tau P'$ for some $P'$ with $Q'=\theta_X(P')$.
Now $[P] \goesto \tau [P']$ by \lem{minimisation activities}, and thus $\theta_X([P]) \goesto \tau \theta_X([P'])$.
Moreover, $R' = [\theta_X(P')] \B^{-1} \theta_X([P'])$.
\item
Let $\theta_X([P]) \goesto a R'$ with $a \in A$.
Then $[P]\goesto a R'$ and either $a \in \I([P])$ or $[P] {\ngoesto\beta}$ for all $\beta \in X \cup \{\tau\}$.
Thus either $a \in \I(P)$ or $P {\ngoesto\beta}$ for all $\beta \in X \cup \{\tau\}$, using \cor{initials minimisation}.
By \lem{minimisation activities},
$P \goesto a P'$ for some $P'$ with $R'=[P']$.
Hence $\theta_X(P) \goesto a P'$ and thus $[\theta_X(P)] \goesto a [P'] = R'$.
\item
Let $[\theta_X(P)] \goesto a R'$ with $a \in A$.
By \lem{minimisation activities},
$\theta_X(P) \goesto a P'$ for some $P'$ with $R'=[P']$.
Thus $P \goesto a P'$ and either $a \in \I(P)$ or $P {\ngoesto\beta}$ for all $\beta \in X \cup \{\tau\}$.
Therefore either $a \in \I([P])$ or $[P] {\ngoesto\beta}$ for all $\beta \in X \cup \{\tau\}$, using
\cor{initials minimisation}.
Moreover, $[P] \goesto a [P']$ by \lem{minimisation activities}. It follows that $\theta_X([P]) \goesto a [P'] = R'$.
\item
Let $\I(\theta_X([P])) \cap(X \cup\{\tau\})=\emptyset$ and $\theta_X([P]) \goesto\rt R'=[Q']$.
Then $[P]\goesto\rt [Q']$ and $[P] {\ngoesto\beta}$ for all $\beta \in X \cup \{\tau\}$.
Thus $P {\ngoesto\beta}$ for all $\beta \in X \cup \{\tau\}$, using \cor{initials minimisation},
so $\I(P) \cap(X \cup\{\tau\})=\emptyset$ and $\I(\theta_X(P)) \cap(X \cup\{\tau\})=\emptyset$.
By \lem{minimisation timeout}, $P \goesto\rt P'$ for some $P'$ with $\theta_X(Q') \bis{r} \theta_X(P')$.
Hence $\theta_X(P) \goesto\rt P'$ and thus, again applying \lem{minimisation timeout},
$[\theta_X(P)] \goesto\rt [T']$ for some $T'$ with $\theta_X(P') \bis{r} \theta_X(T')$.
Moreover, $\theta_X(R') = \theta_X([Q']) \B [\theta_X(Q')] = [\theta_X(T')] \B^{-1} \theta_X([T'])$.
\item
Let $\I([\theta_X(P)])\cap(X\cup\{\tau\})=\emptyset$ and $[\theta_X(P)] \goesto\rt [Q']$.
By \lem{minimisation timeout}, $\theta_X(P)\goesto\rt P'$ for a $P'$ with $\theta_X(Q') \bis{r} \theta_X(P')$.
Hence $P \goesto\rt P'$ and $P {\ngoesto\beta}$ for all $\beta \in X \cup \{\tau\}$,
so $\I(P)\cap(X\cup\{\tau\})=\emptyset$.
Hence, by \lem{minimisation timeout}, $[P] \goesto\rt [T']$ for a $T'$ with $\theta_X(P') \bis{r} \theta_X(T')$.
By \cor{initials minimisation}, $[P] {\ngoesto\beta}$ for all $\beta \mathbin\in X \cup \{\tau\}$.
So $\theta_X([P]) \goesto\rt [T']$.
Moreover, $\theta_X([Q']) \B [\theta_X(Q')] = [\theta_X(T')] \B^{-1} \theta_X([T'])$.
\\\mbox{}\qed
\end{itemize}
\end{proof}

\begin{proposition}{canonical}
$P \bis{r} [P]$ for all $P \in \IP^g$.
\end{proposition}
\begin{proof}
Using \pr{upto r}, I show that the symmetric closure of the relation ${\B} := \{(P,[P])\mid P \mathbin\in \IP^g\}$ is
a strong time-out bisimulation up to $\bis{r}$\,. Here the right-hand side processes come from an LTS
that is closed under $\theta$ and contains the processes $[P]$ for $P \in \IP^g$.
\begin{itemize}
\item
Let $P \goesto\alpha P'$ with $\alpha \mathbin\in A \cup \{\tau\}$.
Then $[P] \goesto\alpha [P']$ by \lem{minimisation activities}, and $P' \B [P']$.
\item
Let $[P] \goesto\alpha R'$ with $\alpha \in A \cup \{\tau\}$.
Then, by \lem{minimisation activities}, $P \goesto\alpha P'$ for some $P'$ with $R'=[P']$.
Moreover, $R' \B^{-1} P'$.
\item
Let $\I(P)\cap(X\cup\{\tau\})=\emptyset$ and $P \goesto\rt P'$.
By \lem{minimisation timeout}, $[P] \goesto\rt [Q']$ for some $Q'$ such that $\theta_X(P') \bis{r} \theta_X(Q')$.
Moreover,  using \lem{minimisation}, $\theta_X(P') \B [\theta_X(P')] = [\theta_X(Q')] \bis{r} \theta_X([Q'])$.
\item
Let $\I([P])\cap(X\cup\{\tau\})=\emptyset$ and $[P] \goesto\rt [Q']$.
Then, by \lem{minimisation timeout}, $P \goesto\rt P'$ for some $P'$ with $\theta_X(Q') \bis{r} \theta_X(P')$.
By \lem{minimisation},
$\theta_X([Q']) \bis{r} [\theta_X(Q')] = [\theta_X(P')] \B^{-1} \theta_X(P')$.
\qed
\end{itemize}
\end{proof}
By \pr{finitely branching} each $P\in\IP^g$ is finitely branching.
By construction, so is $[P]$.

No two states reachable from $[P]$ are strongly reactive bisimilar.  Hence the process $[P]$ with
its above-generated transition relation can be seen as a version of $P$ were each equivalence class of
reachable states is collapsed into a single state---a kind of minimisation. But it is not exactly a
minimisation, as not all states reachable from $[P]$ need be strongly reactive bisimilar with
reachable states of $P$. This is illustrated by Process $6$ of \fig{product}, when $\chi(\{1,6\})=1$.
Now $\{2\}$ and $\{3\}$ are  reachable from $[P]$, but not strongly reactive bisimilar with
reachable states of $6$.

\subsection{Completeness for finitely branching processes}\label{sec:infinite}

I will now give a syntactic representation of each process $[P]$, for $P \in\IP^g$, as a $\CCSP$
process with guarded recursion. Take a different variable $x_R$ for each $\bis{r}$-equivalence
class $R$ of $\CCSP$ processes with guarded recursion. Let $V_\RS$ be the set of all those variables,
and define the recursive specification $\RS$ by\vspace{-1ex}
$$x_{R} = \sum_{R\goesto{\alpha} R'} \alpha.x_{R'}\;.$$
By construction, $R \bis{} \rec{x_{R} | \RS}$, that is, the process
$\rep{P} := \rec{x_{[P]} | \RS} \in \IP^g$ is strongly bisimilar to $[P]$.\linebreak[3]
In fact, the symmetric closure of the relation $\{([P],\rep{P}) \mid P \in \IP^g\}$ is a strong bisimulation.
Thus, $\rep{P}$ serves as a normal form within the $\bis{r}$-equivalence class of $P\in \IP^g$.

The above construction will not work when there are not as many variables as equivalences classes of 
$\CCSP$ processes with guarded recursion. Note that each real number in the interval $[0,1)$ can be
represented as an infinite sequence of $0$s and $1$s, and thus as a $\CCSP$ processes with guarded
recursion employing the finite alphabet $A=\{0,1\}$. Hence there are uncountably many 
equivalences classes of $\CCSP$ processes with guarded recursion.

To solve this problem, one starts here already with the proof of (\ref{reacSC}), and fixes two processes
$P_0$ and $Q_0 \in\IP^g$ with $P_0 \bis{r} Q_0$. The task is to prove $\textit{Ax} \vdash P_0 = Q_0$.
Now call an equivalence class $R$ of $\CCSP$ processes with guarded recursion \emph{relevant}
if either $R$ is reachable from $[P_0] = [Q_0]$, or a member of $R$ is reachable from $P_0$ or $Q_0$.
There are only countably many relevant equivalence classes. It suffices to take a variable
$x_R$ only for relevant $R$. Below, I will call a process $P \in \IP^g$ \emph{relevant}
if it is a member of a relevant equivalence class; in case we had enough variables to start with, all
processes $P \in \IP^g$ may be called relevant.

\begin{lemma}{collapse of normal form}
Let $P,Q\in\IP^g$ be relevant. Then $\rep{P} \bis{r} \rep{Q} \Rightarrow \rep{P}=\rep{Q}$.
\end{lemma}
\begin{proof}
Suppose $\rep{P} \bis{r} \rep{Q}$. Then $P \bis{r} \rep{P} \bis{r} \rep{Q} \bis{r} Q$, so $[P]=[Q]$,
and hence $\rep{P}=\rep{Q}$.
\end{proof}

\begin{lemma}{bisimulation collapse}
Let $P,Q\in\IP^g$ be relevant. Then
$\theta_X(\rep{P}) \bis{r} \theta_X(\rep{Q}) \Rightarrow \theta_X(\rep{P}) \bis{\,} \theta_X(\rep{Q})$.
\end{lemma}
\begin{proof}
I show that ${\B} := \textit{Id} \cup \{(\theta_X(\rep{P}), \theta_X(\rep{Q})) \mid \theta_X(\rep{P}) \bis{r} \theta_X(\rep{Q})\}$
is a strong bisimulation.

Suppose $\theta_X(\rep{P}) \bis{r} \theta_X(\rep{Q})$. 
Then $\rep{P} {\ngoesto\beta}$ for all $\beta\in X\cup\{\tau\}$ iff $\rep{Q} {\ngoesto\beta}$ for all $\beta\in X\cup\{\tau\}$,
since $\I(\rep{P}) \cap (X\cup\{\tau\}) = \I(\theta_X(\rep{P})) \cap (X\cup\{\tau\})
= \I(\theta_X(\rep{Q})) \cap (X\cup\{\tau\}) = \I(\rep{Q}) \cap (X\cup\{\tau\})$.

First consider the case that $\rep{P} {\ngoesto\beta}$ for all $\beta\in X\cup\{\tau\}$.
Then $\theta_X(\rep{P}) \bis{\,} \rep{P}$ and $\theta_X(\rep{Q}) \bis{\,} \rep{Q}$.
Hence $\rep{P} \bis{\,} \theta_X(\rep{P}) \bis{r} \theta_X(\rep{Q}) \bis{\,} \rep{Q}$.
So by \lem{collapse of normal form}, $\rep{P} \mathbin= \rep{Q}$, and thus
$\theta_X(\rep{P}) \mathrel{\textit{Id}} \theta_X(\rep{Q})$.

Henceforth I suppose that $\rep{P} {\goesto\beta}$ for some $\beta\in X\cup\{\tau\}$.
So $\rep{P} {\ngoesto\rt}$ and $\rep{Q} {\ngoesto\rt}$.
\begin{itemize}
\item Let $\theta_X(\rep{P}) \goesto{a} P''$ with $a \in A$.
  Then $\theta_X(\rep{Q}) \goesto{a} Q''$ for some $Q''$ with $P'' \bis{r} Q''$.
  One has $\rep{P} \goesto{a} P''$ and $\rep{Q} \goesto{a} Q''$.
  The process $P''$ must have the form $\rep{P'}$, and likewise $Q''=\rep{Q'}$.
  Since $\rep{P'} \bis{r} \rep{Q'}$, \lem{collapse of normal form} yields $\rep{P'} = \rep{Q'}$.
\item Let $\theta_X(\rep{P}) \goesto{\tau} P''$.
  Then $\theta_X(\rep{Q}) \goesto{\tau} Q''$ for some $Q''$ with $P'' \bis{r} Q''$.
  The process $P''$ must have the form $\theta_X(\rep{P'})$, and likewise $Q''=\theta_X(\rep{Q'})$.
  Hence $P'' \B Q''$.
\qed
\end{itemize}
\end{proof}

\begin{definition}{hnfrep}
Given a relevant $\CCSP$ process $P \in \IP^g$, let
\plat{$\widetilde P := \sum_{\{(\alpha,Q) \mid P \goesto{\scriptscriptstyle\alpha} Q\}}\alpha.\rep{Q}$}.
\end{definition}
Thus, $\widetilde P$ is defined like the head-normal form $\widehat P$ of $P \in\IP^g$,
except that all processes $Q$ reachable from $P$ by performing one transition are 
replaced by the normal form within their $\bis{r}$-equivalence class.\vspace{-1pt}
So $P \bis{\,} \widehat P \bis{r} \widetilde P$.
Note that $\rep{P} = \widetilde{\chi([P])}$ is provable through a single application of the axiom RDP\@.

The following step is the only one where the reactive approximation axiom (RA) is used.

\begin{proposition}{infinite}
Let $P,Q\in\IP^g$ be relevant. Then
$P \bis{r} Q \Rightarrow \textit{Ax} \vdash \widetilde P = \widetilde Q$.
\end{proposition}

\begin{proof}
Suppose $P \bis{r} Q$. Then $\widetilde P \bis{r} P \bis{r} Q \bis{r} \widetilde Q$.\vspace{1pt}
With Axiom RA it suffices to show that $\textit{Ax}_f \vdash \psi_X(\widetilde P) = \psi_X(\widetilde Q)$
for all $X\subseteq A$. So pick $X \subseteq A$. Let
$$\widetilde P = \sum_{i\in I} \alpha_i.P'_i + \sum_{j\in J} \rt.P''_j
\qquad\mbox{and}\qquad\widetilde Q = \sum_{k\in K} \beta_k.Q'_k + \sum_{h\in H} \rt.Q''_h$$
with $\alpha_j,\beta_k \mathbin\in A\cup\{\tau\}$ for all $i\mathbin\in I$ and $k\mathbin\in K$.
As for \thm{finite}, the following two claims are crucial.
\vspace{1ex}

\noindent
\textit{Claim 1:} For each $i \mathbin\in I$ there is a $k \mathbin\in K$ with $\alpha_i=\beta_k$
and $\textit{Ax} \vdash P'_i = Q'_k$.

\noindent
\textit{Claim 2:}  If $\I(P) \cap (X \cup\{\tau\}) = \emptyset$, then for each $j \mathbin\in J$ there is a $h \mathbin\in H$ with
$\textit{Ax} \vdash \theta_X(P''_j) = \theta_X(Q''_h)$.
\vspace{1ex}

\noindent
With these claims the proof proceeds exactly as the one of \thm{finite}.
\vspace{1ex}

\noindent
\textit{Proof of Claim 1:} Pick $i \mathbin\in I$. Then $\widetilde P \goesto{\alpha_i} P'_i$.
So  $\widetilde Q \goesto{\alpha_i} Q'$ for some $Q'$ with $P'_i \bis{r} Q'$.
Hence there is a $k \mathbin\in K$ with $\alpha_i=\beta_k$ and $Q'=Q'_k$.
The processes $P'_i$ and $Q'_k$ must have the form $\rep{P'}$ and $\rep{Q'}$ for some $P',Q'\in\IP^g$.
Hence, by \lem{collapse of normal form}, $P'_i = Q'_k$, and thus certainly $\textit{Ax} \vdash P'_i = Q'_k$.
\vspace{1ex}

\noindent
\textit{Proof of Claim 2:} Pick $j \mathbin\in J$. Then \plat{$\widetilde P \goesto{\rt} P''_j$}.
Since $\I(\widetilde P) \cap (X \cup\{\tau\}) = \emptyset$, there is a $Q''$ such that 
$\widetilde Q \goesto{\rt} Q''$ and $\theta_X(P''_j) \bis{r} \theta_X(Q'')$.
Hence there is a $h \mathbin\in H$ with $Q''=Q''_h$.
The processes $P''_j$ and $Q''_h$ have the form $\rep{P''}$ and $\rep{Q''}$ for some $P'',Q''\in\IP^g$.
So by \lem{bisimulation collapse}, $\theta_X(P''_j) \bis{\,} \theta_X(Q'')$.
The completeness of \textit{Ax} for strong bisimilarity (\thm{completeness thetaX}) now yields
$\textit{Ax} \vdash \theta_X(P''_j) = \theta_X(Q'')$.
\end{proof}

\begin{theorem}{completeness normal form}
Let $P\in\IP^g$ be relevant. Then $\textit{Ax} \vdash P = \rep{P}$.
\end{theorem}

\begin{proof}
Let $\textit{reach}(P)$ be the set of processes reachable from $P$.
Take a different variable $z_R$ for each $R \in \textit{reach}(P)$,
and define the recursive specification $\RS'$ by $V_{\RS'} := \{z_R\mid R \in \textit{reach}(P)\}$
and $$z_{R} = \sum_{R\goesto{\alpha} R'} \alpha.z_{R'}\;.$$
By construction, $R \bis{} \rec{x_{R} | \RS}$.
In fact, the symmetric closure of $\{(R,\rec{x_{R} | \RS}) \mid R \in \textit{reach}(P)\}$ is a strong bisimulation.
To establish \thm{completeness normal form} through an application of RSP, I show that both $P$ and
$\rep{P}$ are $x_P$-components of solutions of $\RS'$. So I show
\[\textit{Ax} \vdash R = \sum_{R\goesto{\alpha} R'} \alpha.R'
\qquad\mbox{and}\qquad
\textit{Ax} \vdash \rep{R} = \sum_{R\goesto{\alpha} R'} \alpha.\rep{R'}\]
for all $R \in \textit{reach}(P)$.
The first of these statements is a direct application of \pr{hnf}.
The second statement can be reformulated as $\textit{Ax} \vdash \rep{R} = \widetilde R$.
As remarked above, $\textit{Ax} \vdash \rep{R} = \widetilde{\chi([R])}$ through a single application of RDP\@.
Hence I need to show that $\textit{Ax} \vdash \widetilde{\chi([R])} = \widetilde R$.
Considering that $\chi([R]) \bis{r} R$, this is a consequence of \pr{infinite}.
\end{proof}

\begin{corollary}{completeness}
Let $P,Q\in\IP^g$ be relevant. Then $P \bis{r} Q \Rightarrow \textit{Ax} \vdash P = Q$.
\end{corollary}
\begin{proof}
Let $P \bis{r} Q$. Then $\rep{P}=\rep{Q}$ by \lem{collapse of normal form}, so
$\textit{Ax} \vdash P = \rep{P} = \rep{Q} = Q$.
\end{proof}

\subsection{Necessity of the axiom of choice}\label{sec:choice}

At first glance it may look like the above proof can be simplified so as to avoid using the axiom
of choice, namely by changing (\ref{normal form}) into
\[ R \goesto\alpha R' ~~\Leftrightarrow~~ \exists P\in R, P'\in R'.~ P \goesto\alpha P'\;.\]
However, this would make some processes $[P]$ infinitely branching, even when $P$ is finitely branching.
\fig{uncountable} shows an uncountable collection of strongly reactive bisimilar finitely branching processes.
Here each pair of a dashed $b$-transition and the dotted one right below it constitutes a design choice:
either the dashed or the dotted $b$-transition is present, but not both.
Since there is this binary choice for infinitely many pairs of $b$-transitions, this figure
represents an uncountable collection of processes. All of them are strongly reactive bisimilar,
because the $\rt$-transition will only be taken in an environment that blocks $b$.
In case $a$ is blocked as well, all the $a$-transitions from a state with an outgoing
$\tau$-transition can be dropped, and the difference between these processes disappears.
In case $a$ is allowed by the environment, all $b$ transitions can be dropped, and again 
the difference between these processes disappears. Hence the above alternative definition would
yield uncountably many outgoing $\rt$-transitions from the equivalence class of all these
processes. This would make it impossible to represent such a ``minimised'' process in $\CCSP$.

\begin{figure}\input{uncountable}
\hfill \raisebox{3ex}{\box\graph} \hfill\mbox{}\vspace{3ex}
\caption{An uncountable variety of strongly reactive bisimilar processes}
\label{fig:uncountable}
\end{figure}

\section{Concluding remarks}\label{sec:conclusion}

This paper laid the foundations of the proper analogue of strong bisimulation semantics for a
process algebra with time-outs. This makes it possible to specify systems in this setting and verify
their correctness properties. The addition of time-outs comes with considerable gains in expressive
power. An illustration of this is mutual exclusion.

As shown in \cite{GH15b}, it is fundamentally impossible to correctly specify mutual exclusion
protocols in standard process algebras, such as CCS \cite{Mi90ccs}, CSP \cite{BHR84,Ho85}, ACP
\cite{BW90,Fok00} or CCSP\@, unless the correctness of the specified protocol hinges on a fairness
assumption. The latter, in the view of \cite{GH15b}, does not provide an adequate solution, as
fairness assumptions are in many situations unwarranted and lead to false conclusions. In
\cite{EPTCS255.2} a correct process-algebraic rendering of mutual exclusion is given, but only after
making two important modifications to standard process algebra.  The first involves making a
justness assumption. Here \emph{justness} \cite{GH19} is an alternative to fairness, in some sense a
much weaker form of fairness---meaning weaker than weak fairness.  Unlike (strong or weak) fairness,
its use typically is warranted and does not lead to false conclusions.  The second modification is
the addition of a new construct---\emph{signals}---to CCS, or any other standard process algebra.
Interestingly, both modifications are necessary; just using justness, or just adding signals, is
insufficient. Bouwman \cite{Bou18,BLW20} points out that since the justness requirement was fairly new, and
needed to be carefully defined to describe its interaction with signals anyway, it is possible to
specify mutual exclusion without adding signals to the language at all, instead reformulating the
justness requirement in such a way that it effectively turns some actions into signals. Yet justness
is essential in all these approaches. This may be seen as problematic, because large parts of the
foundations of process algebra are incompatible with justness, and hence need to be thoroughly
reformulated in a justness-friendly way. This is pointed out in~\cite{vG19c}.

The addition of time-outs to standard process algebra makes it possible to specify mutual exclusion
without assuming justness! Instead, one should make the assumption called \emph{progress} in
\cite{GH19}, which is weaker than justness, uncontroversial, unproblematic, and made (explicitly or
implicitly) in virtually all papers dealing with issues like mutual exclusion.
This claim is substantiated in \cite{vG21b}.

Besides applications to protocol verification, future work includes adapting the work done here to a
form of reactive bisimilarity that abstracts from hidden actions, that is, to provide a counterpart
for process algebras with time-outs of, for instance, branching bisimilarity \cite{GW96}, weak
bisimilarity \cite{Mi90ccs} or coupled similarity \cite{PS92,vG93,BNP20}. Other topics worth
exploring are the extension to probabilistic processes, and especially the relations with timed
process algebras.
Davies \& Schneider in \cite{DS93}, for instance, added a construct with a quantified time-out to the process
algebra CSP \cite{BHR84,Ho85}, elaborating the timed model of CSP presented by Reed \& Roscoe in \cite{RR88}.

\paragraph{Acknowledgement.} $\!$My thanks to the CONCUR'20 and Acta Informatica referees for helpful feedback.

\bibliographystyle{eptcs}
\bibliography{../../../../Stanford/lib/abbreviations,../../../../Stanford/lib/new,../../../../Stanford/lib/dbase,glabbeek}

\appendix
\newpage

\section{Initials congruence}\label{initials congruence}

This appendix contains the proofs of two facts about initials equivalence I need in this paper,
namely that it is a full congruence for $\CCSP$, and that it is not affected by which processes are
substituted for variables whose free occurrences are guarded.

\begin{trivlist} \item[\hspace{\labelsep}\bf \thm{initials congruence}]
Initials equivalence is a full congruence for $\CCSP$.
\end{trivlist}

\begin{proof}
Let $\mathord{\B}\subseteq \IP \times \IP$ be the smallest relation satisfying
\begin{itemize} 
\item if $\RS$ and $\RS'$ are recursive specifications with $x \in V_\RS = V_{\RS'}$ and
      $\rec{x|\RS},\rec{x|\RS'}\in\IP$, such that $\RS_y =_I \RS'_y$ for all $y\in V_\RS$,
      then $\rec{x|\RS} \B \rec{x|\RS'}$,
\item if $P =_\I Q$, then $P \B Q$,
\item if $P \B Q$ and $\alpha\in A\cup\{\tau,\rt\}$, then $\alpha.P \B \alpha.Q$,
\item if $P_1 \B Q_1$ and $P_2 \B Q_2$, then $P_1 + P_2 \B Q_1 + Q_2$,
\item if $P_1 \B Q_1$, $P_2 \B Q_2$ and $S\subseteq A$, then $P_1\spar{S}P_2 \B Q_1\spar{S}Q_2$,
\item if $P \B Q$ and $I\subseteq A$, then $\tau_I(P) \B \tau_I(Q)$,
\item if $P \B Q$ and $\Rn \subseteq A \times A$, then $\Rn(P) \B \Rn(Q)$,
\item if $P \B Q$, $L\subseteq U \subseteq A$ and $X\subseteq A$, then $\theta_L^U(P) \B \theta_L^U(Q)$ and $\psi_X(P) \B \psi_X(Q)$,
\item if $\RS$ is a recursive specification with $z \in V_\RS$,
      and $\rho,\nu:\Var\setminus V_\RS \rightarrow\IP$ are substitutions satisfying $\rho(x) \B \nu(x)$ for all
      $x\in\Var\setminus V_\RS$, then $\rec{z|\RS}[\rho] \B \rec{z|\RS}[\nu]$.
\end{itemize}
A trivial structural induction on $E\in\IT$ (not using the first two clauses) shows that\\[5pt]
\mbox{}\hfill if $\rho,\nu:\Var \rightarrow\IP$ satisfy
$\rho(x) \mathrel\B \nu(x)$ for all $x\in\Var$, then $E[\rho] \mathbin\B E[\nu]$.\hfill ({\color{red}*})\\[5pt]
For $\RS$ a recursive specification and $\rho:\Var\setminus V_\RS \rightarrow\IP$, let $\rho_\RS: \Var\rightarrow\IP$ be the
closed substitution given by $\rho_\RS(x):= \rec{x|\RS}[\rho]$ if $x\in V_\RS$ and $\rho_\RS(x):=\rho(x)$ otherwise.
Then $\rec{E|\RS}[\rho] = E[\rho_\RS]$ for all $E\in\IT$.\linebreak[3]
Hence an application of ({\color{red}*}) with $\rho_\RS$ and $\nu_\RS$ yields that under the
conditions of the last clause for $\B$ above one even has 
$\rec{E|\RS}[\rho] \mathrel\B \rec{E|\RS}[\nu]$ for all expressions $E\in\IT$,\hfill (\$)\\
and likewise, in the first clause, $\rec{E|\RS} \B \rec{E|\RS'}$ for all $E\in\IT$ with variables from $V_\RS$.\hfill (\#)

It suffices to show that $P \B Q \Rightarrow P =_\I Q$, because then ${\B} = {=_\I}$,
and ({\color{red}*}) implies that $\B$ is a lean congruence.
Moreover, the clauses for $\B$ (not needing the last) then imply that $=_\I$ is a full congruence.
This I will do by induction on the \emph{stratum} $(\lambda_R,\kappa_R)$ of processes $R\in\IP$,
as defined in \Sec{stratification}.
So pick a stratum $(\lambda,\kappa)$ and assume that $P' \B Q' \Rightarrow P' =_\I Q'$ for all $P',Q'\in\IP$ with
$(\lambda_P,\kappa_P) < (\lambda,\kappa)$ and $(\lambda_Q,\kappa_Q) < (\lambda,\kappa)$.
I need to show that $P \B Q \Rightarrow P =_\I Q$ for all $P,Q\in\IP$ with
$(\lambda_P,\kappa_P) \leq (\lambda,\kappa)$ and $(\lambda_Q,\kappa_Q) \leq (\lambda,\kappa)$.

Because $=_I$ is symmetric, so is $\B$. Hence, it suffices to show that
$P \B Q \wedge P{\goesto\alpha} \Rightarrow Q{\goesto\alpha}$ for all $P,Q\in\IP$
with $(\lambda_P,\kappa_P), (\lambda_Q,\kappa_Q) \leq (\lambda,\kappa)$ and all $\alpha\in A \cup\{\tau\}$.
This I will do by structural induction on the proof $\pi$ of $P{\goesto{\alpha}}$ from the rules of
\tab{sos CCSP}. I make a case distinction based on the derivation of $P \B Q$.
So assume $P \B Q$, $(\lambda_P,\kappa_P), (\lambda_Q,\kappa_Q) \leq (\lambda,\kappa)$, and
$P{\goesto\alpha}$ with $\alpha\in A \cup\{\tau\}$.

\begin{itemize}
\item
Let $P=\rec{x|\RS}\in\IP$ and $Q=\rec{x|\RS'}\in \IP$ where
$\RS$ and $\RS'$ are recursive specifications with $x \in V_\RS = V_{\RS'}$,
such that $\RS_y =_\I \RS'_y$ for all $y\in V_\RS$, meaning that
for all $y\in V_\RS$ and $\sigma:V_\RS\rightarrow \IP$ one has
$\RS_y[\sigma] =_\I \RS'_y[\sigma]$.
\\
By \tab{sos CCSP} the transition $\rec{\RS_x|\RS}{\goesto{\alpha}}$ is provable by means of a strict subproof of $\pi$.
By (\#) above one has $\rec{\RS_x|\RS} \B \rec{\RS_x|\RS'}$.
So by induction $\rec{\RS_x|\RS'}{\goesto{\alpha}}$.
Since $\rec{\_\!\_\, | \RS'}$ is the application of a substitution of the form $\sigma:V_{\RS'} \rightarrow\IP$,
one has $\rec{\RS_x|\RS'} =_\I \rec{\RS'_x|\RS'}$.
Hence $\rec{\RS'_x|\RS'}{\mathbin{\goesto{\alpha}}}$.
By \tab{sos CCSP}, $Q \mathbin= \rec{x|\RS'} {\goesto{\alpha}}$.

\item The case $P =_\I Q$ is trivial.

\item Let $P = \beta.P^\dagger$ and $Q = \beta.Q^\dagger$ with $\beta\in A\cup\{\tau,\rt\}$ and $P^\dagger \B Q^\dagger$.
  Then $\alpha=\beta$ and $Q {\goesto{\alpha}}$.

\item Let $P = P_1 + P_2$ and $Q = Q_1 + Q_2$ with $P_1 \B Q_1$ and $P_2 \B Q_2$.
  I consider the first rule from \tab{sos CCSP} that could have been responsible for the derivation
  of $P{\goesto{\alpha}}$; the other proceeds symmetrically.
  So suppose that $P_1{\goesto{\alpha}}$. Then by induction $Q_1{\goesto{\alpha}}$.
  By the same rule, $Q{\goesto{\alpha}}$.

\item Let $P = P_1 \spar{S} P_2$ and $Q = Q_1 \spar{S} Q_2$ with $P_1 \B Q_1$ and $P_2 \B Q_2$.
  I consider the three rules from \tab{sos CCSP} that could have been responsible for the derivation
  of $P{\goesto{\alpha}}$.

First suppose that $\alpha \notin S$, and $P_1{\goesto{\alpha}}$.
By induction, $Q_1{\goesto{\alpha}}$.
Consequently, $Q_1 \spar{S} Q_2{\goesto{\alpha}}$.

Next suppose that $\alpha \mathbin\in S$, $P_1{\goesto{\alpha}}$ and $P_2{\goesto{\alpha}}$.
By induction, $Q_1{\goesto{\alpha}}$ and $Q_2 {\goesto{\alpha}}$.
So $Q_1 \spar{S} Q_2 {\goesto{\alpha}}$.

The remaining case proceeds symmetrically to the first.

\item Let $P = \tau_I(P^\dagger)$ and $Q= \tau_I(Q^\dagger)$ with $I\subseteq A$ and $P^\dagger \B Q^\dagger$.
Then $P^\dagger{\goesto{\beta}}$ and either $\beta = \alpha \notin I$, or $\beta\in I$ and $\alpha=\tau$.
By induction, $Q^\dagger{\goesto{\beta}}$.
Consequently, $Q = \tau_I(Q^\dagger){\goesto{\alpha}}$.

\item Let $P = \Rn(P^\dagger)$ and $Q= \Rn(Q^\dagger)$ with $\Rn \subseteq A \times A$ and $P^\dagger \B Q^\dagger$.
Then $P^\dagger{\goesto{\beta}}$ and either $(\beta,\alpha) \in \Rn$ or $\beta=\alpha = \tau$.
By induction, $Q^\dagger {\goesto{\beta}}$.
Consequently, $Q = \Rn(Q^\dagger){\goesto{\alpha}}$.

\item Let $P=\theta_L^U(P^\dagger)$, $Q=\theta_L^U(Q^\dagger)$ and $P^\dagger \B Q^\dagger$.
  Then $(\lambda_{P^\dagger},\kappa_{P^\dagger}) \mathbin< (\lambda,\kappa)$ and $(\lambda_{Q^\dagger},\kappa_{Q^\dagger}) \mathbin< (\lambda,\kappa)$,
  as remarked in \Sec{stratification}. So by induction $P^\dagger =_\I Q^\dagger$.
  (This is the only use of stratum induction.)

  Since $\theta_L^U(P^\dagger){\goesto{\alpha}}$, it must be that $P^\dagger{\goesto{\alpha}}$ and
  either $\alpha\in U\cup\{\tau\}$ or $P^\dagger {\ngoesto\beta}$ for all $\beta\in L\cup\{\tau\}$.
  In the latter case, $Q^\dagger {\ngoesto\beta}$ for all $\beta\mathbin\in L\mathop\cup\{\tau\}$.
  Moreover, $Q^\dagger{\goesto{\alpha}}$.
  So, in both cases, $Q=\theta_L^U(Q^\dagger){\goesto{\alpha}}$.

\item Let $P=\psi_X(P^\dagger)$, $Q=\psi_X(Q^\dagger)$ and $P^\dagger \B Q^\dagger$.
  Since $\psi_X(P^\dagger){\goesto{\alpha}}$, one has $P^\dagger{\goesto{\alpha}}$.
  By induction $Q^\dagger{\goesto{\alpha}}$.
  So $Q=\psi_X(Q^\dagger){\goesto{\alpha}}$.

\item
Let $P\mathbin=\rec{z|\RS}[\rho]\mathbin=\rec{z|\RS[\rho]}$ and
$Q\mathbin=\rec{z|\RS}[\nu]\mathbin=\rec{z|\RS[\nu]}$ where
$\RS$ is a recursive specification with $z \mathbin\in V_\RS$,
and $\rho,\nu:\Var\setminus V_\RS \rightarrow\IP$ satisfy $\rho(x) \B \nu(x)$ for all $x\mathbin\in\Var\setminus V_\RS$.
By \tab{sos CCSP} the transition $\rec{\RS_z|\RS[\rho]} {\goesto{\alpha}}$ is provable by means of a strict subproof of the
proof $\pi$ of $\rec{z|\RS}[\rho]{\goesto{\alpha}}$.
By (\$) above one has $\rec{\RS_z|\RS[\rho]} \B \rec{\RS_z|\RS[\nu]}$.
So by induction, $\rec{\RS_z|\RS[\nu]}{\goesto{a}}$.
By \tab{sos CCSP}, $Q = \rec{z|\RS[\nu]}{\goesto{\alpha}}$.
\qed
\end{itemize}
\end{proof}

\begin{trivlist} \item[\hspace{\labelsep}\bf \lem{8}]
Let $H\in\IT$ be guarded and have free variables from $W\subseteq \Var$ only, and let $\vec{P},\vec{Q}\in\IP^W$.
Then $\I(H[\vec{P}]) = \I(H[\vec{Q}])$.
\end{trivlist}

\begin{proof}
\lem{8} can be strengthened as follows.
\begin{quote}
Let $H\in\IT$ be such that all free occurrences of variables from $W\subseteq \Var$ in $H$ are guarded,
and let $\vec{P},\vec{Q}\in\IP^W$. Then $H[\vec{P}] =_\I H[\vec{Q}]$.
\end{quote}
The proof proceeds with structural induction on $H$.
\begin{itemize} 
\item Let $H = \rec{x|\RS}$, so that $H[\vec{P}] = \rec{x|\RS[\vec{P}^\dagger]}$, where $\vec{P}^\dagger$ is the
$W {\setminus} V_\RS$-tuple that is left of $\vec{P}$ after deleting the $y$-components, for $y\in V_\RS$,
and $H[\vec{Q}] = \rec{x|\RS[\vec{Q}^\dagger]}$.
For each $y \in V_\RS$, all free occurrences of variables from $W{\setminus} V_\RS$ in $\RS_y$ are guarded.
Thus, by induction, $\RS_y[\vec{P}^\dagger] =_\I \RS_y[\vec{Q}^\dagger]$.
Since $=_\I$ is a full congruence for $\CCSP$, it follows that 
$H[\vec{P}] = \rec{x|\RS[\vec{P}^\dagger]} =_\I \rec{x|\RS[\vec{Q}^\dagger]} = H[\vec{Q}]$.

\item Let $H \mathbin= \alpha.H'$ for some $\alpha\in Act$. Then $\I(H[\vec{P}]) \mathbin= \I(H[\vec{Q}])$ (namely
  $\emptyset$ if $\alpha\mathbin=\rt$ and $\{ \alpha \}$ otherwise).

\item Let $H = H_1 \spar{S} H_2$. Since all free occurrences of variables from $W\subseteq \Var$ in
  $H$ are guarded, so are those in $H_1$ and $H_2$. Thus, by induction, $H_1[\vec{P}] =_\I H_1[\vec{Q}]$
  and $H_2[\vec{P}] =_\I H_2[\vec{Q}]$. Since $=_\I$ is a full congruence for $\space{S}$ it follows
  that $H[\vec{P}] =_\I H[\vec{Q}]$.

\item The cases for all other operators go exactly like the case for $\spar{S}$.
\qed
\end{itemize}
\end{proof}

\section{Proofs of lemmas on \texorpdfstring{$\theta_X$}{theta_X} and strong bisimilarity from Section~\ref{sec:strong}}\label{strong proofs}

The following lemmas on the relation between $\theta_X$ and the other operators of $\CCSP$ deal with
strong bisimilarity, but are needed in the congruence proof for strong reactive bisimilarity.

\begin{lemma}{theta}
If $\I(Q)\cap (Y\cup\{\tau\}) = \emptyset$ then $\theta_Y(Q) \bis{} Q$.
\end{lemma}
\begin{proof}
This follows immediately from the operational rules for $\theta_Y$.
\end{proof}

\begin{trivlist} \item[\hspace{\labelsep}\bf \lem{theta-Par}]
If $P{\ngoesto\tau}$, $\I(P) \cap X \subseteq S$ and $Y\mathbin=X\setminus(S\setminus \I(P))$,
then $\theta_X(P \spar{S} Q) \bis{} \theta_X(P \spar{S}\theta_Y(Q))$.
\end{trivlist}
\begin{proof}
  Let $P \in \IP$ and $S,X,Y \subseteq A$ be as indicated in the lemma.
  Let $${\B} :={\bis{}\,} \cup \{(\theta_X(P \spar{S} Q), \theta_X(P \spar{S}\theta_Y(Q))) \mid Q \in \IP\}$$
  It suffices to show that the symmetric closure $\widetilde\B$ of $\B$ is a strong bisimulation.\\
  So let $R \mathrel{\widetilde\B} T$ and $R \goesto{\alpha} R'$ with $\alpha\in A \cup \{\tau,\rt\}$.
I have to find a $T'$ with $T \goesto{\alpha} T'$ and $R' \mathrel{\widetilde\B} T'$.
\begin{itemize}
\item
The case that $R \bis{} T$ is trivial.
\item
  Let $R = \theta_X(P \spar{S} Q)$ and $T = \theta_X(P \spar{S}\theta_Y(Q))$, for some $Q\in \IP$.

  First assume $\alpha=\tau$. Then $Q \goesto\tau Q'$ for some
  $Q'$ with $R' = \theta_X(P\spar{S}Q')$.
  Consequently, $T = \theta_X(P\spar{S}\theta_Y(Q)) \goesto\tau \theta_X(P\spar{S}\theta_Y(Q')) =: T'$ and $R' \B T'$.

  Now assume $\alpha \in A\cup\{\rt\}$. Then $P\spar{S}Q \goesto\alpha R'$.
  I first deal with the case that $\alpha\in X$, and
  consider the three rules from \tab{sos CCSP} that could have derived $P\spar{S}Q \goesto\alpha R'$.
\begin{itemize}
\item
  The case that $\alpha \notin S$ and $P \goesto{\alpha} P'$ cannot occur, because $\I(P) \cap X \subseteq S$.
\item
  Let $\alpha \in S$, $P \goesto{\alpha} P'$, $Q \goesto{\alpha} Q'$ and $R' =  P' \spar{S} Q'$.
  Then $\alpha \in \I(P)$, so $\alpha \notin S\setminus \I(P)$ and thus $\alpha \in Y$.
  Hence $\theta_Y(Q) \goesto{\alpha} Q'$.
  Now $T = \theta_X(P\spar{S}\theta_Y(Q)) \goesto\alpha P'\spar{S}Q' = R'$.
\item
  Let $\alpha \mathbin{\notin} S$, $Q \goesto{\alpha} Q'$ and $R'\mathbin=  P \spar{S} Q'$.
  Then $\alpha \mathbin\in Y$, so $\theta_Y(Q) \goesto{\alpha} Q'$.
  Therefore, $P\spar{S}\theta_Y(Q) \goesto\alpha P\spar{S}Q'$
  and thus $T =  \theta_X(P\spar{S}\theta_Y(Q)) \goesto\alpha P\spar{S}Q' = R'$.
\end{itemize}
  Finally, assume $\alpha \in (A\cup\{\rt\})\setminus X$.
  In that case $P\spar{S}Q \ngoesto\beta$ for all $\beta\in X\cup \{\tau\}$.
  Therefore, $Q \ngoesto\beta$ for all $\beta\in (X\setminus S)\cup \{\tau\}$, and for all
  $\beta\in X\cap S \cap \I(P)$, and thus for all $\beta\in Y\cup \{\tau\}$.
  By \lem{theta}, $\theta_Y(Q) =_\I Q$, and hence $P\spar{S}\theta_Y(Q) \ngoesto\beta$ for all
  $\beta\in X\cup \{\tau\}$.
  Again, I consider the three rules from \tab{sos CCSP} that could have derived $P\spar{S}Q \goesto\alpha R'$.
\begin{itemize}
\item
  Let $\alpha \mathbin{\notin} S$, $P \goesto{\alpha} P'$ and $R'\mathbin=  P' \spar{S} Q$.
  Then $P\spar{S}\theta_Y(Q) \goesto\alpha P'\spar{S}\theta_Y(Q)$
  and thus $T =  \theta_X(P\spar{S}\theta_Y(Q)) \goesto\alpha P'\spar{S}\theta_Y(Q) =: T'$.
  By \lem{theta}, $\theta_Y(Q) \bis{} Q$. Since $\bis{}$ is a congruence for $\spar{S}$, it follows
  that $R' =  P' \spar{S} Q \bis{} P' \spar{S} \theta_Y(Q) = T'$.
\item
  Let $\alpha \in S$, $P \goesto{\alpha} P'$, $Q \goesto{\alpha} Q'$ and $R' =  P' \spar{S} Q'$.
  Then $\theta_Y(Q) \goesto{\alpha} Q'$ and therefore
  $P\spar{S}\theta_Y(Q) \goesto\alpha P'\spar{S}Q'$ and
  $T = \theta_X(P\spar{S}\theta_Y(Q)) \goesto\alpha P'\spar{S}Q' = R'$.
\item
  Let $\alpha \mathbin{\notin} S$, $Q \goesto{\alpha} Q'$ and $R'\mathbin=  P \spar{S} Q'$.
  Then $\theta_Y(Q) \goesto{\alpha} Q'$, so $P\spar{S}\theta_Y(Q) \goesto\alpha P\spar{S}Q'$
  and thus $T =  \theta_X(P\spar{S}\theta_Y(Q)) \goesto\alpha P\spar{S}Q' = R'$.
\end{itemize}
\item
  Let $R \mathbin= \theta_X(P \spar{S}\theta_Y(Q))$ and $T \mathbin= \theta_X(P \spar{S} Q)$, for some $Q\mathbin\in \IP$.

  First assume $\alpha=\tau$. Then $Q \goesto\tau Q'$ for some
  $Q'$ with $R' = \theta_X(P\spar{S}\theta_Y(Q'))$.
  Consequently, $T = \theta_X(P\spar{S}Q) \goesto\tau \theta_X(P\spar{S} Q') =: T'$ and $R' \mathrel{\widetilde\B} T'$.

  Now assume $\alpha \in A\cup\{\rt\}$. Then $P\spar{S}\theta_Y(Q) \goesto\alpha R'$ and either $\alpha \in X$
  or $P\spar{S}\theta_Y(Q) \ngoesto\beta$ for all $\beta\in X\cup \{\tau\}$.
  In the latter case one obtains  $\theta_Y(Q) \ngoesto\beta$ for all $\beta\in Y\cup \{\tau\}$ (as
  above), and thus $Q \ngoesto\beta$ for all $\beta\in Y\cup \{\tau\}$, that is,
  $\I(Q)\cap (Y\cup\{\tau\}) = \emptyset$. Furthermore, this implies that
  $P\spar{S}Q \ngoesto\beta$ for all $\beta\in X\cup \{\tau\}$.

  I consider the three rules from \tab{sos CCSP} that could have derived $P\spar{S}Q \goesto\alpha R'$.
\begin{itemize}
\item
  Let $\alpha \mathbin{\notin} S$, $P \goesto{\alpha} P'$ and $R'\mathbin=  P' \spar{S} \theta_Y(Q)$.
  Then $a \notin X$, because $\I(P) \cap X \subseteq S$.\\
  Hence $P\spar{S}\theta_Y(Q) \ngoesto\beta$ for all $\beta\in X\cup \{\tau\}$, so
  $\I(Q)\cap (Y\cup\{\tau\}) = \emptyset$.\\
  Now $T = \theta_X(P\spar{S}Q) \goesto\alpha P'\spar{S}Q =: T'$ and $R' \bis{}\, T'$, using \lem{theta}.
\item
  Let $\alpha \in S$, $P \goesto{\alpha} P'$, $\theta_Y(Q) \goesto{\alpha} Q'$ and $R' =  P' \spar{S} Q'$.
  Then $Q \goesto{\alpha} Q'$.\\
  Hence $P\spar{S}Q \goesto\alpha P'\spar{S}Q'$ and thus $T = \theta_X(P\spar{S}Q) \goesto\alpha P'\spar{S}Q' = R'$.
\item
  Let $\alpha \notin S$, $\theta_Y(Q) \goesto{\alpha} Q'$ and $R'=  P \spar{S} Q'$.
  Then $Q \goesto{\alpha} Q'$.\\
  Consequently, $P\spar{S}Q \goesto\alpha P\spar{S}Q'$
  and thus $T =  \theta_X(P\spar{S}Q) \goesto\alpha P\spar{S}Q' = R'$.
\qed
\end{itemize}
\end{itemize}
\end{proof}

\begin{trivlist} \item[\hspace{\labelsep}\bf \lem{theta-tau}]
$\theta_X(\tau_I(P)) \bis{} \theta_X(\tau_I(\theta_{X\cup I}(P)))$.
\end{trivlist}
\begin{proof}
  For given $X$ and $I$, let
  ${\B} :=\textit{Id} \cup \{(\theta_X(\tau_I(P)), \theta_X(\tau_I(\theta_{X\cup I}(P)))) \mid P \in \IP\}$.
  It suffices to show that the symmetric closure $\widetilde\B$ of $\B$ is a strong bisimulation.
  So let $R \mathrel{\widetilde\B} T$ and $R \goesto{\alpha} R'$ with $\alpha\in A \cup \{\tau,\rt\}$.
  I have to find a $T'$ with $T \goesto{\alpha} T'$ and $R' \mathrel{\widetilde\B} T'$.
\begin{itemize}
\item
The case that $R = T$ is trivial.
\item
  Let $R = \theta_X(\tau_I(P))$ and $T = \theta_X(\tau_I(\theta_{X\cup I}(P)))$, for some $P\in \IP$.

  First assume $\alpha=\tau$. Then $\tau_I(P)\goesto\tau R''$ for some $R''$ such that $R' = \theta_X(R'')$.
  Therefore, $P \goesto\beta P'$ for some $\beta\in I \cup\{\tau\}$ and some $P'$ with $R''=\tau_I(P')$.
  In case $\beta=\tau$, it turns out that
  $T = \theta_X(\tau_I(\theta_{X\cup I}(P))) \goesto\tau \theta_X(\tau_I(\theta_{X\cup I}(P'))) =:T'$.
  Moreover, $R' \B T'$.
  In case $\beta\in I$, $\theta_{X\cup I}(P) \goesto\beta P'$,
  so $\tau_I(\theta_{X\cup I}(P)) \goesto\tau \tau_I(P')$
  and $T \mathbin= \theta_X(\tau_I(\theta_{X\cup I}(P))) \goesto\tau \theta_X(\tau_I(P')) = R'$.

  Now assume $\alpha \in A \cup\{\rt\}$. Then $\tau_I(P)\goesto\alpha R'$ and either $\alpha\in X$ or
  $\tau_I(P) \ngoesto\beta$ for all $\beta\in X\cup \{\tau\}$.
  It follows that $\alpha\notin I$ and $P \goesto\alpha P'$ for some $P'$ with $R'=\tau_I(P')$.
  Moreover, in case $\alpha\notin X$ one has $P \ngoesto\beta$ for all $\beta\in X\cup I \cup \{\tau\}$,
  and hence also $\theta_{X\cup I}(P)\ngoesto\beta$ for all $\beta\in X\cup I \cup \{\tau\}$,
  and thus $\tau_I(\theta_{X\cup I}(P)) \ngoesto\beta$ for all $\beta\in X\cup \{\tau\}$.
  Now $\theta_{X\cup I}(P) \goesto\alpha P'$, 
  so $\tau_I(\theta_{X\cup I}(P)) \goesto\alpha \tau_I(P')$ and thus
  $T = \theta_X(\tau_I(\theta_{X\cup I}(P))) \goesto\alpha \tau_I(P') = R'$.

\item
  Let $R = \theta_X(\tau_I(\theta_{X\cup I}(P)))$ and $T = \theta_X(\tau_I(P))$, for some $P\in \IP$.

  First assume $\alpha=\tau$. Then $\tau_I(\theta_{X\cup I}(P))\goesto\tau R''$ for some $R''$ such that $R' = \theta_X(R'')$.
  Therefore, $\theta_{X\cup I}(P) \goesto\beta P'$ for some $\beta\in I \cup\{\tau\}$ and some $P'$ with $R''=\tau_I(P')$.
  In case $\beta=\tau$, it turns out that $P \goesto\tau P''$ for some $P''$ such that $P' \mathbin= \theta_{X\cup I}(P'')$.
  So $T = \theta_X(\tau_I(P)) \goesto\tau \theta_X(\tau_I(P'')) =:T'\!$, and $R' \mathrel{\widetilde\B} T'$.
  In case $\beta\in I$, one has $P \goesto\beta P'$, so $\tau_I(P)\goesto\tau \tau_I(P')$ and
  $T = \theta_X(\tau_I(P)) \goesto\tau \theta_X(\tau_I(P')) = R'$.

  Now assume $\alpha \in A \cup\{\rt\}$.
  Then $\tau_I(\theta_{X\cup I}(P))\goesto\alpha R'$,
  so $\alpha \notin I$ and $\theta_{X\cup I}(P) \goesto\alpha P'$ for some $P'$ such that $R' = \tau_I(P')$.
  Thus $P \goesto\alpha P'$ and either $\alpha\in X$ or $P \ngoesto\beta$ for all $\beta\in X\cup I \cup \{\tau\}$.
  In the latter case $\tau_I(P) \ngoesto\beta$ for all $\beta\in X \cup \{\tau\}$.
  Now $\tau_I(P) \goesto\alpha \tau_I(P')$ and consequently
  $T = \theta_X(\tau_I(P)) \goesto\alpha \tau_I(P') = R'$.
\qed
\end{itemize}
\end{proof}

\begin{trivlist} \item[\hspace{\labelsep}\bf \lem{theta-R}]
$\theta_X(\Rn(P)) \bis{} \theta_X(\Rn(\theta_{\Rn^{-1}(X)}(P)))$.
\end{trivlist}
\begin{proof}
  For given $X\subseteq A$ and $\Rn\subseteq A\times A$, let
  ${\B} :=\textit{Id} \cup \{(\theta_X(\Rn(P)), \theta_X(\Rn(\theta_{\Rn^{-1}(X)}(P)))) \mid P \in \IP\}$.
  It suffices to show that the symmetric closure $\widetilde\B$ of $\B$ is a strong bisimulation.
  So let $R \mathrel{\widetilde\B} T$ and $R \goesto{\alpha} R'$ with $\alpha\in A \cup \{\tau,\rt\}$.
  I have to find a $T'$ with $T \goesto{\alpha} T'$ and $R' \mathrel{\widetilde\B} T'$.
\begin{itemize}
\item
The case that $R = T$ is trivial.
\item
Let $R= \theta_X(\Rn(P))$ and $T=\theta_X(\Rn(\theta_{\Rn^{-1}(X)}(P)))$, for some $P \in \IP$.

First assume $\alpha=\tau$. Then $P\goesto\tau P'$ for some $P'$ such that $R' = \theta_X(\Rn(P'))$.\\
Hence $T \mathbin= \theta_X(\Rn(\theta_{\Rn^{-1}(X)}(P))) \goesto\tau \theta_X(\Rn(\theta_{\Rn^{-1}(X)}(P'))) =: T'$,
and $R' \B T'$.

Now assume $\alpha \mathbin\in A \cup\{\rt\}$. Then $\Rn(P)\goesto{\alpha} R'$, and
either $\alpha\mathbin\in X$ or $\Rn(P) \ngoesto\beta$ for all $\beta\mathbin\in X \cup \{\tau\}$.
In the latter case, $P \ngoesto\beta$ for all $\beta\mathbin\in \Rn^{-1}(X) \cup \{\tau\}$.
Moreover, $P \goesto{\gamma} P'$, for some $\gamma$ with $\gamma=\rt=\alpha$ or $(\gamma,\alpha)\in\Rn$,
and some $P'$ with $R' = \Rn(P')$. In case $\alpha \in X$, one has $\gamma \in \Rn^{-1}(X)$.
Therefore, $\theta_{\Rn^{-1}(X)}(P)  \goesto{\gamma} P'$, and thus
$\Rn(\theta_{\Rn^{-1}(X)}(P)) \goesto\alpha \Rn(P')$.

Either $\alpha\mathbin{\in} X$ or
$\theta_{\Rn^{-1}(X)}(P) \ngoesto\beta$ for all $\beta\mathbin\in \Rn^{-1}(X) \cup \{\tau\}$,
in which case $\Rn(\theta_{\Rn^{-1}(X)}(P)) \ngoesto\beta$ for all $\beta\mathbin\in X \cup \{\tau\}$.
Consequently, $T \mathbin= \theta_X(\Rn(\theta_{\Rn^{-1}(X)}(P))) \goesto\alpha  \Rn(P') = R'$.

\item Let $R= \theta_X(\Rn(\theta_{\Rn^{-1}(X)}(P)))$ and $T = \theta_X(\Rn(P))$, for some $P \in \IP$.

First assume $\alpha=\tau$. Then $P\goesto\tau P'$ for some $P'$ such that \plat{$R' = \theta_X(\Rn(\theta_{\Rn^{-1}(X)}(P')))$}.\\
Hence $T \mathbin= \theta_X(\Rn(P)) \goesto\tau \theta_X(\Rn(P')) =: T'$, and \plat{$R' \mathrel{\widetilde\B} T'$.}

Now assume $\alpha \mathbin\in A \cup\{\rt\}$.
Then $\Rn(\theta_{\Rn^{-1}(X)}(P))\goesto{\alpha} R'$ and either $\alpha\in X$ or 
$\Rn(\theta_{\Rn^{-1}(X)}(P)) \ngoesto\beta$ for all $\beta\mathbin\in X \cup \{\tau\}$.
Therefore, \plat{$\theta_{\Rn^{-1}(X)}(P) \goesto{\gamma} P'$} for some $\gamma$ with $\gamma=\rt=\alpha$
or $(\gamma,\alpha)\in\Rn$, and some $P'$ such that $R' = \Rn(P')$.
Hence $P \goesto{\gamma} P'$, and thus $\Rn(P) \goesto{\alpha} \Rn(P')$.
In case $\alpha\notin X$, one has $\theta_{\Rn^{-1}(X)}(P) \ngoesto\beta$ for all $\beta\mathbin\in \Rn^{-1}(X) \cup \{\tau\}$,
and thus $P \ngoesto\beta$ for all $\beta\mathbin\in \Rn^{-1}(X) \cup \{\tau\}$, so
\plat{$\Rn(P)\ngoesto\beta$} for all $\beta\mathbin\in X \cup \{\tau\}$.
Hence $T = \theta_X(\Rn(P)) \goesto{\alpha} \Rn(P') = R'$.
\qed
\end{itemize}
\end{proof}

\section{Reducing Strong Reactive Bisimilarity to Strong Bisimilarity}\label{reduction}

Pohlmann \cite{Po21} introduces unary operators $\vartheta$ and $\vartheta_X$ for $X\subseteq A$
that model placing their argument process in an environment that is triggered to change, or allows
exactly the actions in $X$, respectively. Although inspired by my operators $\theta_X$ from
\Sec{timeout bisimulations},\footnote{Pohlmann~\cite{Po21} follows the original, 2020, version of
  this paper; this appendix was added in September 2021.} their semantics is different,
and given by the following structural operational rules (for all $X \subseteq A$).
\[\begin{array}{c@{\qquad}c}
\displaystyle\frac{x \goesto{\tau} y}{\vartheta(x) \goesto{\tau} \vartheta(y)} &
\displaystyle\frac{}{\vartheta(x) \goesto{\epsilon_X} \vartheta_X(x)}
 \\[1.5em]
\displaystyle\frac{x \goesto{a} y}{\vartheta_X(x) \goesto{a} \vartheta(y)}~(a \in X) &
\displaystyle\frac{x \goesto{\tau} y}{\vartheta_X(x) \goesto{\tau} \vartheta_X(y)}
 \\[1.5em]
\displaystyle\frac{x {\ngoesto\alpha}~\mbox{for all}~\alpha\mathbin\in X\cup\{\tau\}}
{\vartheta_X(x) \goesto{\rt_\epsilon} \vartheta(y)}
&
\displaystyle\frac{x \goesto{\rt} y \quad x {\ngoesto\alpha}~\mbox{for all}~\alpha\mathbin\in X\cup\{\tau\}}
{\vartheta_X(x) \goesto{\rt} \vartheta_X(y)}
\\[1.5em]
\end{array}\]
Here the actions $\rt_\epsilon\notin A$ and $\epsilon_X\notin A$ for $X \subseteq A$ are generated by the new
operators, but may not be used by processes substituted for their arguments $x$.
They model a time-out action taken by the environment, and the stabilisation of an environment into
one that allows exactly the set of actions $X$, respectively.

These rules mirror the clauses of \df{reactive bisimilarity} of a strong reactive bisimulation.
\setlist[itemize]{nosep}
\begin{itemize}
\item $\tau$-transitions can be performed regardless of the environment,
\item triggered environments can stabilise into arbitrary stable
     environments $X$ for $X \subseteq A$,
\item allowed visible transitions can be performed and can trigger a change
     in the environment,
\item $\tau$-transitions cannot be observed by the environment and hence cannot
     trigger a change,
\item if the underlying system is idle, the environment may time-out and
     become triggered to change,
\item if the underlying system is idle, it can perform a $\rt$-transition,
     not observed by the environment.
\end{itemize}
The main result from \cite{Po21} reduces strong reactive bisimilarity to strong bisimilarity:

\begin{theorem}{reduction}
Let $P,\!Q\mathbin\in\IP\!$, $X\mathbin\subseteq A$. Then $P \bis{r} Q$ iff $\vartheta(P) \bis{} \vartheta(Q)$,
and $P \rbis{X}{r} Q$ iff $\vartheta_X(P) \bis{} \vartheta_X(Q)$.
\end{theorem}
\begin{proof}
If $\R$ is a strong reactive bisimulation, then
\[\B := \{(\vartheta(P),\vartheta(Q))\mid (P,Q)\in\R\} \cup \{(\vartheta(P),\vartheta(Q))\mid (P,X,Q)\in\R\}\]
is a strong bisimulation. Moreover,
\[\R := \{(P,Q) \mid \vartheta(P) \bis{} \vartheta(Q)\} \cup \{(P,X,Q) \mid \vartheta_X(P) \bis{} \vartheta_X(Q)\}\]
is a strong reactive bisimulation. Both statements follows directly from the definitions, and they
imply the theorem. This proof stems from \cite{Po21}, where it is formalised in Isabelle. 
\end{proof}
Another notable result from \cite{Po21} is a function $\varsigma$ that turns any formula $\phi$ from
my extension of the Hennessy-Milner logic into a formula $\varsigma(\phi)$ in the regular Hennessy-Milner logic,
such that $P \models \phi$ iff $\vartheta(P) \models \varsigma(\phi)$
and $P \models_X \phi$ iff $\vartheta_X(P) \models \varsigma(\phi)$.

Interestingly, the operators $\vartheta$ and $\vartheta_X$ from \cite{Po21} can be expressed in
terms of (fairly) standard process algebra operators. Define the universal environment $\E$  as
the recursive specification
\[ \{ U = \sum_{X\subseteq A} \epsilon_X . X \} ~~\cup~~ \{X = \rt_\epsilon.U + \sum_{a \in X} a.U \mid X \subseteq A\}.\]
In case $A$ is infinite, this requires an infinite choice operator $\sum$, which was not included in
the syntax of CCSP$_\rt$ used in \Sec{ccsp}.
Here $V_\E= \{U\}\cup \{X \mid X \subseteq A\}$ are the bound variables of $\E$.
The process $\rec{U|\E}$ denotes an environment that is triggered to change, and
$\rec{X|\E}$ one that allows exactly the actions in $X$.
The only actions that $\rec{U|\E}$ can do are stabilising into any $\rec{X|\E}$.
The process $\rec{X|\E}$ can either synchronise on any action $a \in X$ or perform a time-out,
in both cases returning to the state $\rec{U|\E}$.

If we now drop the negative premises from the structural operational rules of the operators
$\vartheta_X$, and add a rule $\frac{x \goesto{\rt} y}{\vartheta(x) \goesto{\rt} \vartheta(y)}$,
then $\vartheta(P) \bis{}  \rec{U|\E} \spar{A} P$ and
$\vartheta_X(P) \bis{} \rec{X|\E} \spar{A} P$. Here the operator $\|_A$ enforces synchronisation on
all visible actions $a\in A$, although actions $\epsilon_X$ and $\rt_\epsilon$ can occur when the
environment is ready do do them, and actions $\tau$ and $\rt$ can be triggered by just the process $P$.
Checking strong bisimilarity between $\vartheta(P)$ and $\rec{U|\E} \spar{A} P$, and
between $\vartheta_X(P)$ and $\rec{X|\E} \spar{A} P$, is straightforward.

To obtain the real process $\vartheta(P)$ from $\rec{U|\E} \spar{A} P$, or $\vartheta_X(P)$ from
$\rec{X|\E} \spar{A} P$, all one has to do is to inhibit any $\rt$- or $\rt_\epsilon$-transition
when a transition with a label in $A \cup \{\tau\} \cup \{\epsilon_X\mid X \subseteq A\}$ is
possible.  This can be achieved with the priority operator of Baeten, Bergstra \& Klop \cite{BBK86}.
This unary operator $\Theta$ is parametrised by a partial order $<$ on the set of actions, the
\emph{priority} order, and passes through a transition of its argument process only if no transition
with a higher priority is possible. Its operational semantic is given by\vspace{-1ex}
\[\displaystyle\frac{x \goesto{\alpha} y \quad x {\ngoesto\beta}~\mbox{for all}~\beta > \alpha}
{\Theta(x) \goesto{\alpha} \Theta(y)}\;.\]
For the present application I take
${<} := \{(\rt,\alpha),(\rt_\epsilon,\alpha)\mid \alpha \in Act{\setminus}\{\rt,\rt_\epsilon\}\}$,
thus giving $\rt$ and $\rt_\epsilon$ a lower priority than all other actions.
This yields the desired properties
\[\vartheta(P) \bis{} \Theta(\rec{U|\E} \spar{A} P) \qquad\mbox{and}\qquad
\vartheta_X(P) \bis{} \Theta(\rec{X|\E} \spar{A} P) \;.\]
\end{document}

%% file: L3.tex
\expandafter\ifx\csname graph\endcsname\relax
   \csname newbox\expandafter\endcsname\csname graph\endcsname
\fi
\ifx\graphtemp\undefined
  \csname newdimen\endcsname\graphtemp
\fi
\expandafter\setbox\csname graph\endcsname
 =\vtop{\vskip 0pt\hbox{%
\pdfliteral{
q [] 0 d 1 J 1 j
0.576 w
0.576 w
54 -20.016 m
54 -22.242812 52.194812 -24.048 49.968 -24.048 c
47.741188 -24.048 45.936 -22.242812 45.936 -20.016 c
45.936 -17.789188 47.741188 -15.984 49.968 -15.984 c
52.194812 -15.984 54 -17.789188 54 -20.016 c
h q 0.5 g
B Q
20.016 -59.976 m
20.016 -65.503266 15.535266 -69.984 10.008 -69.984 c
4.480734 -69.984 0 -65.503266 0 -59.976 c
0 -54.448734 4.480734 -49.968 10.008 -49.968 c
15.535266 -49.968 20.016 -54.448734 20.016 -59.976 c
S
Q
}%
    \graphtemp=.5ex
    \advance\graphtemp by 0.833in
    \rlap{\kern 0.139in\lower\graphtemp\hbox to 0pt{\hss $P$\hss}}%
\pdfliteral{
q [] 0 d 1 J 1 j
0.576 w
0.072 w
q 0 g
51.768 -8.784 m
49.968 -15.984 l
48.168 -8.784 l
51.768 -8.784 l
B Q
0.576 w
49.968 0 m
49.968 -8.784 l
S
0.072 w
q 0 g
23.4 -49.104 m
17.064 -52.92 l
20.88 -46.584 l
23.4 -49.104 l
B Q
0.576 w
47.16 -22.824 m
22.176 -47.808 l
S
Q
}%
    \graphtemp=.5ex
    \advance\graphtemp by 0.526in
    \rlap{\kern 0.446in\lower\graphtemp\hbox to 0pt{\hss $b$~~~~~\hss}}%
\pdfliteral{
q [] 0 d 1 J 1 j
0.576 w
54 -59.976 m
54 -62.202812 52.194812 -64.008 49.968 -64.008 c
47.741188 -64.008 45.936 -62.202812 45.936 -59.976 c
45.936 -57.749188 47.741188 -55.944 49.968 -55.944 c
52.194812 -55.944 54 -57.749188 54 -59.976 c
h q 0.5 g
B Q
0.072 w
q 0 g
51.768 -48.816 m
49.968 -56.016 l
48.168 -48.816 l
51.768 -48.816 l
B Q
0.576 w
49.968 -23.976 m
49.968 -48.816 l
S
Q
}%
    \graphtemp=.5ex
    \advance\graphtemp by 0.556in
    \rlap{\kern 0.694in\lower\graphtemp\hbox to 0pt{\hss ~~~~$\rt$\hss}}%
\pdfliteral{
q [] 0 d 1 J 1 j
0.576 w
54 -100.008 m
54 -102.234812 52.194812 -104.04 49.968 -104.04 c
47.741188 -104.04 45.936 -102.234812 45.936 -100.008 c
45.936 -97.781188 47.741188 -95.976 49.968 -95.976 c
52.194812 -95.976 54 -97.781188 54 -100.008 c
h q 0.5 g
B Q
0.072 w
q 0 g
51.768 -88.776 m
49.968 -95.976 l
48.168 -88.776 l
51.768 -88.776 l
B Q
0.576 w
49.968 -64.008 m
49.968 -88.776 l
S
Q
}%
    \graphtemp=.5ex
    \advance\graphtemp by 1.111in
    \rlap{\kern 0.694in\lower\graphtemp\hbox to 0pt{\hss ~~~~$\tau$\hss}}%
\pdfliteral{
q [] 0 d 1 J 1 j
0.576 w
20.016 -100.008 m
20.016 -105.535266 15.535266 -110.016 10.008 -110.016 c
4.480734 -110.016 0 -105.535266 0 -100.008 c
0 -94.480734 4.480734 -90 10.008 -90 c
15.535266 -90 20.016 -94.480734 20.016 -100.008 c
S
Q
}%
    \graphtemp=.5ex
    \advance\graphtemp by 1.389in
    \rlap{\kern 0.139in\lower\graphtemp\hbox to 0pt{\hss $Q$\hss}}%
\pdfliteral{
q [] 0 d 1 J 1 j
0.576 w
0.072 w
q 0 g
23.4 -89.136 m
17.064 -92.952 l
20.88 -86.544 l
23.4 -89.136 l
B Q
0.576 w
47.16 -62.856 m
22.176 -87.84 l
S
Q
}%
    \graphtemp=.5ex
    \advance\graphtemp by 1.082in
    \rlap{\kern 0.446in\lower\graphtemp\hbox to 0pt{\hss $a$~~~~~\hss}}%
\pdfliteral{
q [] 0 d 1 J 1 j
0.576 w
20.016 -139.968 m
20.016 -145.495266 15.535266 -149.976 10.008 -149.976 c
4.480734 -149.976 0 -145.495266 0 -139.968 c
0 -134.440734 4.480734 -129.96 10.008 -129.96 c
15.535266 -129.96 20.016 -134.440734 20.016 -139.968 c
S
Q
}%
    \graphtemp=.5ex
    \advance\graphtemp by 1.944in
    \rlap{\kern 0.139in\lower\graphtemp\hbox to 0pt{\hss $R$\hss}}%
\pdfliteral{
q [] 0 d 1 J 1 j
0.576 w
0.072 w
q 0 g
23.4 -129.096 m
17.064 -132.912 l
20.88 -126.576 l
23.4 -129.096 l
B Q
0.576 w
47.16 -102.816 m
22.176 -127.872 l
S
Q
}%
    \graphtemp=.5ex
    \advance\graphtemp by 1.637in
    \rlap{\kern 0.446in\lower\graphtemp\hbox to 0pt{\hss $b$~~~~~\hss}}%
\pdfliteral{
q [] 0 d 1 J 1 j
0.576 w
100.008 -139.968 m
100.008 -145.495266 95.527266 -149.976 90 -149.976 c
84.472734 -149.976 79.992 -145.495266 79.992 -139.968 c
79.992 -134.440734 84.472734 -129.96 90 -129.96 c
95.527266 -129.96 100.008 -134.440734 100.008 -139.968 c
S
Q
}%
    \graphtemp=.5ex
    \advance\graphtemp by 1.944in
    \rlap{\kern 1.250in\lower\graphtemp\hbox to 0pt{\hss $S$\hss}}%
\pdfliteral{
q [] 0 d 1 J 1 j
0.576 w
0.072 w
q 0 g
79.128 -126.576 m
82.944 -132.912 l
76.536 -129.096 l
79.128 -126.576 l
B Q
0.576 w
52.848 -102.816 m
77.832 -127.872 l
S
Q
}%
    \graphtemp=.5ex
    \advance\graphtemp by 1.637in
    \rlap{\kern 0.943in\lower\graphtemp\hbox to 0pt{\hss ~~~~~$a$\hss}}%
\pdfliteral{
q [] 0 d 1 J 1 j
0.576 w
94.032 -59.976 m
94.032 -62.202812 92.226812 -64.008 90 -64.008 c
87.773188 -64.008 85.968 -62.202812 85.968 -59.976 c
85.968 -57.749188 87.773188 -55.944 90 -55.944 c
92.226812 -55.944 94.032 -57.749188 94.032 -59.976 c
h q 0.5 g
B Q
0.072 w
q 0 g
83.376 -50.832 m
87.192 -57.168 l
80.784 -53.352 l
83.376 -50.832 l
B Q
0.576 w
52.848 -22.824 m
82.08 -52.056 l
S
Q
}%
    \graphtemp=.5ex
    \advance\graphtemp by 0.556in
    \rlap{\kern 0.972in\lower\graphtemp\hbox to 0pt{\hss ~~~~~$\rt$\hss}}%
\pdfliteral{
q [] 0 d 1 J 1 j
0.576 w
94.032 -100.008 m
94.032 -102.234812 92.226812 -104.04 90 -104.04 c
87.773188 -104.04 85.968 -102.234812 85.968 -100.008 c
85.968 -97.781188 87.773188 -95.976 90 -95.976 c
92.226812 -95.976 94.032 -97.781188 94.032 -100.008 c
h q 0.5 g
B Q
0.072 w
q 0 g
91.8 -88.776 m
90 -95.976 l
88.2 -88.776 l
91.8 -88.776 l
B Q
0.576 w
90 -64.008 m
90 -88.776 l
S
Q
}%
    \graphtemp=.5ex
    \advance\graphtemp by 1.111in
    \rlap{\kern 1.250in\lower\graphtemp\hbox to 0pt{\hss ~~~~$\tau$\hss}}%
\pdfliteral{
q [] 0 d 1 J 1 j
0.576 w
0.072 w
q 0 g
91.8 -122.832 m
90 -130.032 l
88.2 -122.832 l
91.8 -122.832 l
B Q
0.576 w
90 -103.968 m
90 -122.832 l
S
Q
}%
    \graphtemp=.5ex
    \advance\graphtemp by 1.625in
    \rlap{\kern 1.250in\lower\graphtemp\hbox to 0pt{\hss ~~~~$a$\hss}}%
\pdfliteral{
q [] 0 d 1 J 1 j
0.576 w
214.056 -20.016 m
214.056 -22.242812 212.250812 -24.048 210.024 -24.048 c
207.797188 -24.048 205.992 -22.242812 205.992 -20.016 c
205.992 -17.789188 207.797188 -15.984 210.024 -15.984 c
212.250812 -15.984 214.056 -17.789188 214.056 -20.016 c
h q 0.5 g
B Q
0.072 w
q 0 g
211.824 -8.784 m
210.024 -15.984 l
208.224 -8.784 l
211.824 -8.784 l
B Q
0.576 w
210.024 0 m
210.024 -8.784 l
S
180 -59.976 m
180 -65.503266 175.519266 -69.984 169.992 -69.984 c
164.464734 -69.984 159.984 -65.503266 159.984 -59.976 c
159.984 -54.448734 164.464734 -49.968 169.992 -49.968 c
175.519266 -49.968 180 -54.448734 180 -59.976 c
S
Q
}%
    \graphtemp=.5ex
    \advance\graphtemp by 0.833in
    \rlap{\kern 2.361in\lower\graphtemp\hbox to 0pt{\hss $P$\hss}}%
\pdfliteral{
q [] 0 d 1 J 1 j
0.576 w
0.072 w
q 0 g
183.456 -49.104 m
177.048 -52.92 l
180.864 -46.584 l
183.456 -49.104 l
B Q
0.576 w
207.144 -22.824 m
182.16 -47.808 l
S
Q
}%
    \graphtemp=.5ex
    \advance\graphtemp by 0.526in
    \rlap{\kern 2.668in\lower\graphtemp\hbox to 0pt{\hss $b$~~~~~\hss}}%
\pdfliteral{
q [] 0 d 1 J 1 j
0.576 w
214.056 -59.976 m
214.056 -62.202812 212.250812 -64.008 210.024 -64.008 c
207.797188 -64.008 205.992 -62.202812 205.992 -59.976 c
205.992 -57.749188 207.797188 -55.944 210.024 -55.944 c
212.250812 -55.944 214.056 -57.749188 214.056 -59.976 c
h q 0.5 g
B Q
0.072 w
q 0 g
211.824 -48.816 m
210.024 -56.016 l
208.224 -48.816 l
211.824 -48.816 l
B Q
0.576 w
210.024 -23.976 m
210.024 -48.816 l
S
Q
}%
    \graphtemp=.5ex
    \advance\graphtemp by 0.556in
    \rlap{\kern 2.917in\lower\graphtemp\hbox to 0pt{\hss ~~~~$\rt$\hss}}%
\pdfliteral{
q [] 0 d 1 J 1 j
0.576 w
214.056 -100.008 m
214.056 -102.234812 212.250812 -104.04 210.024 -104.04 c
207.797188 -104.04 205.992 -102.234812 205.992 -100.008 c
205.992 -97.781188 207.797188 -95.976 210.024 -95.976 c
212.250812 -95.976 214.056 -97.781188 214.056 -100.008 c
h q 0.5 g
B Q
180 -100.008 m
180 -105.535266 175.519266 -110.016 169.992 -110.016 c
164.464734 -110.016 159.984 -105.535266 159.984 -100.008 c
159.984 -94.480734 164.464734 -90 169.992 -90 c
175.519266 -90 180 -94.480734 180 -100.008 c
S
Q
}%
    \graphtemp=.5ex
    \advance\graphtemp by 1.389in
    \rlap{\kern 2.361in\lower\graphtemp\hbox to 0pt{\hss $Q$\hss}}%
\pdfliteral{
q [] 0 d 1 J 1 j
0.576 w
0.072 w
q 0 g
183.456 -89.136 m
177.048 -92.952 l
180.864 -86.544 l
183.456 -89.136 l
B Q
0.576 w
207.144 -62.856 m
182.16 -87.84 l
S
Q
}%
    \graphtemp=.5ex
    \advance\graphtemp by 1.082in
    \rlap{\kern 2.668in\lower\graphtemp\hbox to 0pt{\hss $a$~~~~~\hss}}%
\pdfliteral{
q [] 0 d 1 J 1 j
0.576 w
180 -139.968 m
180 -145.495266 175.519266 -149.976 169.992 -149.976 c
164.464734 -149.976 159.984 -145.495266 159.984 -139.968 c
159.984 -134.440734 164.464734 -129.96 169.992 -129.96 c
175.519266 -129.96 180 -134.440734 180 -139.968 c
S
Q
}%
    \graphtemp=.5ex
    \advance\graphtemp by 1.944in
    \rlap{\kern 2.361in\lower\graphtemp\hbox to 0pt{\hss $R$\hss}}%
\pdfliteral{
q [] 0 d 1 J 1 j
0.576 w
0.072 w
q 0 g
183.456 -129.096 m
177.048 -132.912 l
180.864 -126.576 l
183.456 -129.096 l
B Q
0.576 w
207.144 -102.816 m
182.16 -127.872 l
S
Q
}%
    \graphtemp=.5ex
    \advance\graphtemp by 1.637in
    \rlap{\kern 2.668in\lower\graphtemp\hbox to 0pt{\hss $b$~~~~~\hss}}%
\pdfliteral{
q [] 0 d 1 J 1 j
0.576 w
259.992 -139.968 m
259.992 -145.495266 255.511266 -149.976 249.984 -149.976 c
244.456734 -149.976 239.976 -145.495266 239.976 -139.968 c
239.976 -134.440734 244.456734 -129.96 249.984 -129.96 c
255.511266 -129.96 259.992 -134.440734 259.992 -139.968 c
S
Q
}%
    \graphtemp=.5ex
    \advance\graphtemp by 1.944in
    \rlap{\kern 3.472in\lower\graphtemp\hbox to 0pt{\hss $S$\hss}}%
\pdfliteral{
q [] 0 d 1 J 1 j
0.576 w
0.072 w
q 0 g
239.112 -126.576 m
242.928 -132.912 l
236.592 -129.096 l
239.112 -126.576 l
B Q
0.576 w
212.832 -102.816 m
237.816 -127.872 l
S
Q
}%
    \graphtemp=.5ex
    \advance\graphtemp by 1.637in
    \rlap{\kern 3.165in\lower\graphtemp\hbox to 0pt{\hss ~~~~~$a$\hss}}%
\pdfliteral{
q [] 0 d 1 J 1 j
0.576 w
254.016 -59.976 m
254.016 -62.202812 252.210812 -64.008 249.984 -64.008 c
247.757188 -64.008 245.952 -62.202812 245.952 -59.976 c
245.952 -57.749188 247.757188 -55.944 249.984 -55.944 c
252.210812 -55.944 254.016 -57.749188 254.016 -59.976 c
h q 0.5 g
B Q
0.072 w
q 0 g
243.36 -50.832 m
247.176 -57.168 l
240.84 -53.352 l
243.36 -50.832 l
B Q
0.576 w
212.832 -22.824 m
242.064 -52.056 l
S
Q
}%
    \graphtemp=.5ex
    \advance\graphtemp by 0.556in
    \rlap{\kern 3.194in\lower\graphtemp\hbox to 0pt{\hss ~~~~~$\rt$\hss}}%
\pdfliteral{
q [] 0 d 1 J 1 j
0.576 w
0.072 w
q 0 g
219.168 -93.384 m
212.832 -97.2 l
216.648 -90.792 l
219.168 -93.384 l
B Q
0.576 w
247.176 -62.856 m
217.944 -92.088 l
S
Q
}%
    \graphtemp=\baselineskip
    \multiply\graphtemp by 1
    \divide\graphtemp by 2
    \advance\graphtemp by .5ex
    \advance\graphtemp by 1.111in
    \rlap{\kern 3.194in\lower\graphtemp\hbox to 0pt{\hss $\tau$~~~~~~~~~\hss}}%
\pdfliteral{
q [] 0 d 1 J 1 j
0.576 w
254.016 -100.008 m
254.016 -102.234812 252.210812 -104.04 249.984 -104.04 c
247.757188 -104.04 245.952 -102.234812 245.952 -100.008 c
245.952 -97.781188 247.757188 -95.976 249.984 -95.976 c
252.210812 -95.976 254.016 -97.781188 254.016 -100.008 c
h q 0.5 g
B Q
0.072 w
q 0 g
243.36 -90.792 m
247.176 -97.2 l
240.84 -93.384 l
243.36 -90.792 l
B Q
0.576 w
212.832 -62.856 m
242.064 -92.088 l
S
Q
}%
    \graphtemp=\baselineskip
    \multiply\graphtemp by 1
    \divide\graphtemp by 2
    \advance\graphtemp by .5ex
    \advance\graphtemp by 1.111in
    \rlap{\kern 3.194in\lower\graphtemp\hbox to 0pt{\hss ~~~~~~~~~$\tau$\hss}}%
\pdfliteral{
q [] 0 d 1 J 1 j
0.576 w
0.072 w
q 0 g
251.784 -122.832 m
249.984 -130.032 l
248.184 -122.832 l
251.784 -122.832 l
B Q
0.576 w
249.984 -103.968 m
249.984 -122.832 l
S
Q
}%
    \graphtemp=.5ex
    \advance\graphtemp by 1.625in
    \rlap{\kern 3.472in\lower\graphtemp\hbox to 0pt{\hss ~~~~$a$\hss}}%
    \hbox{\vrule depth2.083in width0pt height 0pt}%
    \kern 3.611in
  }%
}%

%% file: pairs.tex
\expandafter\ifx\csname graph\endcsname\relax
   \csname newbox\expandafter\endcsname\csname graph\endcsname
\fi
\ifx\graphtemp\undefined
  \csname newdimen\endcsname\graphtemp
\fi
\expandafter\setbox\csname graph\endcsname
 =\vtop{\vskip 0pt\hbox{%
\pdfliteral{
q [] 0 d 1 J 1 j
0.576 w
0.576 w
54 -20.016 m
54 -22.242812 52.194812 -24.048 49.968 -24.048 c
47.741188 -24.048 45.936 -22.242812 45.936 -20.016 c
45.936 -17.789188 47.741188 -15.984 49.968 -15.984 c
52.194812 -15.984 54 -17.789188 54 -20.016 c
h q 0.5 g
B Q
0.072 w
q 0 g
51.768 -8.784 m
49.968 -15.984 l
48.168 -8.784 l
51.768 -8.784 l
B Q
0.576 w
49.968 0 m
49.968 -8.784 l
S
Q
}%
    \graphtemp=.5ex
    \advance\graphtemp by 0.278in
    \rlap{\kern 0.528in\lower\graphtemp\hbox to 0pt{\hss $U$\hss}}%
\pdfliteral{
q [] 0 d 1 J 1 j
0.576 w
54 -59.976 m
54 -62.202812 52.194812 -64.008 49.968 -64.008 c
47.741188 -64.008 45.936 -62.202812 45.936 -59.976 c
45.936 -57.749188 47.741188 -55.944 49.968 -55.944 c
52.194812 -55.944 54 -57.749188 54 -59.976 c
h q 0.5 g
B Q
0.072 w
q 0 g
51.768 -48.816 m
49.968 -56.016 l
48.168 -48.816 l
51.768 -48.816 l
B Q
0.576 w
49.968 -23.976 m
49.968 -48.816 l
S
Q
}%
    \graphtemp=.5ex
    \advance\graphtemp by 0.556in
    \rlap{\kern 0.694in\lower\graphtemp\hbox to 0pt{\hss ~~~~$\tau$\hss}}%
\pdfliteral{
q [] 0 d 1 J 1 j
0.576 w
54 -100.008 m
54 -102.234812 52.194812 -104.04 49.968 -104.04 c
47.741188 -104.04 45.936 -102.234812 45.936 -100.008 c
45.936 -97.781188 47.741188 -95.976 49.968 -95.976 c
52.194812 -95.976 54 -97.781188 54 -100.008 c
h q 0.5 g
B Q
0.072 w
q 0 g
51.768 -88.776 m
49.968 -95.976 l
48.168 -88.776 l
51.768 -88.776 l
B Q
0.576 w
49.968 -64.008 m
49.968 -88.776 l
S
Q
}%
    \graphtemp=.5ex
    \advance\graphtemp by 1.111in
    \rlap{\kern 0.694in\lower\graphtemp\hbox to 0pt{\hss ~~~~$\tau$\hss}}%
\pdfliteral{
q [] 0 d 1 J 1 j
0.576 w
20.016 -100.008 m
20.016 -105.535266 15.535266 -110.016 10.008 -110.016 c
4.480734 -110.016 0 -105.535266 0 -100.008 c
0 -94.480734 4.480734 -90 10.008 -90 c
15.535266 -90 20.016 -94.480734 20.016 -100.008 c
S
Q
}%
    \graphtemp=.5ex
    \advance\graphtemp by 1.389in
    \rlap{\kern 0.139in\lower\graphtemp\hbox to 0pt{\hss $Q$\hss}}%
\pdfliteral{
q [] 0 d 1 J 1 j
0.576 w
0.072 w
q 0 g
23.4 -89.136 m
17.064 -92.952 l
20.88 -86.544 l
23.4 -89.136 l
B Q
0.576 w
47.16 -62.856 m
22.176 -87.84 l
S
Q
}%
    \graphtemp=.5ex
    \advance\graphtemp by 1.082in
    \rlap{\kern 0.446in\lower\graphtemp\hbox to 0pt{\hss $a$~~~~~\hss}}%
\pdfliteral{
q [] 0 d 1 J 1 j
0.576 w
20.016 -139.968 m
20.016 -145.495266 15.535266 -149.976 10.008 -149.976 c
4.480734 -149.976 0 -145.495266 0 -139.968 c
0 -134.440734 4.480734 -129.96 10.008 -129.96 c
15.535266 -129.96 20.016 -134.440734 20.016 -139.968 c
S
Q
}%
    \graphtemp=.5ex
    \advance\graphtemp by 1.944in
    \rlap{\kern 0.139in\lower\graphtemp\hbox to 0pt{\hss $R$\hss}}%
\pdfliteral{
q [] 0 d 1 J 1 j
0.576 w
0.072 w
q 0 g
23.4 -129.096 m
17.064 -132.912 l
20.88 -126.576 l
23.4 -129.096 l
B Q
0.576 w
47.16 -102.816 m
22.176 -127.872 l
S
Q
}%
    \graphtemp=.5ex
    \advance\graphtemp by 1.637in
    \rlap{\kern 0.446in\lower\graphtemp\hbox to 0pt{\hss $\rt$~~~~~\hss}}%
\pdfliteral{
q [] 0 d 1 J 1 j
0.576 w
100.008 -139.968 m
100.008 -145.495266 95.527266 -149.976 90 -149.976 c
84.472734 -149.976 79.992 -145.495266 79.992 -139.968 c
79.992 -134.440734 84.472734 -129.96 90 -129.96 c
95.527266 -129.96 100.008 -134.440734 100.008 -139.968 c
S
Q
}%
    \graphtemp=.5ex
    \advance\graphtemp by 1.944in
    \rlap{\kern 1.250in\lower\graphtemp\hbox to 0pt{\hss $S$\hss}}%
\pdfliteral{
q [] 0 d 1 J 1 j
0.576 w
0.072 w
q 0 g
79.128 -126.576 m
82.944 -132.912 l
76.536 -129.096 l
79.128 -126.576 l
B Q
0.576 w
52.848 -102.816 m
77.832 -127.872 l
S
Q
}%
    \graphtemp=.5ex
    \advance\graphtemp by 1.637in
    \rlap{\kern 0.943in\lower\graphtemp\hbox to 0pt{\hss ~~~~~$a$\hss}}%
\pdfliteral{
q [] 0 d 1 J 1 j
0.576 w
94.032 -59.976 m
94.032 -62.202812 92.226812 -64.008 90 -64.008 c
87.773188 -64.008 85.968 -62.202812 85.968 -59.976 c
85.968 -57.749188 87.773188 -55.944 90 -55.944 c
92.226812 -55.944 94.032 -57.749188 94.032 -59.976 c
h q 0.5 g
B Q
0.072 w
q 0 g
83.376 -50.832 m
87.192 -57.168 l
80.784 -53.352 l
83.376 -50.832 l
B Q
0.576 w
52.848 -22.824 m
82.08 -52.056 l
S
Q
}%
    \graphtemp=.5ex
    \advance\graphtemp by 0.556in
    \rlap{\kern 0.972in\lower\graphtemp\hbox to 0pt{\hss ~~~~~$\tau$\hss}}%
\pdfliteral{
q [] 0 d 1 J 1 j
0.576 w
94.032 -100.008 m
94.032 -102.234812 92.226812 -104.04 90 -104.04 c
87.773188 -104.04 85.968 -102.234812 85.968 -100.008 c
85.968 -97.781188 87.773188 -95.976 90 -95.976 c
92.226812 -95.976 94.032 -97.781188 94.032 -100.008 c
h q 0.5 g
B Q
0.072 w
q 0 g
91.8 -88.776 m
90 -95.976 l
88.2 -88.776 l
91.8 -88.776 l
B Q
0.576 w
90 -64.008 m
90 -88.776 l
S
Q
}%
    \graphtemp=.5ex
    \advance\graphtemp by 1.111in
    \rlap{\kern 1.250in\lower\graphtemp\hbox to 0pt{\hss ~~~~$\tau$\hss}}%
\pdfliteral{
q [] 0 d 1 J 1 j
0.576 w
0.072 w
q 0 g
91.8 -122.832 m
90 -130.032 l
88.2 -122.832 l
91.8 -122.832 l
B Q
0.576 w
90 -103.968 m
90 -122.832 l
S
Q
}%
    \graphtemp=.5ex
    \advance\graphtemp by 1.625in
    \rlap{\kern 1.250in\lower\graphtemp\hbox to 0pt{\hss ~~~~$a$\hss}}%
\pdfliteral{
q [] 0 d 1 J 1 j
0.576 w
214.056 -20.016 m
214.056 -22.242812 212.250812 -24.048 210.024 -24.048 c
207.797188 -24.048 205.992 -22.242812 205.992 -20.016 c
205.992 -17.789188 207.797188 -15.984 210.024 -15.984 c
212.250812 -15.984 214.056 -17.789188 214.056 -20.016 c
h q 0.5 g
B Q
0.072 w
q 0 g
211.824 -8.784 m
210.024 -15.984 l
208.224 -8.784 l
211.824 -8.784 l
B Q
0.576 w
210.024 0 m
210.024 -8.784 l
S
Q
}%
    \graphtemp=.5ex
    \advance\graphtemp by 0.278in
    \rlap{\kern 2.750in\lower\graphtemp\hbox to 0pt{\hss $V$\hss}}%
\pdfliteral{
q [] 0 d 1 J 1 j
0.576 w
214.056 -59.976 m
214.056 -62.202812 212.250812 -64.008 210.024 -64.008 c
207.797188 -64.008 205.992 -62.202812 205.992 -59.976 c
205.992 -57.749188 207.797188 -55.944 210.024 -55.944 c
212.250812 -55.944 214.056 -57.749188 214.056 -59.976 c
h q 0.5 g
B Q
0.072 w
q 0 g
211.824 -48.816 m
210.024 -56.016 l
208.224 -48.816 l
211.824 -48.816 l
B Q
0.576 w
210.024 -23.976 m
210.024 -48.816 l
S
Q
}%
    \graphtemp=.5ex
    \advance\graphtemp by 0.556in
    \rlap{\kern 2.917in\lower\graphtemp\hbox to 0pt{\hss ~~~~$\tau$\hss}}%
\pdfliteral{
q [] 0 d 1 J 1 j
0.576 w
214.056 -100.008 m
214.056 -102.234812 212.250812 -104.04 210.024 -104.04 c
207.797188 -104.04 205.992 -102.234812 205.992 -100.008 c
205.992 -97.781188 207.797188 -95.976 210.024 -95.976 c
212.250812 -95.976 214.056 -97.781188 214.056 -100.008 c
h q 0.5 g
B Q
180 -100.008 m
180 -105.535266 175.519266 -110.016 169.992 -110.016 c
164.464734 -110.016 159.984 -105.535266 159.984 -100.008 c
159.984 -94.480734 164.464734 -90 169.992 -90 c
175.519266 -90 180 -94.480734 180 -100.008 c
S
Q
}%
    \graphtemp=.5ex
    \advance\graphtemp by 1.389in
    \rlap{\kern 2.361in\lower\graphtemp\hbox to 0pt{\hss $Q$\hss}}%
\pdfliteral{
q [] 0 d 1 J 1 j
0.576 w
0.072 w
q 0 g
183.456 -89.136 m
177.048 -92.952 l
180.864 -86.544 l
183.456 -89.136 l
B Q
0.576 w
207.144 -62.856 m
182.16 -87.84 l
S
Q
}%
    \graphtemp=.5ex
    \advance\graphtemp by 1.082in
    \rlap{\kern 2.668in\lower\graphtemp\hbox to 0pt{\hss $a$~~~~~\hss}}%
\pdfliteral{
q [] 0 d 1 J 1 j
0.576 w
180 -139.968 m
180 -145.495266 175.519266 -149.976 169.992 -149.976 c
164.464734 -149.976 159.984 -145.495266 159.984 -139.968 c
159.984 -134.440734 164.464734 -129.96 169.992 -129.96 c
175.519266 -129.96 180 -134.440734 180 -139.968 c
S
Q
}%
    \graphtemp=.5ex
    \advance\graphtemp by 1.944in
    \rlap{\kern 2.361in\lower\graphtemp\hbox to 0pt{\hss $R$\hss}}%
\pdfliteral{
q [] 0 d 1 J 1 j
0.576 w
0.072 w
q 0 g
183.456 -129.096 m
177.048 -132.912 l
180.864 -126.576 l
183.456 -129.096 l
B Q
0.576 w
207.144 -102.816 m
182.16 -127.872 l
S
Q
}%
    \graphtemp=.5ex
    \advance\graphtemp by 1.637in
    \rlap{\kern 2.668in\lower\graphtemp\hbox to 0pt{\hss $\rt$~~~~~\hss}}%
\pdfliteral{
q [] 0 d 1 J 1 j
0.576 w
259.992 -139.968 m
259.992 -145.495266 255.511266 -149.976 249.984 -149.976 c
244.456734 -149.976 239.976 -145.495266 239.976 -139.968 c
239.976 -134.440734 244.456734 -129.96 249.984 -129.96 c
255.511266 -129.96 259.992 -134.440734 259.992 -139.968 c
S
Q
}%
    \graphtemp=.5ex
    \advance\graphtemp by 1.944in
    \rlap{\kern 3.472in\lower\graphtemp\hbox to 0pt{\hss $S$\hss}}%
\pdfliteral{
q [] 0 d 1 J 1 j
0.576 w
0.072 w
q 0 g
239.112 -126.576 m
242.928 -132.912 l
236.592 -129.096 l
239.112 -126.576 l
B Q
0.576 w
212.832 -102.816 m
237.816 -127.872 l
S
Q
}%
    \graphtemp=.5ex
    \advance\graphtemp by 1.637in
    \rlap{\kern 3.165in\lower\graphtemp\hbox to 0pt{\hss ~~~~~$a$\hss}}%
\pdfliteral{
q [] 0 d 1 J 1 j
0.576 w
254.016 -59.976 m
254.016 -62.202812 252.210812 -64.008 249.984 -64.008 c
247.757188 -64.008 245.952 -62.202812 245.952 -59.976 c
245.952 -57.749188 247.757188 -55.944 249.984 -55.944 c
252.210812 -55.944 254.016 -57.749188 254.016 -59.976 c
h q 0.5 g
B Q
0.072 w
q 0 g
243.36 -50.832 m
247.176 -57.168 l
240.84 -53.352 l
243.36 -50.832 l
B Q
0.576 w
212.832 -22.824 m
242.064 -52.056 l
S
Q
}%
    \graphtemp=.5ex
    \advance\graphtemp by 0.556in
    \rlap{\kern 3.194in\lower\graphtemp\hbox to 0pt{\hss ~~~~~$\tau$\hss}}%
\pdfliteral{
q [] 0 d 1 J 1 j
0.576 w
0.072 w
q 0 g
219.168 -93.384 m
212.832 -97.2 l
216.648 -90.792 l
219.168 -93.384 l
B Q
0.576 w
247.176 -62.856 m
217.944 -92.088 l
S
Q
}%
    \graphtemp=\baselineskip
    \multiply\graphtemp by 1
    \divide\graphtemp by 2
    \advance\graphtemp by .5ex
    \advance\graphtemp by 1.111in
    \rlap{\kern 3.194in\lower\graphtemp\hbox to 0pt{\hss $\tau$~~~~~~~~~\hss}}%
\pdfliteral{
q [] 0 d 1 J 1 j
0.576 w
254.016 -100.008 m
254.016 -102.234812 252.210812 -104.04 249.984 -104.04 c
247.757188 -104.04 245.952 -102.234812 245.952 -100.008 c
245.952 -97.781188 247.757188 -95.976 249.984 -95.976 c
252.210812 -95.976 254.016 -97.781188 254.016 -100.008 c
h q 0.5 g
B Q
0.072 w
q 0 g
243.36 -90.792 m
247.176 -97.2 l
240.84 -93.384 l
243.36 -90.792 l
B Q
0.576 w
212.832 -62.856 m
242.064 -92.088 l
S
Q
}%
    \graphtemp=\baselineskip
    \multiply\graphtemp by 1
    \divide\graphtemp by 2
    \advance\graphtemp by .5ex
    \advance\graphtemp by 1.111in
    \rlap{\kern 3.194in\lower\graphtemp\hbox to 0pt{\hss ~~~~~~~~~$\tau$\hss}}%
\pdfliteral{
q [] 0 d 1 J 1 j
0.576 w
0.072 w
q 0 g
251.784 -122.832 m
249.984 -130.032 l
248.184 -122.832 l
251.784 -122.832 l
B Q
0.576 w
249.984 -103.968 m
249.984 -122.832 l
S
Q
}%
    \graphtemp=.5ex
    \advance\graphtemp by 1.625in
    \rlap{\kern 3.472in\lower\graphtemp\hbox to 0pt{\hss ~~~~$a$\hss}}%
    \hbox{\vrule depth2.083in width0pt height 0pt}%
    \kern 3.611in
  }%
}%

%% file: product.tex
\expandafter\ifx\csname graph\endcsname\relax
   \csname newbox\expandafter\endcsname\csname graph\endcsname
\fi
\ifx\graphtemp\undefined
  \csname newdimen\endcsname\graphtemp
\fi
\expandafter\setbox\csname graph\endcsname
 =\vtop{\vskip 0pt\hbox{%
\pdfliteral{
q [] 0 d 1 J 1 j
0.576 w
0.576 w
55.944 -20.016 m
55.944 -23.316454 53.268454 -25.992 49.968 -25.992 c
46.667546 -25.992 43.992 -23.316454 43.992 -20.016 c
43.992 -16.715546 46.667546 -14.04 49.968 -14.04 c
53.268454 -14.04 55.944 -16.715546 55.944 -20.016 c
S
Q
}%
    \graphtemp=.5ex
    \advance\graphtemp by 0.278in
    \rlap{\kern 0.694in\lower\graphtemp\hbox to 0pt{\hss \scriptsize $1$\hss}}%
\pdfliteral{
q [] 0 d 1 J 1 j
0.576 w
20.016 -59.976 m
20.016 -65.503266 15.535266 -69.984 10.008 -69.984 c
4.480734 -69.984 0 -65.503266 0 -59.976 c
0 -54.448734 4.480734 -49.968 10.008 -49.968 c
15.535266 -49.968 20.016 -54.448734 20.016 -59.976 c
S
Q
}%
    \graphtemp=.5ex
    \advance\graphtemp by 0.833in
    \rlap{\kern 0.139in\lower\graphtemp\hbox to 0pt{\hss $P$\hss}}%
\pdfliteral{
q [] 0 d 1 J 1 j
0.576 w
0.072 w
q 0 g
51.768 -6.768 m
49.968 -13.968 l
48.168 -6.768 l
51.768 -6.768 l
B Q
0.576 w
49.968 0 m
49.968 -6.768 l
S
0.072 w
q 0 g
23.4 -49.104 m
17.064 -52.92 l
20.88 -46.584 l
23.4 -49.104 l
B Q
0.576 w
45.792 -24.264 m
22.176 -47.808 l
S
Q
}%
    \graphtemp=.5ex
    \advance\graphtemp by 0.536in
    \rlap{\kern 0.436in\lower\graphtemp\hbox to 0pt{\hss $b$~~~~~\hss}}%
\pdfliteral{
q [] 0 d 1 J 1 j
0.576 w
55.944 -59.976 m
55.944 -63.276454 53.268454 -65.952 49.968 -65.952 c
46.667546 -65.952 43.992 -63.276454 43.992 -59.976 c
43.992 -56.675546 46.667546 -54 49.968 -54 c
53.268454 -54 55.944 -56.675546 55.944 -59.976 c
S
Q
}%
    \graphtemp=.5ex
    \advance\graphtemp by 0.833in
    \rlap{\kern 0.694in\lower\graphtemp\hbox to 0pt{\hss \scriptsize $2$\hss}}%
\pdfliteral{
q [] 0 d 1 J 1 j
0.576 w
0.072 w
q 0 g
51.768 -46.8 m
49.968 -54 l
48.168 -46.8 l
51.768 -46.8 l
B Q
0.576 w
49.968 -25.992 m
49.968 -46.8 l
S
Q
}%
    \graphtemp=.5ex
    \advance\graphtemp by 0.556in
    \rlap{\kern 0.694in\lower\graphtemp\hbox to 0pt{\hss ~~~~$\rt$\hss}}%
\pdfliteral{
q [] 0 d 1 J 1 j
0.576 w
55.944 -100.008 m
55.944 -103.308454 53.268454 -105.984 49.968 -105.984 c
46.667546 -105.984 43.992 -103.308454 43.992 -100.008 c
43.992 -96.707546 46.667546 -94.032 49.968 -94.032 c
53.268454 -94.032 55.944 -96.707546 55.944 -100.008 c
S
Q
}%
    \graphtemp=.5ex
    \advance\graphtemp by 1.389in
    \rlap{\kern 0.694in\lower\graphtemp\hbox to 0pt{\hss \scriptsize $4$\hss}}%
\pdfliteral{
q [] 0 d 1 J 1 j
0.576 w
0.072 w
q 0 g
51.768 -86.832 m
49.968 -94.032 l
48.168 -86.832 l
51.768 -86.832 l
B Q
0.576 w
49.968 -66.024 m
49.968 -86.832 l
S
Q
}%
    \graphtemp=.5ex
    \advance\graphtemp by 1.111in
    \rlap{\kern 0.694in\lower\graphtemp\hbox to 0pt{\hss ~~~~$\tau$\hss}}%
\pdfliteral{
q [] 0 d 1 J 1 j
0.576 w
20.016 -100.008 m
20.016 -105.535266 15.535266 -110.016 10.008 -110.016 c
4.480734 -110.016 0 -105.535266 0 -100.008 c
0 -94.480734 4.480734 -90 10.008 -90 c
15.535266 -90 20.016 -94.480734 20.016 -100.008 c
S
Q
}%
    \graphtemp=.5ex
    \advance\graphtemp by 1.389in
    \rlap{\kern 0.139in\lower\graphtemp\hbox to 0pt{\hss $Q$\hss}}%
\pdfliteral{
q [] 0 d 1 J 1 j
0.576 w
0.072 w
q 0 g
23.4 -89.136 m
17.064 -92.952 l
20.88 -86.544 l
23.4 -89.136 l
B Q
0.576 w
45.792 -64.224 m
22.176 -87.84 l
S
Q
}%
    \graphtemp=.5ex
    \advance\graphtemp by 1.091in
    \rlap{\kern 0.436in\lower\graphtemp\hbox to 0pt{\hss $a$~~~~~\hss}}%
\pdfliteral{
q [] 0 d 1 J 1 j
0.576 w
20.016 -139.968 m
20.016 -145.495266 15.535266 -149.976 10.008 -149.976 c
4.480734 -149.976 0 -145.495266 0 -139.968 c
0 -134.440734 4.480734 -129.96 10.008 -129.96 c
15.535266 -129.96 20.016 -134.440734 20.016 -139.968 c
S
Q
}%
    \graphtemp=.5ex
    \advance\graphtemp by 1.944in
    \rlap{\kern 0.139in\lower\graphtemp\hbox to 0pt{\hss $R$\hss}}%
\pdfliteral{
q [] 0 d 1 J 1 j
0.576 w
0.072 w
q 0 g
23.4 -129.096 m
17.064 -132.912 l
20.88 -126.576 l
23.4 -129.096 l
B Q
0.576 w
45.792 -104.256 m
22.176 -127.872 l
S
Q
}%
    \graphtemp=.5ex
    \advance\graphtemp by 1.647in
    \rlap{\kern 0.436in\lower\graphtemp\hbox to 0pt{\hss $b$~~~~~\hss}}%
\pdfliteral{
q [] 0 d 1 J 1 j
0.576 w
100.008 -139.968 m
100.008 -145.495266 95.527266 -149.976 90 -149.976 c
84.472734 -149.976 79.992 -145.495266 79.992 -139.968 c
79.992 -134.440734 84.472734 -129.96 90 -129.96 c
95.527266 -129.96 100.008 -134.440734 100.008 -139.968 c
S
Q
}%
    \graphtemp=.5ex
    \advance\graphtemp by 1.944in
    \rlap{\kern 1.250in\lower\graphtemp\hbox to 0pt{\hss $S$\hss}}%
\pdfliteral{
q [] 0 d 1 J 1 j
0.576 w
0.072 w
q 0 g
79.128 -126.576 m
82.944 -132.912 l
76.536 -129.096 l
79.128 -126.576 l
B Q
0.576 w
54.216 -104.256 m
77.832 -127.872 l
S
Q
}%
    \graphtemp=.5ex
    \advance\graphtemp by 1.647in
    \rlap{\kern 0.953in\lower\graphtemp\hbox to 0pt{\hss ~~~~~$a$\hss}}%
\pdfliteral{
q [] 0 d 1 J 1 j
0.576 w
95.976 -59.976 m
95.976 -63.276454 93.300454 -65.952 90 -65.952 c
86.699546 -65.952 84.024 -63.276454 84.024 -59.976 c
84.024 -56.675546 86.699546 -54 90 -54 c
93.300454 -54 95.976 -56.675546 95.976 -59.976 c
S
Q
}%
    \graphtemp=.5ex
    \advance\graphtemp by 0.833in
    \rlap{\kern 1.250in\lower\graphtemp\hbox to 0pt{\hss \scriptsize $3$\hss}}%
\pdfliteral{
q [] 0 d 1 J 1 j
0.576 w
0.072 w
q 0 g
81.936 -49.392 m
85.752 -55.728 l
79.416 -51.912 l
81.936 -49.392 l
B Q
0.576 w
54.216 -24.264 m
80.64 -50.688 l
S
Q
}%
    \graphtemp=.5ex
    \advance\graphtemp by 0.556in
    \rlap{\kern 0.972in\lower\graphtemp\hbox to 0pt{\hss ~~~~~$\rt$\hss}}%
\pdfliteral{
q [] 0 d 1 J 1 j
0.576 w
140.04 -100.008 m
140.04 -105.535266 135.559266 -110.016 130.032 -110.016 c
124.504734 -110.016 120.024 -105.535266 120.024 -100.008 c
120.024 -94.480734 124.504734 -90 130.032 -90 c
135.559266 -90 140.04 -94.480734 140.04 -100.008 c
S
Q
}%
    \graphtemp=.5ex
    \advance\graphtemp by 1.389in
    \rlap{\kern 1.806in\lower\graphtemp\hbox to 0pt{\hss $T$\hss}}%
\pdfliteral{
q [] 0 d 1 J 1 j
0.576 w
0.072 w
q 0 g
119.088 -86.544 m
122.904 -92.952 l
116.568 -89.136 l
119.088 -86.544 l
B Q
0.576 w
94.248 -64.224 m
117.864 -87.84 l
S
Q
}%
    \graphtemp=.5ex
    \advance\graphtemp by 1.091in
    \rlap{\kern 1.508in\lower\graphtemp\hbox to 0pt{\hss ~~~~~$a$\hss}}%
\pdfliteral{
q [] 0 d 1 J 1 j
0.576 w
95.976 -100.008 m
95.976 -103.308454 93.300454 -105.984 90 -105.984 c
86.699546 -105.984 84.024 -103.308454 84.024 -100.008 c
84.024 -96.707546 86.699546 -94.032 90 -94.032 c
93.300454 -94.032 95.976 -96.707546 95.976 -100.008 c
S
Q
}%
    \graphtemp=.5ex
    \advance\graphtemp by 1.389in
    \rlap{\kern 1.250in\lower\graphtemp\hbox to 0pt{\hss \scriptsize $5$\hss}}%
\pdfliteral{
q [] 0 d 1 J 1 j
0.576 w
0.072 w
q 0 g
91.8 -86.832 m
90 -94.032 l
88.2 -86.832 l
91.8 -86.832 l
B Q
0.576 w
90 -66.024 m
90 -86.832 l
S
Q
}%
    \graphtemp=.5ex
    \advance\graphtemp by 1.111in
    \rlap{\kern 1.250in\lower\graphtemp\hbox to 0pt{\hss ~~~~$\tau$\hss}}%
\pdfliteral{
q [] 0 d 1 J 1 j
0.576 w
0.072 w
q 0 g
91.8 -122.832 m
90 -130.032 l
88.2 -122.832 l
91.8 -122.832 l
B Q
0.576 w
90 -105.984 m
90 -122.832 l
S
Q
}%
    \graphtemp=.5ex
    \advance\graphtemp by 1.639in
    \rlap{\kern 1.250in\lower\graphtemp\hbox to 0pt{\hss ~~~~$a$\hss}}%
    \graphtemp=.5ex
    \advance\graphtemp by 1.111in
    \rlap{\kern 2.639in\lower\graphtemp\hbox to 0pt{\hss \huge $\times$\hss}}%
\pdfliteral{
q [] 0 d 1 J 1 j
0.576 w
295.992 -20.016 m
295.992 -23.316454 293.316454 -25.992 290.016 -25.992 c
286.715546 -25.992 284.04 -23.316454 284.04 -20.016 c
284.04 -16.715546 286.715546 -14.04 290.016 -14.04 c
293.316454 -14.04 295.992 -16.715546 295.992 -20.016 c
S
Q
}%
    \graphtemp=.5ex
    \advance\graphtemp by 0.278in
    \rlap{\kern 4.028in\lower\graphtemp\hbox to 0pt{\hss \scriptsize $6$\hss}}%
\pdfliteral{
q [] 0 d 1 J 1 j
0.576 w
0.072 w
q 0 g
291.816 -6.768 m
290.016 -13.968 l
288.216 -6.768 l
291.816 -6.768 l
B Q
0.576 w
290.016 0 m
290.016 -6.768 l
S
259.992 -59.976 m
259.992 -65.503266 255.511266 -69.984 249.984 -69.984 c
244.456734 -69.984 239.976 -65.503266 239.976 -59.976 c
239.976 -54.448734 244.456734 -49.968 249.984 -49.968 c
255.511266 -49.968 259.992 -54.448734 259.992 -59.976 c
S
Q
}%
    \graphtemp=.5ex
    \advance\graphtemp by 0.833in
    \rlap{\kern 3.472in\lower\graphtemp\hbox to 0pt{\hss $P$\hss}}%
\pdfliteral{
q [] 0 d 1 J 1 j
0.576 w
0.072 w
q 0 g
263.448 -49.104 m
257.04 -52.92 l
260.856 -46.584 l
263.448 -49.104 l
B Q
0.576 w
285.768 -24.264 m
262.152 -47.808 l
S
Q
}%
    \graphtemp=.5ex
    \advance\graphtemp by 0.536in
    \rlap{\kern 3.770in\lower\graphtemp\hbox to 0pt{\hss $b$~~~~~\hss}}%
\pdfliteral{
q [] 0 d 1 J 1 j
0.576 w
295.992 -59.976 m
295.992 -63.276454 293.316454 -65.952 290.016 -65.952 c
286.715546 -65.952 284.04 -63.276454 284.04 -59.976 c
284.04 -56.675546 286.715546 -54 290.016 -54 c
293.316454 -54 295.992 -56.675546 295.992 -59.976 c
S
Q
}%
    \graphtemp=.5ex
    \advance\graphtemp by 0.833in
    \rlap{\kern 4.028in\lower\graphtemp\hbox to 0pt{\hss \scriptsize $7$\hss}}%
\pdfliteral{
q [] 0 d 1 J 1 j
0.576 w
0.072 w
q 0 g
291.816 -46.8 m
290.016 -54 l
288.216 -46.8 l
291.816 -46.8 l
B Q
0.576 w
290.016 -25.992 m
290.016 -46.8 l
S
Q
}%
    \graphtemp=.5ex
    \advance\graphtemp by 0.556in
    \rlap{\kern 4.028in\lower\graphtemp\hbox to 0pt{\hss ~~~~$\rt$\hss}}%
\pdfliteral{
q [] 0 d 1 J 1 j
0.576 w
295.992 -100.008 m
295.992 -103.308454 293.316454 -105.984 290.016 -105.984 c
286.715546 -105.984 284.04 -103.308454 284.04 -100.008 c
284.04 -96.707546 286.715546 -94.032 290.016 -94.032 c
293.316454 -94.032 295.992 -96.707546 295.992 -100.008 c
S
Q
}%
    \graphtemp=.5ex
    \advance\graphtemp by 1.389in
    \rlap{\kern 4.028in\lower\graphtemp\hbox to 0pt{\hss \scriptsize $9$\hss}}%
\pdfliteral{
q [] 0 d 1 J 1 j
0.576 w
259.992 -100.008 m
259.992 -105.535266 255.511266 -110.016 249.984 -110.016 c
244.456734 -110.016 239.976 -105.535266 239.976 -100.008 c
239.976 -94.480734 244.456734 -90 249.984 -90 c
255.511266 -90 259.992 -94.480734 259.992 -100.008 c
S
Q
}%
    \graphtemp=.5ex
    \advance\graphtemp by 1.389in
    \rlap{\kern 3.472in\lower\graphtemp\hbox to 0pt{\hss $Q$\hss}}%
\pdfliteral{
q [] 0 d 1 J 1 j
0.576 w
0.072 w
q 0 g
263.448 -89.136 m
257.04 -92.952 l
260.856 -86.544 l
263.448 -89.136 l
B Q
0.576 w
285.768 -64.224 m
262.152 -87.84 l
S
Q
}%
    \graphtemp=.5ex
    \advance\graphtemp by 1.091in
    \rlap{\kern 3.770in\lower\graphtemp\hbox to 0pt{\hss $a$~~~~~\hss}}%
\pdfliteral{
q [] 0 d 1 J 1 j
0.576 w
259.992 -139.968 m
259.992 -145.495266 255.511266 -149.976 249.984 -149.976 c
244.456734 -149.976 239.976 -145.495266 239.976 -139.968 c
239.976 -134.440734 244.456734 -129.96 249.984 -129.96 c
255.511266 -129.96 259.992 -134.440734 259.992 -139.968 c
S
Q
}%
    \graphtemp=.5ex
    \advance\graphtemp by 1.944in
    \rlap{\kern 3.472in\lower\graphtemp\hbox to 0pt{\hss $R$\hss}}%
\pdfliteral{
q [] 0 d 1 J 1 j
0.576 w
0.072 w
q 0 g
263.448 -129.096 m
257.04 -132.912 l
260.856 -126.576 l
263.448 -129.096 l
B Q
0.576 w
285.768 -104.256 m
262.152 -127.872 l
S
Q
}%
    \graphtemp=.5ex
    \advance\graphtemp by 1.647in
    \rlap{\kern 3.770in\lower\graphtemp\hbox to 0pt{\hss $b$~~~~~\hss}}%
\pdfliteral{
q [] 0 d 1 J 1 j
0.576 w
339.984 -139.968 m
339.984 -145.495266 335.503266 -149.976 329.976 -149.976 c
324.448734 -149.976 319.968 -145.495266 319.968 -139.968 c
319.968 -134.440734 324.448734 -129.96 329.976 -129.96 c
335.503266 -129.96 339.984 -134.440734 339.984 -139.968 c
S
Q
}%
    \graphtemp=.5ex
    \advance\graphtemp by 1.944in
    \rlap{\kern 4.583in\lower\graphtemp\hbox to 0pt{\hss $S$\hss}}%
\pdfliteral{
q [] 0 d 1 J 1 j
0.576 w
0.072 w
q 0 g
319.104 -126.576 m
322.92 -132.912 l
316.584 -129.096 l
319.104 -126.576 l
B Q
0.576 w
294.264 -104.256 m
317.808 -127.872 l
S
Q
}%
    \graphtemp=.5ex
    \advance\graphtemp by 1.647in
    \rlap{\kern 4.286in\lower\graphtemp\hbox to 0pt{\hss ~~~~~$a$\hss}}%
\pdfliteral{
q [] 0 d 1 J 1 j
0.576 w
335.952 -59.976 m
335.952 -63.276454 333.276454 -65.952 329.976 -65.952 c
326.675546 -65.952 324 -63.276454 324 -59.976 c
324 -56.675546 326.675546 -54 329.976 -54 c
333.276454 -54 335.952 -56.675546 335.952 -59.976 c
S
Q
}%
    \graphtemp=.5ex
    \advance\graphtemp by 0.833in
    \rlap{\kern 4.583in\lower\graphtemp\hbox to 0pt{\hss \scriptsize $8$\hss}}%
\pdfliteral{
q [] 0 d 1 J 1 j
0.576 w
0.072 w
q 0 g
321.912 -49.392 m
325.728 -55.728 l
319.392 -51.912 l
321.912 -49.392 l
B Q
0.576 w
294.264 -24.264 m
320.688 -50.688 l
S
Q
}%
    \graphtemp=.5ex
    \advance\graphtemp by 0.556in
    \rlap{\kern 4.306in\lower\graphtemp\hbox to 0pt{\hss ~~~~~$\rt$\hss}}%
\pdfliteral{
q [] 0 d 1 J 1 j
0.576 w
380.016 -100.008 m
380.016 -105.535266 375.535266 -110.016 370.008 -110.016 c
364.480734 -110.016 360 -105.535266 360 -100.008 c
360 -94.480734 364.480734 -90 370.008 -90 c
375.535266 -90 380.016 -94.480734 380.016 -100.008 c
S
Q
}%
    \graphtemp=.5ex
    \advance\graphtemp by 1.389in
    \rlap{\kern 5.139in\lower\graphtemp\hbox to 0pt{\hss $T$\hss}}%
\pdfliteral{
q [] 0 d 1 J 1 j
0.576 w
0.072 w
q 0 g
359.136 -86.544 m
362.952 -92.952 l
356.544 -89.136 l
359.136 -86.544 l
B Q
0.576 w
334.224 -64.224 m
357.84 -87.84 l
S
Q
}%
    \graphtemp=.5ex
    \advance\graphtemp by 1.091in
    \rlap{\kern 4.841in\lower\graphtemp\hbox to 0pt{\hss ~~~~~$a$\hss}}%
\pdfliteral{
q [] 0 d 1 J 1 j
0.576 w
0.072 w
q 0 g
300.6 -91.944 m
294.264 -95.76 l
298.08 -89.424 l
300.6 -91.944 l
B Q
0.576 w
325.728 -64.224 m
299.304 -90.648 l
S
Q
}%
    \graphtemp=\baselineskip
    \multiply\graphtemp by 1
    \divide\graphtemp by 2
    \advance\graphtemp by .5ex
    \advance\graphtemp by 1.111in
    \rlap{\kern 4.306in\lower\graphtemp\hbox to 0pt{\hss $\tau$~~~~~~~~~~~\hss}}%
\pdfliteral{
q [] 0 d 1 J 1 j
0.576 w
335.952 -100.008 m
335.952 -103.308454 333.276454 -105.984 329.976 -105.984 c
326.675546 -105.984 324 -103.308454 324 -100.008 c
324 -96.707546 326.675546 -94.032 329.976 -94.032 c
333.276454 -94.032 335.952 -96.707546 335.952 -100.008 c
S
Q
}%
    \graphtemp=.5ex
    \advance\graphtemp by 1.389in
    \rlap{\kern 4.583in\lower\graphtemp\hbox to 0pt{\hss \scriptsize $10$\hss}}%
\pdfliteral{
q [] 0 d 1 J 1 j
0.576 w
0.072 w
q 0 g
321.912 -89.424 m
325.728 -95.76 l
319.392 -91.944 l
321.912 -89.424 l
B Q
0.576 w
294.264 -64.224 m
320.688 -90.648 l
S
Q
}%
    \graphtemp=\baselineskip
    \multiply\graphtemp by 1
    \divide\graphtemp by 2
    \advance\graphtemp by .5ex
    \advance\graphtemp by 1.111in
    \rlap{\kern 4.306in\lower\graphtemp\hbox to 0pt{\hss ~~~~~~~~~~~~$\tau$\hss}}%
\pdfliteral{
q [] 0 d 1 J 1 j
0.576 w
0.072 w
q 0 g
331.776 -122.832 m
329.976 -130.032 l
328.176 -122.832 l
331.776 -122.832 l
B Q
0.576 w
329.976 -105.984 m
329.976 -122.832 l
S
Q
}%
    \graphtemp=.5ex
    \advance\graphtemp by 1.639in
    \rlap{\kern 4.583in\lower\graphtemp\hbox to 0pt{\hss ~~~~$a$\hss}}%
    \graphtemp=.5ex
    \advance\graphtemp by 3.333in
    \rlap{\kern 0.139in\lower\graphtemp\hbox to 0pt{\hss \huge $=$\hss}}%
\pdfliteral{
q [] 0 d 1 J 1 j
0.576 w
179.784 -180 m
179.784 -183.300454 175.399972 -185.976 169.992 -185.976 c
164.584028 -185.976 160.2 -183.300454 160.2 -180 c
160.2 -176.699546 164.584028 -174.024 169.992 -174.024 c
175.399972 -174.024 179.784 -176.699546 179.784 -180 c
S
Q
}%
    \graphtemp=.5ex
    \advance\graphtemp by 2.500in
    \rlap{\kern 2.361in\lower\graphtemp\hbox to 0pt{\hss \scriptsize $1$,$6$\hss}}%
\pdfliteral{
q [] 0 d 1 J 1 j
0.576 w
0.072 w
q 0 g
171.792 -166.824 m
169.992 -174.024 l
168.192 -166.824 l
171.792 -166.824 l
B Q
0.576 w
169.992 -159.984 m
169.992 -166.824 l
S
100.008 -220.032 m
100.008 -225.559266 95.527266 -230.04 90 -230.04 c
84.472734 -230.04 79.992 -225.559266 79.992 -220.032 c
79.992 -214.504734 84.472734 -210.024 90 -210.024 c
95.527266 -210.024 100.008 -214.504734 100.008 -220.032 c
S
Q
}%
    \graphtemp=.5ex
    \advance\graphtemp by 3.056in
    \rlap{\kern 1.250in\lower\graphtemp\hbox to 0pt{\hss $P$\hss}}%
\pdfliteral{
q [] 0 d 1 J 1 j
0.576 w
0.072 w
q 0 g
106.272 -214.056 m
99 -215.64 l
104.688 -210.816 l
106.272 -214.056 l
B Q
0.576 w
163.08 -184.248 m
105.48 -212.472 l
S
Q
}%
    \graphtemp=\baselineskip
    \multiply\graphtemp by -1
    \divide\graphtemp by 2
    \advance\graphtemp by .5ex
    \advance\graphtemp by 2.777in
    \rlap{\kern 1.820in\lower\graphtemp\hbox to 0pt{\hss $b$\hss}}%
\pdfliteral{
q [] 0 d 1 J 1 j
0.576 w
139.824 -220.032 m
139.824 -223.332454 135.439972 -226.008 130.032 -226.008 c
124.624028 -226.008 120.24 -223.332454 120.24 -220.032 c
120.24 -216.731546 124.624028 -214.056 130.032 -214.056 c
135.439972 -214.056 139.824 -216.731546 139.824 -220.032 c
S
Q
}%
    \graphtemp=.5ex
    \advance\graphtemp by 3.056in
    \rlap{\kern 1.806in\lower\graphtemp\hbox to 0pt{\hss \scriptsize $2$,$7$\hss}}%
\pdfliteral{
q [] 0 d 1 J 1 j
0.576 w
0.072 w
q 0 g
141.48 -211.104 m
135.072 -214.92 l
138.888 -208.512 l
141.48 -211.104 l
B Q
0.576 w
164.88 -185.112 m
140.184 -209.808 l
S
Q
}%
    \graphtemp=\baselineskip
    \multiply\graphtemp by 1
    \divide\graphtemp by 2
    \advance\graphtemp by .5ex
    \advance\graphtemp by 2.778in
    \rlap{\kern 2.083in\lower\graphtemp\hbox to 0pt{\hss ~~$\rt$\hss}}%
\pdfliteral{
q [] 0 d 1 J 1 j
0.576 w
179.784 -220.032 m
179.784 -223.332454 175.399972 -226.008 169.992 -226.008 c
164.584028 -226.008 160.2 -223.332454 160.2 -220.032 c
160.2 -216.731546 164.584028 -214.056 169.992 -214.056 c
175.399972 -214.056 179.784 -216.731546 179.784 -220.032 c
S
Q
}%
    \graphtemp=.5ex
    \advance\graphtemp by 3.056in
    \rlap{\kern 2.361in\lower\graphtemp\hbox to 0pt{\hss \scriptsize $2$,$8$\hss}}%
\pdfliteral{
q [] 0 d 1 J 1 j
0.576 w
0.072 w
q 0 g
171.792 -206.784 m
169.992 -213.984 l
168.192 -206.784 l
171.792 -206.784 l
B Q
0.576 w
169.992 -185.976 m
169.992 -206.784 l
S
Q
}%
    \graphtemp=.5ex
    \advance\graphtemp by 2.778in
    \rlap{\kern 2.361in\lower\graphtemp\hbox to 0pt{\hss ~~~~$\rt$\hss}}%
\pdfliteral{
q [] 0 d 1 J 1 j
0.576 w
179.784 -259.992 m
179.784 -263.292454 175.399972 -265.968 169.992 -265.968 c
164.584028 -265.968 160.2 -263.292454 160.2 -259.992 c
160.2 -256.691546 164.584028 -254.016 169.992 -254.016 c
175.399972 -254.016 179.784 -256.691546 179.784 -259.992 c
S
Q
}%
    \graphtemp=.5ex
    \advance\graphtemp by 3.611in
    \rlap{\kern 2.361in\lower\graphtemp\hbox to 0pt{\hss \scriptsize $4$,$9$\hss}}%
\pdfliteral{
q [] 0 d 1 J 1 j
0.576 w
0.072 w
q 0 g
171.792 -246.816 m
169.992 -254.016 l
168.192 -246.816 l
171.792 -246.816 l
B Q
0.576 w
169.992 -226.008 m
169.992 -246.816 l
S
Q
}%
    \graphtemp=.5ex
    \advance\graphtemp by 3.333in
    \rlap{\kern 2.361in\lower\graphtemp\hbox to 0pt{\hss ~~~~$\tau$\hss}}%
\pdfliteral{
q [] 0 d 1 J 1 j
0.576 w
140.04 -300.024 m
140.04 -305.551266 135.559266 -310.032 130.032 -310.032 c
124.504734 -310.032 120.024 -305.551266 120.024 -300.024 c
120.024 -294.496734 124.504734 -290.016 130.032 -290.016 c
135.559266 -290.016 140.04 -294.496734 140.04 -300.024 c
S
Q
}%
    \graphtemp=.5ex
    \advance\graphtemp by 4.167in
    \rlap{\kern 1.806in\lower\graphtemp\hbox to 0pt{\hss $R$\hss}}%
\pdfliteral{
q [] 0 d 1 J 1 j
0.576 w
0.072 w
q 0 g
142.992 -288.576 m
136.8 -292.68 l
140.328 -286.128 l
142.992 -288.576 l
B Q
0.576 w
163.08 -264.24 m
141.696 -287.352 l
S
Q
}%
    \graphtemp=.5ex
    \advance\graphtemp by 3.867in
    \rlap{\kern 2.082in\lower\graphtemp\hbox to 0pt{\hss $b$~~~~~\hss}}%
\pdfliteral{
q [] 0 d 1 J 1 j
0.576 w
220.032 -300.024 m
220.032 -305.551266 215.551266 -310.032 210.024 -310.032 c
204.496734 -310.032 200.016 -305.551266 200.016 -300.024 c
200.016 -294.496734 204.496734 -290.016 210.024 -290.016 c
215.551266 -290.016 220.032 -294.496734 220.032 -300.024 c
S
Q
}%
    \graphtemp=.5ex
    \advance\graphtemp by 4.167in
    \rlap{\kern 2.917in\lower\graphtemp\hbox to 0pt{\hss $S$\hss}}%
\pdfliteral{
q [] 0 d 1 J 1 j
0.576 w
0.072 w
q 0 g
199.656 -286.128 m
203.184 -292.68 l
196.992 -288.576 l
199.656 -286.128 l
B Q
0.576 w
176.904 -264.24 m
198.288 -287.352 l
S
Q
}%
    \graphtemp=.5ex
    \advance\graphtemp by 3.867in
    \rlap{\kern 2.640in\lower\graphtemp\hbox to 0pt{\hss ~~~~~$a$\hss}}%
\pdfliteral{
q [] 0 d 1 J 1 j
0.576 w
219.816 -220.032 m
219.816 -223.332454 215.431972 -226.008 210.024 -226.008 c
204.616028 -226.008 200.232 -223.332454 200.232 -220.032 c
200.232 -216.731546 204.616028 -214.056 210.024 -214.056 c
215.431972 -214.056 219.816 -216.731546 219.816 -220.032 c
S
Q
}%
    \graphtemp=.5ex
    \advance\graphtemp by 3.056in
    \rlap{\kern 2.917in\lower\graphtemp\hbox to 0pt{\hss \scriptsize $3$,$7$\hss}}%
\pdfliteral{
q [] 0 d 1 J 1 j
0.576 w
0.072 w
q 0 g
201.096 -208.512 m
204.912 -214.92 l
198.576 -211.104 l
201.096 -208.512 l
B Q
0.576 w
175.104 -185.112 m
199.8 -209.808 l
S
Q
}%
    \graphtemp=.5ex
    \advance\graphtemp by 2.778in
    \rlap{\kern 2.639in\lower\graphtemp\hbox to 0pt{\hss ~~~~$\rt$\hss}}%
\pdfliteral{
q [] 0 d 1 J 1 j
0.576 w
219.816 -259.992 m
219.816 -263.292454 215.431972 -265.968 210.024 -265.968 c
204.616028 -265.968 200.232 -263.292454 200.232 -259.992 c
200.232 -256.691546 204.616028 -254.016 210.024 -254.016 c
215.431972 -254.016 219.816 -256.691546 219.816 -259.992 c
S
Q
}%
    \graphtemp=.5ex
    \advance\graphtemp by 3.611in
    \rlap{\kern 2.917in\lower\graphtemp\hbox to 0pt{\hss \scriptsize $5$,$10$\hss}}%
\pdfliteral{
q [] 0 d 1 J 1 j
0.576 w
0.072 w
q 0 g
211.824 -246.816 m
210.024 -254.016 l
208.224 -246.816 l
211.824 -246.816 l
B Q
0.576 w
210.024 -226.008 m
210.024 -246.816 l
S
Q
}%
    \graphtemp=.5ex
    \advance\graphtemp by 3.333in
    \rlap{\kern 2.917in\lower\graphtemp\hbox to 0pt{\hss ~~~~$\tau$\hss}}%
\pdfliteral{
q [] 0 d 1 J 1 j
0.576 w
0.072 w
q 0 g
211.824 -282.816 m
210.024 -290.016 l
208.224 -282.816 l
211.824 -282.816 l
B Q
0.576 w
210.024 -265.968 m
210.024 -282.816 l
S
Q
}%
    \graphtemp=.5ex
    \advance\graphtemp by 3.861in
    \rlap{\kern 2.917in\lower\graphtemp\hbox to 0pt{\hss ~~~~$a$\hss}}%
\pdfliteral{
q [] 0 d 1 J 1 j
0.576 w
139.824 -259.992 m
139.824 -263.292454 135.439972 -265.968 130.032 -265.968 c
124.624028 -265.968 120.24 -263.292454 120.24 -259.992 c
120.24 -256.691546 124.624028 -254.016 130.032 -254.016 c
135.439972 -254.016 139.824 -256.691546 139.824 -259.992 c
S
Q
}%
    \graphtemp=.5ex
    \advance\graphtemp by 3.611in
    \rlap{\kern 1.806in\lower\graphtemp\hbox to 0pt{\hss \scriptsize $4$,$10$\hss}}%
\pdfliteral{
q [] 0 d 1 J 1 j
0.576 w
0.072 w
q 0 g
131.832 -246.816 m
130.032 -254.016 l
128.232 -246.816 l
131.832 -246.816 l
B Q
0.576 w
130.032 -226.008 m
130.032 -246.816 l
S
Q
}%
    \graphtemp=.5ex
    \advance\graphtemp by 3.333in
    \rlap{\kern 1.806in\lower\graphtemp\hbox to 0pt{\hss ~~~~$\tau$\hss}}%
\pdfliteral{
q [] 0 d 1 J 1 j
0.576 w
0.072 w
q 0 g
195.768 -290.88 m
201.384 -295.704 l
194.184 -294.12 l
195.768 -290.88 l
B Q
0.576 w
137.16 -263.592 m
194.976 -292.464 l
S
Q
}%
    \graphtemp=.5ex
    \advance\graphtemp by 3.884in
    \rlap{\kern 2.351in\lower\graphtemp\hbox to 0pt{\hss ~~~~~~~~$a$\hss}}%
\pdfliteral{
q [] 0 d 1 J 1 j
0.576 w
100.008 -259.992 m
100.008 -265.519266 95.527266 -270 90 -270 c
84.472734 -270 79.992 -265.519266 79.992 -259.992 c
79.992 -254.464734 84.472734 -249.984 90 -249.984 c
95.527266 -249.984 100.008 -254.464734 100.008 -259.992 c
S
Q
}%
    \graphtemp=.5ex
    \advance\graphtemp by 3.611in
    \rlap{\kern 1.250in\lower\graphtemp\hbox to 0pt{\hss $Q$\hss}}%
\pdfliteral{
q [] 0 d 1 J 1 j
0.576 w
0.072 w
q 0 g
103.032 -248.616 m
96.768 -252.648 l
100.368 -246.168 l
103.032 -248.616 l
B Q
0.576 w
123.048 -224.208 m
101.664 -247.392 l
S
Q
}%
    \graphtemp=.5ex
    \advance\graphtemp by 3.312in
    \rlap{\kern 1.527in\lower\graphtemp\hbox to 0pt{\hss $a$~~~~~\hss}}%
\pdfliteral{
q [] 0 d 1 J 1 j
0.576 w
259.776 -220.032 m
259.776 -223.332454 255.391972 -226.008 249.984 -226.008 c
244.576028 -226.008 240.192 -223.332454 240.192 -220.032 c
240.192 -216.731546 244.576028 -214.056 249.984 -214.056 c
255.391972 -214.056 259.776 -216.731546 259.776 -220.032 c
S
Q
}%
    \graphtemp=.5ex
    \advance\graphtemp by 3.056in
    \rlap{\kern 3.472in\lower\graphtemp\hbox to 0pt{\hss \scriptsize $3$,$8$\hss}}%
\pdfliteral{
q [] 0 d 1 J 1 j
0.576 w
0.072 w
q 0 g
236.88 -211.392 m
242.496 -216.216 l
235.224 -214.632 l
236.88 -211.392 l
B Q
0.576 w
177.48 -183.744 m
236.016 -213.048 l
S
Q
}%
    \graphtemp=.5ex
    \advance\graphtemp by 2.778in
    \rlap{\kern 2.917in\lower\graphtemp\hbox to 0pt{\hss ~~~~~~$\rt$\hss}}%
\pdfliteral{
q [] 0 d 1 J 1 j
0.576 w
300.024 -259.992 m
300.024 -265.519266 295.543266 -270 290.016 -270 c
284.488734 -270 280.008 -265.519266 280.008 -259.992 c
280.008 -254.464734 284.488734 -249.984 290.016 -249.984 c
295.543266 -249.984 300.024 -254.464734 300.024 -259.992 c
S
Q
}%
    \graphtemp=.5ex
    \advance\graphtemp by 3.611in
    \rlap{\kern 4.028in\lower\graphtemp\hbox to 0pt{\hss $T$\hss}}%
\pdfliteral{
q [] 0 d 1 J 1 j
0.576 w
0.072 w
q 0 g
279.648 -246.168 m
283.176 -252.648 l
276.984 -248.616 l
279.648 -246.168 l
B Q
0.576 w
256.896 -224.208 m
278.352 -247.392 l
S
Q
}%
    \graphtemp=.5ex
    \advance\graphtemp by 3.312in
    \rlap{\kern 3.751in\lower\graphtemp\hbox to 0pt{\hss ~~~~~$a$\hss}}%
\pdfliteral{
q [] 0 d 1 J 1 j
0.576 w
259.776 -259.992 m
259.776 -263.292454 255.391972 -265.968 249.984 -265.968 c
244.576028 -265.968 240.192 -263.292454 240.192 -259.992 c
240.192 -256.691546 244.576028 -254.016 249.984 -254.016 c
255.391972 -254.016 259.776 -256.691546 259.776 -259.992 c
S
Q
}%
    \graphtemp=.5ex
    \advance\graphtemp by 3.611in
    \rlap{\kern 3.472in\lower\graphtemp\hbox to 0pt{\hss \scriptsize $5$,$9$\hss}}%
\pdfliteral{
q [] 0 d 1 J 1 j
0.576 w
0.072 w
q 0 g
251.784 -246.816 m
249.984 -254.016 l
248.184 -246.816 l
251.784 -246.816 l
B Q
0.576 w
249.984 -226.008 m
249.984 -246.816 l
S
Q
}%
    \graphtemp=.5ex
    \advance\graphtemp by 3.333in
    \rlap{\kern 3.472in\lower\graphtemp\hbox to 0pt{\hss ~~~~$\tau$\hss}}%
\pdfliteral{
q [] 0 d 1 J 1 j
0.576 w
0.072 w
q 0 g
223.272 -288.792 m
217.08 -292.896 l
220.608 -286.416 l
223.272 -288.792 l
B Q
0.576 w
243.072 -264.24 m
221.904 -287.568 l
S
Q
}%
    \graphtemp=.5ex
    \advance\graphtemp by 3.869in
    \rlap{\kern 3.195in\lower\graphtemp\hbox to 0pt{\hss ~~~~~$a$\hss}}%
    \hbox{\vrule depth4.306in width0pt height 0pt}%
    \kern 5.278in
  }%
}%

%% file: uncountable.tex
\expandafter\ifx\csname graph\endcsname\relax
   \csname newbox\expandafter\endcsname\csname graph\endcsname
\fi
\ifx\graphtemp\undefined
  \csname newdimen\endcsname\graphtemp
\fi
\expandafter\setbox\csname graph\endcsname
 =\vtop{\vskip 0pt\hbox{%
\pdfliteral{
q [] 0 d 1 J 1 j
0.576 w
0.576 w
24.048 -87.984 m
24.048 -90.210812 22.242812 -92.016 20.016 -92.016 c
17.789188 -92.016 15.984 -90.210812 15.984 -87.984 c
15.984 -85.757188 17.789188 -83.952 20.016 -83.952 c
22.242812 -83.952 24.048 -85.757188 24.048 -87.984 c
S
0.072 w
q 0 g
8.784 -86.184 m
15.984 -87.984 l
8.784 -89.784 l
8.784 -86.184 l
B Q
0.576 w
0 -87.984 m
8.784 -87.984 l
S
52.056 -59.976 m
52.056 -62.202812 50.250812 -64.008 48.024 -64.008 c
45.797188 -64.008 43.992 -62.202812 43.992 -59.976 c
43.992 -57.749188 45.797188 -55.944 48.024 -55.944 c
50.250812 -55.944 52.056 -57.749188 52.056 -59.976 c
S
0.072 w
q 0 g
38.808 -66.672 m
45.144 -62.856 l
41.328 -69.192 l
38.808 -66.672 l
B Q
0.576 w
22.824 -85.176 m
40.104 -67.896 l
S
Q
}%
    \graphtemp=.5ex
    \advance\graphtemp by 1.028in
    \rlap{\kern 0.472in\lower\graphtemp\hbox to 0pt{\hss $b~~~~~$\hss}}%
\pdfliteral{
q [] 0 d 1 J 1 j
0.576 w
64.008 -87.984 m
64.008 -90.210812 62.202812 -92.016 59.976 -92.016 c
57.749188 -92.016 55.944 -90.210812 55.944 -87.984 c
55.944 -85.757188 57.749188 -83.952 59.976 -83.952 c
62.202812 -83.952 64.008 -85.757188 64.008 -87.984 c
S
0.072 w
q 0 g
48.816 -86.184 m
56.016 -87.984 l
48.816 -89.784 l
48.816 -86.184 l
B Q
0.576 w
23.976 -87.984 m
48.816 -87.984 l
S
Q
}%
    \graphtemp=\baselineskip
    \multiply\graphtemp by 1
    \divide\graphtemp by 2
    \advance\graphtemp by .5ex
    \advance\graphtemp by 1.222in
    \rlap{\kern 0.556in\lower\graphtemp\hbox to 0pt{\hss $\rt$\hss}}%
\pdfliteral{
q [] 0 d 1 J 1 j
0.576 w
92.016 -59.976 m
92.016 -62.202812 90.210812 -64.008 87.984 -64.008 c
85.757188 -64.008 83.952 -62.202812 83.952 -59.976 c
83.952 -57.749188 85.757188 -55.944 87.984 -55.944 c
90.210812 -55.944 92.016 -57.749188 92.016 -59.976 c
S
0.072 w
q 0 g
78.84 -66.672 m
85.176 -62.856 l
81.36 -69.192 l
78.84 -66.672 l
B Q
0.576 w
62.856 -85.176 m
80.064 -67.896 l
S
Q
}%
    \graphtemp=.5ex
    \advance\graphtemp by 1.028in
    \rlap{\kern 1.028in\lower\graphtemp\hbox to 0pt{\hss $\tau~~~~~$\hss}}%
\pdfliteral{
q [] 0 d 1 J 1 j
0.576 w
120.024 -31.968 m
120.024 -34.194812 118.218812 -36 115.992 -36 c
113.765188 -36 111.96 -34.194812 111.96 -31.968 c
111.96 -29.741188 113.765188 -27.936 115.992 -27.936 c
118.218812 -27.936 120.024 -29.741188 120.024 -31.968 c
S
0.072 w
q 0 g
106.776 -38.664 m
113.184 -34.848 l
109.368 -41.184 l
106.776 -38.664 l
B Q
0.576 w
90.864 -57.168 m
108.072 -39.888 l
S
Q
}%
    \graphtemp=.5ex
    \advance\graphtemp by 0.639in
    \rlap{\kern 1.417in\lower\graphtemp\hbox to 0pt{\hss $\tau~~~~~$\hss}}%
\pdfliteral{
q [] 0 d 1 J 1 j
0.576 w
160.056 -31.968 m
160.056 -34.194812 158.250812 -36 156.024 -36 c
153.797188 -36 151.992 -34.194812 151.992 -31.968 c
151.992 -29.741188 153.797188 -27.936 156.024 -27.936 c
158.250812 -27.936 160.056 -29.741188 160.056 -31.968 c
S
0.072 w
q 0 g
144.792 -30.168 m
151.992 -31.968 l
144.792 -33.768 l
144.792 -30.168 l
B Q
0.576 w
120.024 -31.968 m
144.792 -31.968 l
S
Q
}%
    \graphtemp=\baselineskip
    \multiply\graphtemp by -1
    \divide\graphtemp by 2
    \advance\graphtemp by .5ex
    \advance\graphtemp by 0.444in
    \rlap{\kern 1.889in\lower\graphtemp\hbox to 0pt{\hss $a$\hss}}%
\pdfliteral{
q [] 0 d 1 J 1 j
0.576 w
148.032 -4.032 m
148.032 -6.258812 146.226812 -8.064 144 -8.064 c
141.773188 -8.064 139.968 -6.258812 139.968 -4.032 c
139.968 -1.805188 141.773188 0 144 0 c
146.226812 0 148.032 -1.805188 148.032 -4.032 c
S
0.072 w
q 0 g
134.784 -10.656 m
141.192 -6.84 l
137.376 -13.176 l
134.784 -10.656 l
B Q
0.576 w
q [3.6 3.328917] 0 d
118.8 -29.16 m
136.08 -11.952 l
S Q
Q
}%
    \graphtemp=.5ex
    \advance\graphtemp by 0.250in
    \rlap{\kern 1.806in\lower\graphtemp\hbox to 0pt{\hss $b~~~~~$\hss}}%
\pdfliteral{
q [] 0 d 1 J 1 j
0.576 w
92.016 -115.992 m
92.016 -118.218812 90.210812 -120.024 87.984 -120.024 c
85.757188 -120.024 83.952 -118.218812 83.952 -115.992 c
83.952 -113.765188 85.757188 -111.96 87.984 -111.96 c
90.210812 -111.96 92.016 -113.765188 92.016 -115.992 c
S
0.072 w
q 0 g
81.36 -106.776 m
85.176 -113.184 l
78.84 -109.368 l
81.36 -106.776 l
B Q
0.576 w
62.856 -90.864 m
80.064 -108.072 l
S
Q
}%
    \graphtemp=.5ex
    \advance\graphtemp by 1.417in
    \rlap{\kern 1.028in\lower\graphtemp\hbox to 0pt{\hss $~~~~~\tau$\hss}}%
\pdfliteral{
q [] 0 d 1 J 1 j
0.576 w
132.048 -115.992 m
132.048 -118.218812 130.242812 -120.024 128.016 -120.024 c
125.789188 -120.024 123.984 -118.218812 123.984 -115.992 c
123.984 -113.765188 125.789188 -111.96 128.016 -111.96 c
130.242812 -111.96 132.048 -113.765188 132.048 -115.992 c
S
0.072 w
q 0 g
116.784 -114.192 m
123.984 -115.992 l
116.784 -117.792 l
116.784 -114.192 l
B Q
0.576 w
92.016 -115.992 m
116.784 -115.992 l
S
Q
}%
    \graphtemp=\baselineskip
    \multiply\graphtemp by -1
    \divide\graphtemp by 2
    \advance\graphtemp by .5ex
    \advance\graphtemp by 1.611in
    \rlap{\kern 1.500in\lower\graphtemp\hbox to 0pt{\hss $a$\hss}}%
\pdfliteral{
q [] 0 d 1 J 1 j
0.576 w
120.024 -144 m
120.024 -146.226812 118.218812 -148.032 115.992 -148.032 c
113.765188 -148.032 111.96 -146.226812 111.96 -144 c
111.96 -141.773188 113.765188 -139.968 115.992 -139.968 c
118.218812 -139.968 120.024 -141.773188 120.024 -144 c
S
0.072 w
q 0 g
109.368 -134.784 m
113.184 -141.192 l
106.776 -137.376 l
109.368 -134.784 l
B Q
0.576 w
90.864 -118.8 m
108.072 -136.08 l
S
Q
}%
    \graphtemp=.5ex
    \advance\graphtemp by 1.806in
    \rlap{\kern 1.417in\lower\graphtemp\hbox to 0pt{\hss $\tau~~~~~$\hss}}%
\pdfliteral{
q [] 0 d 1 J 1 j
0.576 w
160.056 -144 m
160.056 -146.226812 158.250812 -148.032 156.024 -148.032 c
153.797188 -148.032 151.992 -146.226812 151.992 -144 c
151.992 -141.773188 153.797188 -139.968 156.024 -139.968 c
158.250812 -139.968 160.056 -141.773188 160.056 -144 c
S
0.072 w
q 0 g
144.792 -142.2 m
151.992 -144 l
144.792 -145.8 l
144.792 -142.2 l
B Q
0.576 w
120.024 -144 m
144.792 -144 l
S
Q
}%
    \graphtemp=\baselineskip
    \multiply\graphtemp by -1
    \divide\graphtemp by 2
    \advance\graphtemp by .5ex
    \advance\graphtemp by 2.000in
    \rlap{\kern 1.889in\lower\graphtemp\hbox to 0pt{\hss $a$\hss}}%
\pdfliteral{
q [] 0 d 1 J 1 j
0.576 w
148.032 -172.008 m
148.032 -174.234812 146.226812 -176.04 144 -176.04 c
141.773188 -176.04 139.968 -174.234812 139.968 -172.008 c
139.968 -169.781188 141.773188 -167.976 144 -167.976 c
146.226812 -167.976 148.032 -169.781188 148.032 -172.008 c
S
0.072 w
q 0 g
137.376 -162.792 m
141.192 -169.2 l
134.784 -165.384 l
137.376 -162.792 l
B Q
0.576 w
q [0 3.491087] 0 d
118.8 -146.808 m
136.08 -164.088 l
S Q
Q
}%
    \graphtemp=\baselineskip
    \multiply\graphtemp by 1
    \divide\graphtemp by 2
    \advance\graphtemp by .5ex
    \advance\graphtemp by 2.194in
    \rlap{\kern 1.806in\lower\graphtemp\hbox to 0pt{\hss $b~~~$\hss}}%
\pdfliteral{
q [] 0 d 1 J 1 j
0.576 w
120.024 -87.984 m
120.024 -90.210812 118.218812 -92.016 115.992 -92.016 c
113.765188 -92.016 111.96 -90.210812 111.96 -87.984 c
111.96 -85.757188 113.765188 -83.952 115.992 -83.952 c
118.218812 -83.952 120.024 -85.757188 120.024 -87.984 c
S
0.072 w
q 0 g
106.776 -94.68 m
113.184 -90.864 l
109.368 -97.2 l
106.776 -94.68 l
B Q
0.576 w
90.864 -113.184 m
108.072 -95.904 l
S
Q
}%
    \graphtemp=.5ex
    \advance\graphtemp by 1.417in
    \rlap{\kern 1.417in\lower\graphtemp\hbox to 0pt{\hss $\tau~~~~~$\hss}}%
\pdfliteral{
q [] 0 d 1 J 1 j
0.576 w
0.072 w
q 0 g
109.368 -78.84 m
113.184 -85.176 l
106.776 -81.36 l
109.368 -78.84 l
B Q
0.576 w
90.864 -62.856 m
108.072 -80.064 l
S
Q
}%
    \graphtemp=.5ex
    \advance\graphtemp by 1.028in
    \rlap{\kern 1.417in\lower\graphtemp\hbox to 0pt{\hss $~~~~~\tau$\hss}}%
\pdfliteral{
q [] 0 d 1 J 1 j
0.576 w
160.056 -87.984 m
160.056 -90.210812 158.250812 -92.016 156.024 -92.016 c
153.797188 -92.016 151.992 -90.210812 151.992 -87.984 c
151.992 -85.757188 153.797188 -83.952 156.024 -83.952 c
158.250812 -83.952 160.056 -85.757188 160.056 -87.984 c
S
0.072 w
q 0 g
144.792 -86.184 m
151.992 -87.984 l
144.792 -89.784 l
144.792 -86.184 l
B Q
0.576 w
120.024 -87.984 m
144.792 -87.984 l
S
Q
}%
    \graphtemp=\baselineskip
    \multiply\graphtemp by -1
    \divide\graphtemp by 2
    \advance\graphtemp by .5ex
    \advance\graphtemp by 1.222in
    \rlap{\kern 1.889in\lower\graphtemp\hbox to 0pt{\hss $\tau$\hss}}%
\pdfliteral{
q [] 0 d 1 J 1 j
0.576 w
188.064 -59.976 m
188.064 -62.202812 186.258812 -64.008 184.032 -64.008 c
181.805188 -64.008 180 -62.202812 180 -59.976 c
180 -57.749188 181.805188 -55.944 184.032 -55.944 c
186.258812 -55.944 188.064 -57.749188 188.064 -59.976 c
S
0.072 w
q 0 g
174.816 -66.672 m
181.152 -62.856 l
177.336 -69.192 l
174.816 -66.672 l
B Q
0.576 w
158.832 -85.176 m
176.112 -67.896 l
S
Q
}%
    \graphtemp=.5ex
    \advance\graphtemp by 1.028in
    \rlap{\kern 2.361in\lower\graphtemp\hbox to 0pt{\hss $\tau~~~~~$\hss}}%
\pdfliteral{
q [] 0 d 1 J 1 j
0.576 w
216 -31.968 m
216 -34.194812 214.194812 -36 211.968 -36 c
209.741188 -36 207.936 -34.194812 207.936 -31.968 c
207.936 -29.741188 209.741188 -27.936 211.968 -27.936 c
214.194812 -27.936 216 -29.741188 216 -31.968 c
S
0.072 w
q 0 g
202.824 -38.664 m
209.16 -34.848 l
205.344 -41.184 l
202.824 -38.664 l
B Q
0.576 w
186.84 -57.168 m
204.048 -39.888 l
S
Q
}%
    \graphtemp=.5ex
    \advance\graphtemp by 0.639in
    \rlap{\kern 2.750in\lower\graphtemp\hbox to 0pt{\hss $\tau~~~~~$\hss}}%
\pdfliteral{
q [] 0 d 1 J 1 j
0.576 w
256.032 -31.968 m
256.032 -34.194812 254.226812 -36 252 -36 c
249.773188 -36 247.968 -34.194812 247.968 -31.968 c
247.968 -29.741188 249.773188 -27.936 252 -27.936 c
254.226812 -27.936 256.032 -29.741188 256.032 -31.968 c
S
0.072 w
q 0 g
240.768 -30.168 m
247.968 -31.968 l
240.768 -33.768 l
240.768 -30.168 l
B Q
0.576 w
216 -31.968 m
240.768 -31.968 l
S
Q
}%
    \graphtemp=\baselineskip
    \multiply\graphtemp by -1
    \divide\graphtemp by 2
    \advance\graphtemp by .5ex
    \advance\graphtemp by 0.444in
    \rlap{\kern 3.222in\lower\graphtemp\hbox to 0pt{\hss $a$\hss}}%
\pdfliteral{
q [] 0 d 1 J 1 j
0.576 w
244.008 -4.032 m
244.008 -6.258812 242.202812 -8.064 239.976 -8.064 c
237.749188 -8.064 235.944 -6.258812 235.944 -4.032 c
235.944 -1.805188 237.749188 0 239.976 0 c
242.202812 0 244.008 -1.805188 244.008 -4.032 c
S
0.072 w
q 0 g
230.832 -10.656 m
237.168 -6.84 l
233.352 -13.176 l
230.832 -10.656 l
B Q
0.576 w
q [3.6 3.311929] 0 d
214.848 -29.16 m
232.056 -11.952 l
S Q
Q
}%
    \graphtemp=.5ex
    \advance\graphtemp by 0.250in
    \rlap{\kern 3.139in\lower\graphtemp\hbox to 0pt{\hss $b~~~~~$\hss}}%
\pdfliteral{
q [] 0 d 1 J 1 j
0.576 w
188.064 -115.992 m
188.064 -118.218812 186.258812 -120.024 184.032 -120.024 c
181.805188 -120.024 180 -118.218812 180 -115.992 c
180 -113.765188 181.805188 -111.96 184.032 -111.96 c
186.258812 -111.96 188.064 -113.765188 188.064 -115.992 c
S
0.072 w
q 0 g
177.336 -106.776 m
181.152 -113.184 l
174.816 -109.368 l
177.336 -106.776 l
B Q
0.576 w
158.832 -90.864 m
176.112 -108.072 l
S
Q
}%
    \graphtemp=.5ex
    \advance\graphtemp by 1.417in
    \rlap{\kern 2.361in\lower\graphtemp\hbox to 0pt{\hss $~~~~~\tau$\hss}}%
\pdfliteral{
q [] 0 d 1 J 1 j
0.576 w
228.024 -115.992 m
228.024 -118.218812 226.218812 -120.024 223.992 -120.024 c
221.765188 -120.024 219.96 -118.218812 219.96 -115.992 c
219.96 -113.765188 221.765188 -111.96 223.992 -111.96 c
226.218812 -111.96 228.024 -113.765188 228.024 -115.992 c
S
0.072 w
q 0 g
212.832 -114.192 m
220.032 -115.992 l
212.832 -117.792 l
212.832 -114.192 l
B Q
0.576 w
187.992 -115.992 m
212.832 -115.992 l
S
Q
}%
    \graphtemp=\baselineskip
    \multiply\graphtemp by -1
    \divide\graphtemp by 2
    \advance\graphtemp by .5ex
    \advance\graphtemp by 1.611in
    \rlap{\kern 2.833in\lower\graphtemp\hbox to 0pt{\hss $a$\hss}}%
\pdfliteral{
q [] 0 d 1 J 1 j
0.576 w
216 -144 m
216 -146.226812 214.194812 -148.032 211.968 -148.032 c
209.741188 -148.032 207.936 -146.226812 207.936 -144 c
207.936 -141.773188 209.741188 -139.968 211.968 -139.968 c
214.194812 -139.968 216 -141.773188 216 -144 c
S
0.072 w
q 0 g
205.344 -134.784 m
209.16 -141.192 l
202.824 -137.376 l
205.344 -134.784 l
B Q
0.576 w
186.84 -118.8 m
204.048 -136.08 l
S
Q
}%
    \graphtemp=.5ex
    \advance\graphtemp by 1.806in
    \rlap{\kern 2.750in\lower\graphtemp\hbox to 0pt{\hss $\tau~~~~~$\hss}}%
\pdfliteral{
q [] 0 d 1 J 1 j
0.576 w
256.032 -144 m
256.032 -146.226812 254.226812 -148.032 252 -148.032 c
249.773188 -148.032 247.968 -146.226812 247.968 -144 c
247.968 -141.773188 249.773188 -139.968 252 -139.968 c
254.226812 -139.968 256.032 -141.773188 256.032 -144 c
S
0.072 w
q 0 g
240.768 -142.2 m
247.968 -144 l
240.768 -145.8 l
240.768 -142.2 l
B Q
0.576 w
216 -144 m
240.768 -144 l
S
Q
}%
    \graphtemp=\baselineskip
    \multiply\graphtemp by -1
    \divide\graphtemp by 2
    \advance\graphtemp by .5ex
    \advance\graphtemp by 2.000in
    \rlap{\kern 3.222in\lower\graphtemp\hbox to 0pt{\hss $a$\hss}}%
\pdfliteral{
q [] 0 d 1 J 1 j
0.576 w
244.008 -172.008 m
244.008 -174.234812 242.202812 -176.04 239.976 -176.04 c
237.749188 -176.04 235.944 -174.234812 235.944 -172.008 c
235.944 -169.781188 237.749188 -167.976 239.976 -167.976 c
242.202812 -167.976 244.008 -169.781188 244.008 -172.008 c
S
0.072 w
q 0 g
233.352 -162.792 m
237.168 -169.2 l
230.832 -165.384 l
233.352 -162.792 l
B Q
0.576 w
q [0 3.483822] 0 d
214.848 -146.808 m
232.056 -164.088 l
S Q
Q
}%
    \graphtemp=\baselineskip
    \multiply\graphtemp by 1
    \divide\graphtemp by 2
    \advance\graphtemp by .5ex
    \advance\graphtemp by 2.194in
    \rlap{\kern 3.139in\lower\graphtemp\hbox to 0pt{\hss $b~~~$\hss}}%
\pdfliteral{
q [] 0 d 1 J 1 j
0.576 w
216 -87.984 m
216 -90.210812 214.194812 -92.016 211.968 -92.016 c
209.741188 -92.016 207.936 -90.210812 207.936 -87.984 c
207.936 -85.757188 209.741188 -83.952 211.968 -83.952 c
214.194812 -83.952 216 -85.757188 216 -87.984 c
S
0.072 w
q 0 g
202.824 -94.68 m
209.16 -90.864 l
205.344 -97.2 l
202.824 -94.68 l
B Q
0.576 w
186.84 -113.184 m
204.048 -95.904 l
S
Q
}%
    \graphtemp=.5ex
    \advance\graphtemp by 1.417in
    \rlap{\kern 2.750in\lower\graphtemp\hbox to 0pt{\hss $\tau~~~~~$\hss}}%
\pdfliteral{
q [] 0 d 1 J 1 j
0.576 w
0.072 w
q 0 g
205.344 -78.84 m
209.16 -85.176 l
202.824 -81.36 l
205.344 -78.84 l
B Q
0.576 w
186.84 -62.856 m
204.048 -80.064 l
S
Q
}%
    \graphtemp=.5ex
    \advance\graphtemp by 1.028in
    \rlap{\kern 2.750in\lower\graphtemp\hbox to 0pt{\hss $~~~~~\tau$\hss}}%
\pdfliteral{
q [] 0 d 1 J 1 j
0.576 w
256.032 -87.984 m
256.032 -90.210812 254.226812 -92.016 252 -92.016 c
249.773188 -92.016 247.968 -90.210812 247.968 -87.984 c
247.968 -85.757188 249.773188 -83.952 252 -83.952 c
254.226812 -83.952 256.032 -85.757188 256.032 -87.984 c
S
0.072 w
q 0 g
240.768 -86.184 m
247.968 -87.984 l
240.768 -89.784 l
240.768 -86.184 l
B Q
0.576 w
216 -87.984 m
240.768 -87.984 l
S
Q
}%
    \graphtemp=\baselineskip
    \multiply\graphtemp by -1
    \divide\graphtemp by 2
    \advance\graphtemp by .5ex
    \advance\graphtemp by 1.222in
    \rlap{\kern 3.222in\lower\graphtemp\hbox to 0pt{\hss $\tau$\hss}}%
\pdfliteral{
q [] 0 d 1 J 1 j
0.576 w
284.04 -59.976 m
284.04 -62.202812 282.234812 -64.008 280.008 -64.008 c
277.781188 -64.008 275.976 -62.202812 275.976 -59.976 c
275.976 -57.749188 277.781188 -55.944 280.008 -55.944 c
282.234812 -55.944 284.04 -57.749188 284.04 -59.976 c
S
0.072 w
q 0 g
270.792 -66.672 m
277.2 -62.856 l
273.384 -69.192 l
270.792 -66.672 l
B Q
0.576 w
254.808 -85.176 m
272.088 -67.896 l
S
Q
}%
    \graphtemp=.5ex
    \advance\graphtemp by 1.028in
    \rlap{\kern 3.694in\lower\graphtemp\hbox to 0pt{\hss $\tau~~~~~$\hss}}%
\pdfliteral{
q [] 0 d 1 J 1 j
0.576 w
312.048 -31.968 m
312.048 -34.194812 310.242812 -36 308.016 -36 c
305.789188 -36 303.984 -34.194812 303.984 -31.968 c
303.984 -29.741188 305.789188 -27.936 308.016 -27.936 c
310.242812 -27.936 312.048 -29.741188 312.048 -31.968 c
S
0.072 w
q 0 g
298.8 -38.664 m
305.136 -34.848 l
301.32 -41.184 l
298.8 -38.664 l
B Q
0.576 w
282.816 -57.168 m
300.096 -39.888 l
S
Q
}%
    \graphtemp=.5ex
    \advance\graphtemp by 0.639in
    \rlap{\kern 4.083in\lower\graphtemp\hbox to 0pt{\hss $\tau~~~~~$\hss}}%
\pdfliteral{
q [] 0 d 1 J 1 j
0.576 w
352.008 -31.968 m
352.008 -34.194812 350.202812 -36 347.976 -36 c
345.749188 -36 343.944 -34.194812 343.944 -31.968 c
343.944 -29.741188 345.749188 -27.936 347.976 -27.936 c
350.202812 -27.936 352.008 -29.741188 352.008 -31.968 c
S
0.072 w
q 0 g
336.816 -30.168 m
344.016 -31.968 l
336.816 -33.768 l
336.816 -30.168 l
B Q
0.576 w
311.976 -31.968 m
336.816 -31.968 l
S
Q
}%
    \graphtemp=\baselineskip
    \multiply\graphtemp by -1
    \divide\graphtemp by 2
    \advance\graphtemp by .5ex
    \advance\graphtemp by 0.444in
    \rlap{\kern 4.556in\lower\graphtemp\hbox to 0pt{\hss $a$\hss}}%
\pdfliteral{
q [] 0 d 1 J 1 j
0.576 w
340.056 -4.032 m
340.056 -6.258812 338.250812 -8.064 336.024 -8.064 c
333.797188 -8.064 331.992 -6.258812 331.992 -4.032 c
331.992 -1.805188 333.797188 0 336.024 0 c
338.250812 0 340.056 -1.805188 340.056 -4.032 c
S
0.072 w
q 0 g
326.808 -10.656 m
333.144 -6.84 l
329.328 -13.176 l
326.808 -10.656 l
B Q
0.576 w
q [3.6 3.328917] 0 d
310.824 -29.16 m
328.104 -11.952 l
S Q
Q
}%
    \graphtemp=.5ex
    \advance\graphtemp by 0.250in
    \rlap{\kern 4.472in\lower\graphtemp\hbox to 0pt{\hss $b~~~~~$\hss}}%
\pdfliteral{
q [] 0 d 1 J 1 j
0.576 w
284.04 -115.992 m
284.04 -118.218812 282.234812 -120.024 280.008 -120.024 c
277.781188 -120.024 275.976 -118.218812 275.976 -115.992 c
275.976 -113.765188 277.781188 -111.96 280.008 -111.96 c
282.234812 -111.96 284.04 -113.765188 284.04 -115.992 c
S
0.072 w
q 0 g
273.384 -106.776 m
277.2 -113.184 l
270.792 -109.368 l
273.384 -106.776 l
B Q
0.576 w
254.808 -90.864 m
272.088 -108.072 l
S
Q
}%
    \graphtemp=.5ex
    \advance\graphtemp by 1.417in
    \rlap{\kern 3.694in\lower\graphtemp\hbox to 0pt{\hss $~~~~~\tau$\hss}}%
\pdfliteral{
q [] 0 d 1 J 1 j
0.576 w
324 -115.992 m
324 -118.218812 322.194812 -120.024 319.968 -120.024 c
317.741188 -120.024 315.936 -118.218812 315.936 -115.992 c
315.936 -113.765188 317.741188 -111.96 319.968 -111.96 c
322.194812 -111.96 324 -113.765188 324 -115.992 c
S
0.072 w
q 0 g
308.808 -114.192 m
316.008 -115.992 l
308.808 -117.792 l
308.808 -114.192 l
B Q
0.576 w
283.968 -115.992 m
308.808 -115.992 l
S
Q
}%
    \graphtemp=\baselineskip
    \multiply\graphtemp by -1
    \divide\graphtemp by 2
    \advance\graphtemp by .5ex
    \advance\graphtemp by 1.611in
    \rlap{\kern 4.167in\lower\graphtemp\hbox to 0pt{\hss $a$\hss}}%
\pdfliteral{
q [] 0 d 1 J 1 j
0.576 w
312.048 -144 m
312.048 -146.226812 310.242812 -148.032 308.016 -148.032 c
305.789188 -148.032 303.984 -146.226812 303.984 -144 c
303.984 -141.773188 305.789188 -139.968 308.016 -139.968 c
310.242812 -139.968 312.048 -141.773188 312.048 -144 c
S
0.072 w
q 0 g
301.32 -134.784 m
305.136 -141.192 l
298.8 -137.376 l
301.32 -134.784 l
B Q
0.576 w
282.816 -118.8 m
300.096 -136.08 l
S
Q
}%
    \graphtemp=.5ex
    \advance\graphtemp by 1.806in
    \rlap{\kern 4.083in\lower\graphtemp\hbox to 0pt{\hss $\tau~~~~~$\hss}}%
\pdfliteral{
q [] 0 d 1 J 1 j
0.576 w
352.008 -144 m
352.008 -146.226812 350.202812 -148.032 347.976 -148.032 c
345.749188 -148.032 343.944 -146.226812 343.944 -144 c
343.944 -141.773188 345.749188 -139.968 347.976 -139.968 c
350.202812 -139.968 352.008 -141.773188 352.008 -144 c
S
0.072 w
q 0 g
336.816 -142.2 m
344.016 -144 l
336.816 -145.8 l
336.816 -142.2 l
B Q
0.576 w
311.976 -144 m
336.816 -144 l
S
Q
}%
    \graphtemp=\baselineskip
    \multiply\graphtemp by -1
    \divide\graphtemp by 2
    \advance\graphtemp by .5ex
    \advance\graphtemp by 2.000in
    \rlap{\kern 4.556in\lower\graphtemp\hbox to 0pt{\hss $a$\hss}}%
\pdfliteral{
q [] 0 d 1 J 1 j
0.576 w
340.056 -172.008 m
340.056 -174.234812 338.250812 -176.04 336.024 -176.04 c
333.797188 -176.04 331.992 -174.234812 331.992 -172.008 c
331.992 -169.781188 333.797188 -167.976 336.024 -167.976 c
338.250812 -167.976 340.056 -169.781188 340.056 -172.008 c
S
0.072 w
q 0 g
329.328 -162.792 m
333.144 -169.2 l
326.808 -165.384 l
329.328 -162.792 l
B Q
0.576 w
q [0 3.491087] 0 d
310.824 -146.808 m
328.104 -164.088 l
S Q
Q
}%
    \graphtemp=\baselineskip
    \multiply\graphtemp by 1
    \divide\graphtemp by 2
    \advance\graphtemp by .5ex
    \advance\graphtemp by 2.194in
    \rlap{\kern 4.472in\lower\graphtemp\hbox to 0pt{\hss $b~~~$\hss}}%
\pdfliteral{
q [] 0 d 1 J 1 j
0.576 w
312.048 -87.984 m
312.048 -90.210812 310.242812 -92.016 308.016 -92.016 c
305.789188 -92.016 303.984 -90.210812 303.984 -87.984 c
303.984 -85.757188 305.789188 -83.952 308.016 -83.952 c
310.242812 -83.952 312.048 -85.757188 312.048 -87.984 c
S
0.072 w
q 0 g
298.8 -94.68 m
305.136 -90.864 l
301.32 -97.2 l
298.8 -94.68 l
B Q
0.576 w
282.816 -113.184 m
300.096 -95.904 l
S
Q
}%
    \graphtemp=.5ex
    \advance\graphtemp by 1.417in
    \rlap{\kern 4.083in\lower\graphtemp\hbox to 0pt{\hss $\tau~~~~~$\hss}}%
\pdfliteral{
q [] 0 d 1 J 1 j
0.576 w
0.072 w
q 0 g
301.32 -78.84 m
305.136 -85.176 l
298.8 -81.36 l
301.32 -78.84 l
B Q
0.576 w
282.816 -62.856 m
300.096 -80.064 l
S
Q
}%
    \graphtemp=.5ex
    \advance\graphtemp by 1.028in
    \rlap{\kern 4.083in\lower\graphtemp\hbox to 0pt{\hss $~~~~~\tau$\hss}}%
\pdfliteral{
q [] 0 d 1 J 1 j
0.576 w
0.072 w
q 0 g
336.816 -86.184 m
344.016 -87.984 l
336.816 -89.784 l
336.816 -86.184 l
B Q
0.576 w
311.976 -87.984 m
336.816 -87.984 l
S
Q
}%
    \graphtemp=\baselineskip
    \multiply\graphtemp by -1
    \divide\graphtemp by 2
    \advance\graphtemp by .5ex
    \advance\graphtemp by 1.222in
    \rlap{\kern 4.556in\lower\graphtemp\hbox to 0pt{\hss $\tau$\hss}}%
    \graphtemp=.5ex
    \advance\graphtemp by 1.222in
    \rlap{\kern 5.056in\lower\graphtemp\hbox to 0pt{\hss $\dots$\hss}}%
    \hbox{\vrule depth2.444in width0pt height 0pt}%
    \kern 5.056in
  }%
}%